\documentclass[12pt]{article}

\usepackage[ansinew]{inputenc} 
\usepackage[UKenglish]{babel}
\usepackage{textcomp}
\usepackage{mathrsfs}
\usepackage{amsfonts}
\usepackage{mathtools}
\usepackage{amssymb}
\usepackage{amsmath}
\usepackage{amsthm}
\usepackage{amscd}
\usepackage{stmaryrd}
\usepackage{rotating}
\usepackage[text={6.8in,8.5in},centering]{geometry}
\usepackage{tikz-cd}
\usepackage{relsize}
\usepackage[none]{hyphenat}
\usepackage{hyperref}

\usepackage{graphicx}
\graphicspath{ {./Images/} }

\usepackage{hyperref,xcolor}
\usepackage{makeidx}
\definecolor{deepred}{rgb}{0.5,0,0}
\definecolor{deepblue}{rgb}{0,0,0.5}
\definecolor{deepgreen}{rgb}{0,0.5,0}

\hypersetup{
    colorlinks=true,
    linkcolor=deepred,
    filecolor=deepred,      
    urlcolor=deepred,
    citecolor=deepblue,
}

\newcommand{\Int}{\mathbb{Z}}
\newcommand{\Rat}{\mathbb{Q}}
\newcommand{\Real}{\mathbb{R}}
\newcommand{\dimm}{\text{dim}}
\newcommand{\Cin}[1]{\text{C}^\infty(#1)}
\newcommand{\Tan}{\text{T}}
\newcommand{\Cot }{\text{T}^*}
\newcommand{\Sec}[1]{\Gamma(#1)}
\newcommand{\Der}{\text{D}}
\newcommand{\Jet}{\text{J}}
\newcommand{\Proj}{\text{pr}}
\newcommand{\Id}{\text{id}}
\newcommand{\Hom}{\text{Hom}}
\newcommand{\Dr}[1]{\text{Der}(#1)}
\newcommand{\Ker}[1]{\text{ker}(#1)}
\newcommand{\Grph}[1]{\text{grph}(#1)}
\newcommand{\LGrph}[1]{\text{Lgrph}(#1)}

\newcommand{\Diff}{\text{Diff}}
\newcommand{\Obs}[1]{\text{Obs}(#1)}
\newcommand{\Dyn}[1]{\text{Dyn}(#1)}
\newcommand{\Lie}{\textsf{Lie}}
\newcommand{\Conf}{\mathcal{Q}}
\newcommand{\Man}{\mathcal{M}}

\newcommand{\Ring}{\textsf{Ring}}
\newcommand{\Vect}{\textsf{Vect}}
\newcommand{\LVect}{\textsf{LVect}}
\newcommand{\Line}{\textsf{Line}}
\newcommand{\Symp}{\textsf{Symp}}
\newcommand{\Poiss}{\textsf{Poiss}}
\newcommand{\Acts}{\mathbin{\rotatebox[origin=c]{-90}{$\circlearrowright$}}}
\newcommand{\dtimes}{\mathbin{\rotatebox[origin=c]{90}{$\ltimes$}}}
\newcommand{\utimes}{\mathbin{\rotatebox[origin=c]{-90}{$\ltimes$}}}

\newtheorem{prop}{Proposition}
\numberwithin{prop}{subsection}

\usepackage[
backend=biber,
style=alphabetic,
sorting=ynt
]{biblatex}
\begin{document}

\title{\vspace{-3.0cm}Measurand Spaces \\ and \\ Dimensioned Hamiltonian Mechanics}

\author{Carlos Zapata-Carratala}
\date{}

\maketitle

\sloppy

\begin{abstract}
    In this paper we introduce a generalization of Hamiltonian mechanics that replaces configuration spaces, conventionally regarded simply as smooth manifolds, with line bundles over smooth manifolds. Classical observables are then identified with the sections of these (generically non-trivial) line bundles. This generalization, mathematically articulated with theory of Jacobi manifolds, is motivated by a conceptual revision of the mathematical foundations of the notion of measurand and unit of measurement in practical science. We prove several technical results for the contact structures present on jet bundles in order to argue that our proposal does indeed successfully generalize Hamiltonian mechanics while incorporating a systematic treatment of physical dimension and units.
\end{abstract}

\tableofcontents

\newpage

\section{Introduction - The lack of a systematic treatment of physical dimension in geometric mechanics.} \label{Intro}

It is generally agreed among physicists that the foundations of classical mechanics, including the Hamiltonian and Lagrangian formalisms, are firmly rooted in the mathematical discipline known today as differential geometry. The development history of these two subjects is, in fact, inseparable, being, at a sufficiently early date, altogether indistinguishable. As a landmark publication that asserts this close connection we direct the reader to \cite{abraham1978foundations} and its many historical notes. Modern fields of mathematical enquiry such as Poisson geometry or the theory of Lie groups, have their origins in the early days of analytical mechanics. It is, then, perhaps not too surprising to find that most of the modern approaches to foundational questions in mechanics are formulated within the mathematical dialect of those disciplines. The present work will follow this tradition, as we will be drawing heavily from modern differential geometry, but we will also try to enforce the language of category theory in our exposition, with its common terminology and notation found for instance in \cite{leinster2014basic} or \cite{mac2013categories}, as the author deems this discipline the most apt to form broad and conceptually clear vistas of different areas of the mathematical sciences.\newline

The fundamental notion upon which geometric mechanics is built is that of \textbf{phase space}: the set of states of a physical system, generally taken to be a smooth manifold $P$, whose characteristics to be measured, the \textbf{observables}, are identified with real-valued functions, $\Obs{P}:=\Cin{P,\Real}$. More concretely, in \textbf{Hamiltonian mechanics}, the fundamental theoretical framework at the heart of the most widely-applied dynamical models of both classical and quantum physics, observables exhibit a Lie algebra structure compatible with the local structure of the manifold $P$ in such a way that $(\Obs{P},\cdot, \{,\})$ is a Poisson algebra, where $\cdot$ is the point-wise product of functions. We will only be considering the case of classical systems in our discussion, although many of the points we raise apply to quantum theories equally. A modern account of the unifying and diverging patterns between classical and quantum observables can be found in \cite{zalamea2016chasing} and references therein.\newline

Looking now to a relatively independent timeline of the history of science, we consider the development of the notion of physical dimension. Units of measurement have been in use essentially since the inception of mathematics thousands of years ago but the modern concept of \textbf{physical quantity}, as a mathematical entity carrying the numerical magnitude of a measurement and a unit or patron, was not formulated until the 19\textsuperscript{th} century. A good overview of the history of this subject, generally known as \textbf{dimensional analysis}, can be found in \cite{macagno1971historico}. As any high school student taking a first class in physics or chemistry will tell you, two different physical quantities must only be combined multiplicatively, not additively, to produce other physically meaningful quantities. For instance, given numerical values for mass $M$ and volume $V$, the numerical value $MV^{-1}$ may be identified as a density whereas the value $M + V$ has no physical meaning. Indeed, units of measurement in standard use today are of the form m/s\textsuperscript{2}, kg/m\textsuperscript{3}, etc. but not of the form kg+s, m+m\textsuperscript{2}, etc. A set of, more often than not, implicit directions and rules indicating the correct use of physical quantities seems to be common knowledge in the scientific community; however, the topic is rarely discussed beyond elementary undergraduate courses. This piece of, somewhat, scientific folklore is systematically accounted for by the \emph{Boureau International des Poids et Mesures} BIPM, which keeps the standards for the international system of units SI and defines the experimental common practice and terminology in their International Vocabulary of Metrology \cite{bimp2012metrology}.\newline

It is common for theoreticians to develop mathematical models for science taking physical quantities to be real numbers (or approximations thereof), leaving any dimensional considerations to the ulterior applied use of their theories. This is particularly evident in the case of the aforementioned phase space formalism, where all the observables of a physical system are collectively considered as part of the same set of functions. Working scientists and engineers are all too familiar with how painstakingly systematic one must be with the correct use of units so that theoretical models are of any practical use at all, and thus simply considering observables as real-valued functions seems to be an oversimplification of the formal structure of physical quantities as employed in real science.\newline

This lack of a systematic treatment of units of measurement has been noted by many authors. Coming from a background of applications in engineering, in \cite{hart2012multidimensional} G.W. Hart introduces the linear algebra and calculus of \textbf{dimensioned numbers}, the algebraic objects capturing the multiplicativity-but-non-additivity of physical quantities, in order to provide a formal framework for the standard techniques of dimensional analysis, such as Buckingham's $\pi$-theorem \cite{buckingham1915principle} or the linear algebra of vectors with components of different physical dimension \cite{taylor2008matrix}. On the mathematical side we find  the \textbf{dimensional categories} of J. Baez and J. Dolan \cite{dolan2009doctrines}, which elegantly capture the algebraic properties of physical quantities in categorical form, and the more explicit articulation of systems of units of measurement via \textbf{semi-vector spaces} of R. Vitolo \emph{et al.} \cite{janyvska2007semi}. These references provide rigorous accounts of the formal structure of units of measurement but they don't give much physical insight on why such mathematical structure should be used in the first place.\newline

The author then identifies \underline{the problematic that motivates the present work}: it is desirable to find a formulation of Hamiltonian mechanics that leads to a definition of observable which canonically exhibits the structure of dimensioned numbers, so no \emph{ad hoc} considerations are needed for the treatment of physical dimension. Such a generalization should, however, retain all the distinctive features that makes conventional Hamiltonian mechanics such a powerful foundation for modern physical theories.\newline

The author proposes \underline{a resolution of this problematic} in a twofold way: firstly, via a conceptual revision of the basic notions of metrology from which a mathematical formalization of \textbf{measurand} will ensue; and secondly, via a slight generalization of the Hamiltonian phase space formalism, where observables are sections of some line bundle, that will link measurands and classical mechanics in a natural way. Recent advances in the field of \textbf{Jacobi geometry}, in particular the works of L. Vitagliano \emph{et al.} championing the line bundle point of view for Jacobi manifolds \cite{vitagliano2014deformations}, hint that a line bundle version of the Hamiltonian formalism would capture the essential physical notions of interest in mechanics: time-evolution, conservation laws, constraints, symmetry, reduction, relations, etc. In this paper we develop all the necessary formalism to make this claim precise. We will identify the key features of conventional Hamiltonian mechanics as a functor between the category of configuration spaces and the category of phase spaces and we will show that our generalization exhibits a very similar categorical structure that, furthermore, accounts for the introduction of dimensioned observables. We are thus compelled to call this formulation \textbf{dimensioned Hamiltonian mechanics}.\newline

Jacobi geometry has seen a rise in popularity in recent years, both as a mathematical subject -  for the most part of the 20\textsuperscript{th} century, contact (Jacobi) geometry has been the neglected odd-dimensional sister of symplectic (Poisson) geometry - and for its applications to mathematical physics and other sciences. Just a few examples of recent developments are: the integrability of Jacobi manifolds to contact groupoids \cite{crainic2015jacobi}, the generalization of Dirac geometry for Jacobi manifolds \cite{vitagliano2015dirac} \cite{schnitzer2019normal}, a dissipative version of Liouville's theorem in contact manifolds \cite{bravetti2015liouville}, the identification of a contact structure in thermodynamics \cite{mrugala1991contact} \cite{grmela2014contact} or even applications to neuroscience \cite{petitot2017elements}. Most prominently for geometric mechanics, the work of M. de Le\'on \emph{et al.} \cite{deLeon2011methods} \cite{deLeon2017cosymplectic} \cite{valcazar2018contact} has shown how contact geometry is the natural framework for the dynamical formulation of mechanical systems subject to time-dependent forces and dissipative effects. The author believes that the introduction of the dimensioned Hamiltonian formalism presents yet another avenue for the application of Jacobi geometry to physics, this time tackling the issue of a systematic mathematical treatment of physical dimension and units in mechanics.\newline 

Contents are organized as follows: Section \ref{MathPrel} contains a broad summary of definitions and results of the mathematical structures that will be used throughout the paper, for these we assume the reader is familiar with differential geometry and the basic language of category theory. In \ref{DimNum} the algebraic structure of physical quantities is formalized as typed fields and \ref{CatLine} contains a review of the category of lines.  In \ref{CatLineB} and \ref{JacGeo} we present an original approach to line bundles and Jacobi structures as the ``unit-free'' analogues of manifolds and Poisson structures. Section \ref{ModMes} contains a detailed discussion of the metrological foundations of the notion of measurand and its connection to the category of lines and typed fields. In section \ref{HamMec} we review the fundamentals of Hamiltonian mechanics and identify the canonical formalism with a functor (forming the cotangent bundle) from the category of configuration spaces (smooth manifolds) to the category of phase spaces (symplectic manifolds). Section \ref{DimHamM} contains our proposed generalization of dimensioned configuration spaces as line bundles and dimensioned phase spaces as contact manifolds. In \ref{CanHamJ} we prove several technical results to support the claim that taking the jet bundle of a line bundle gives the functor for canonical dimensioned Hamiltonian mechanics. Lastly, a few closing remarks about open questions and future lines of research can be found in section \ref{Conc}.

\newpage

\section{Mathematical Preliminaries} \label{MathPrel}

\subsection{Typed Fields and Dimensioned Numbers} \label{DimNum}

The definitions in this section expand on the notion of typed family of fields introduced in section 1.2.1 of \cite{hart2012multidimensional}, that reference also contains a detailed discussion on the many nuances about the adequacy of such a mathematical definition for practical science and engineering. Let us begin with an explicit example that may help illustrate how the abstract definitions below indeed capture the intended properties of physical quantities. Consider the case of the thermodynamic theory of gases. Empirically, one determines that there are four parameters that will uniquely characterise the state of a gas, its pressure $P$, volume $V$, amount of matter $N$ and temperature $T$, and that they can all be measured independently by comparison to some patron unit, for example an atmosphere atm, a litre L, a mol and a Kelvin K. The gas under study will then have a state specified by $(P=p \text{atm},V=v\text{L},N=n\text{mol},T=t\text{K})$. Once assuming the ideal gas law, for instance, one can derive other physical quantities such as energy $U=PVN^{-1}$ or entropy $E=PVN^{-1}T^{-1}$ carrying units atm$\cdot$L$/$mol and atm$\cdot$L$/$mol$/$K, respectively. In general, any physical quantity will be of the form $Q=qP^aV^bN^cT^d$ for $q\in\mathbb{R}$ and $a,b,c,d\in\mathbb{Z}$. The product of two physical quantities $Q$ and $Q'$ will be
\begin{equation*}
    Q\cdot Q'= qq'P^{a+a'}V^{b+b'}N^{c+c'}T^{d+d'}.
\end{equation*}

If one tries to isolate the minimal axioms that are necessary to capture the structure of physical quantities as introduced in the example above, one arrives at the following definition: we call a \textbf{typed field} a set $M$ together with two surjective maps $m:M\to \mathbb{F}$ and $t:M\to G$, where $\mathbb{F}$ is a field and $G$ is a group, and two binary operations $\cdot:M\times M\to M$ and $+_g:t^{-1}(g)\to t^{-1}(g)$ for all $g\in G$ such that
\begin{enumerate}
    \item[i)] $t:(M,\cdot)\to G$ is a homomorphism of binary algebras,
    \item[ii)] $m:(t^{-1}(g), +_g)\to (\mathbb{F}$,+) is an isomorphism of abelian groups,
    \item[iii)] and the following compatibility condition holds for all $a,a'\in M$:
    \begin{equation*}
        m(a\cdot_M a')=m(a)\cdot_\mathbb{F}m(a').
    \end{equation*}
\end{enumerate}
Elements of $M$ are called \textbf{typed numbers}, the map $m$ is called the \textbf{magnitude map} and we say that a typed number $a\in M$ has magnitude $m(a)\in\mathbb{F}$, the map $t$ is called the \textbf{type map} and we say that a typed number $a\in M$ is of type $t(a)\in G$. We may refer to $M$ as a $\mathbb{F}$-field with $G$ types. Given any field $\mathbb{F}$ and group $G$ we can form the \textbf{standard typed field} simply by taking the Cartesian product $M:= \mathbb{F}\times G$, setting $m:=$pr$_1$, $t:=$pr$_2$ and defining the following operations:
\begin{enumerate}
    \item[i)] $(f,g)\cdot (f',g') := (ff',gg')$ using the field and group multiplication operations,
    \item[ii)] $(f,g) +_g (f',g) := (f+f',g)$ using the field addition.
\end{enumerate}

If we define \textbf{typed field morphism} as a map between typed fields $\psi:M\to N$ that induces group and field homomorphisms between the relevant sets in the obvious way, it follows from the definition above that any $\mathbb{F}$-field with $G$ types is, in fact, isomorphic to the standard $\mathbb{F}$-field with $G$ types $\mathbb{F}\times G$. This makes precise the intuition that a typed field corresponds to a collection of isomorphic fields indexed by elements in $G$. Note that there is a unique multiplicative identity, given by $1:=(1_{\mathbb{F}},1_G)$, that acts as the common multiplicative identity for the fields of all types, however there are as many additive identities as there are types, $0_g:=(0_{\mathbb{F}},g)$, $g\in G$.\newline

A $\mathbb{F}$-field with $G$ types is called \textbf{abelian} when the multiplication operation is commutative, equivalently, when $G$ is abelian, since we assume field multiplication to be commutative. The abelian typed fields whose abelian group is a module over a ring $R$ are of particular importance in this paper. In this case, when $M$ is a $\mathbb{F}$-field with $G$ types and $G$ is a $R$-module, we say that $M$ is a $R$-dimensioned $\mathbb{F}$-field with $G$ types. When the relevant module and field structures are understood, $M$ is simply called a \textbf{$R$-dimensioned field}. In particular, when $R=\mathbb{Z}$ and $G=\mathbb{Z}^k$ with the natural module structure, the typed field $M$ is called a \textbf{$k$-dimensioned field}. More generally, the $R$-module $G$ of a dimensioned $\mathbb{F}$-field $M$ may be finitely generated, in which case one can find a basis $\{g_1,\dots,g_k\}\subset G$, making $M$ typed field isomorphic to a $k$-dimensioned field with dimensions in $R$; such a choice of basis is reflected in the typed field naturally as $\{(1_\mathbb{F},g_1),\dots,(1_\mathbb{F},g_k)\}\subset M$ and we call these a choice of \textbf{base types}.\newline

Let us now return to the example of the thermodynamic gas, where a physical quantity was given by $Q=qP^aV^bN^cT^d$ for $q\in\mathbb{R}$ and $a,b,c,d\in\mathbb{Z}$. It is clear that the physical quantities characterizing a thermodynamic gas form a $\mathbb{R}$-field typed by $\mathbb{Z}^4$, what we also called a $4$-dimensioned field of real numbers. This is of course a typed field with a finitely-generated module over $\mathbb{Z}$, so we may choose a different set of base types. Physically, this simply corresponds to a different choice of four parameters to fully characterise the state of the gas.

\subsection{The Category of Lines} \label{CatLine}

In this section we present the essential properties of a class of real vector spaces that will be of fundamental relevance in our discussion about measurands in Section \ref{ModMes}. To this end, we will identify a subcategory of $\Vect$ and study its structure in some detail. Our terminology and notation for monoidal categories and canonical constructions will follow that of the standard literature, see e.g. \cite{leinster2014basic} or \cite{mac2013categories}.\newline

Let us define \textbf{the category of lines}, $\Line$. Objects are vector spaces over the field of real numbers $\mathbb{R}$ of dimension $1$, a useful way to think of these in the context of the present work is as sets of numbers without the choice of a unit. An object $L\in\Line$ will be appropriately called a \textbf{line}. A morphism in this category $b\in\Hom_{\Line}(L,L')$, usually simply denoted by $b:L\to L'$, is an invertible (equivalently non-zero) linear map, with composition simply the composition of maps. If we think of $L$ and $L'$ as numbers without a choice of a unit, a morphism $b$ between them can be thought of as a unit-free conversion factor, for this reason we will often refer to a morphism of lines as a \textbf{factor}. We consider the field of real numbers $\mathbb{R}$, trivially a line when regarded as real a vector space, as a singled out object in $\Line$. We mention at this point that the category $\Line$ is a prime example of what J. Baez and J. Dolan call dimensional categories in \cite{dolan2009doctrines}.\newline 

Note that all the morphisms in this category are, by definition, isomorphisms, thus making $\Line$ into a groupoid; however, one shouldn't think of all the objects in the category as being equivalent. As we shall see below, there are times when one finds factors between lines (invertible morphisms by definition) in a canonical way, that is, without making any further choices beyond the information that specifies the lines, these will be called \textbf{canonical factors}. In this manner, if given two lines $L,L'\in\Line$ there exists a canonical factor $b_{LL'}:L\to L'$ then we do regard $L$ and $L'$ as being equivalent. Our notation will reflect this fact by displaying a representative line in place of the equivalence class of lines for which there exist canonical factors mapping them to the representative and by using the equals sign ``='' between any of them. \newline

It is a simple linear algebra fact that any two lines $L,L'\in\Line$ satisfy
\begin{equation*}
    \dimm (L \oplus L')= \dimm L + \dimm L' = 2 > 1, \qquad \dimm L^* = \dimm L = 1, \qquad \dimm (L \otimes L') = 1.
\end{equation*}
Then, we note that the direct sum $\oplus$, which acts as a categorical product and coproduct in $\Vect$, is no longer defined in $\Line$, however, it is straightforward to check that $(\Line, \otimes, \mathbb{R})$ forms a symmetric monoidal category and that $*:\Line \to \Line$ is a duality contravariant autofunctor. The usual linear-algebraic definition of the tensor product via the universal property in $\Line$ implies, in particular, the facts listed below, which always hold in any symmetric monoidal category but that we mention here for convenience: 
\begin{itemize}
    \item For any two lines $L_1,L_2\in\Line$ there exists a canonical factor swapping their tensor product
    \begin{equation*}
        b_{12}:L_1\otimes L_2 \to L_2\otimes L_1.
    \end{equation*}
    \item The tensor unit $1_\otimes$ is precisely $\mathbb{R}$, this means that for any line $L\in\Line$ there is a canonical factor
    \begin{equation*}
        t_L:\Real\otimes L \to L.
    \end{equation*}
    \item For any three lines $L_1,L_2,L_3\in\Line$ there exists a canonical factor associating them
    \begin{equation*}
        a_{123}:L_1\otimes (L_2 \otimes L_3) \to (L_1\otimes L_2)\otimes L_3,
    \end{equation*}
    thus we can omit the brackets.
    \item More generally, given any collection of lines $L_1,\dots,L_k\in\Line$ and a permutation of $k$ elements $\sigma\in S_k$, there exists a canonical factor
    \begin{equation*}
        b_\sigma:L_1\otimes \cdots \otimes L_k \to L_{\sigma(1)}\otimes \cdots \otimes L_{\sigma(k)}.
    \end{equation*}
    \item Any two factors $b:L_1\to L_2$ and $b':L_1'\to L_2'$ have tensor product
    \begin{equation*}
        b\otimes b' :L_1\otimes L_1'\to L_2\otimes L_2'.
    \end{equation*}
    which is again a factor.
\end{itemize}
The autofunctor $*$ is in fact given by $\Real$-adjunction in the category of vector spaces $L^*:=\Hom_{\Vect_\Real}(L,\Real)$. It is then easy to show that for any lines $L,L'\in\Line$ there are canonical factors $d_L:(L^*)^*\to L$ and $e_{LL'}:(L\otimes L)^*\to L^*\otimes L'^*$. Similarly, one can show that $L^*\otimes L'$ is canonically isomorphic to $\Hom_{\Vect}(L,L')$ in the category of vector spaces and is indeed of dimension 1. In particular, $\Hom_{\Vect}(L,L)\in\Line$ has a distinguished non-zero element, the identity id$_L$, thus we find that it is canonically isomorphic to $\Real$ as lines. Therefore, for any line $L\in\Line$ we find a canonical factor
\begin{equation*}
    p_L:L^*\otimes L\to \Real.
\end{equation*}
This last result, under the intuition of lines as numbers without a choice of unit, allows us to reinterpret the singled out line $\Real$ informally as the set of procedures common to all lines by which a number gives any other number in a linear way (preserving ratios). This interpretation somewhat justifies the following adjustment in terminology: we will refer to the tensor unit $\Real\in\Line$ as \textbf{the patron line} or simply \textbf{the patron}. This will be useful for our discussion in Section \ref{ModMes}, where the word ``unit'' will be reserved for a construction within $\Line$ that directly reflects the physical notion of unit of measurement. To make this intuition precise let $L\in\Line$ and define the map
\begin{align*}
\lambda: L\times L^{\times} & \to L^*\otimes L\\
(a,b) & \mapsto \lambda_{ab} \text{ such that } a=\lambda_{ab}(b),
\end{align*}
where $L^{\times}:= L\diagdown \{0\}$. And reciprocally also define $\rho:L^{\times}\times L\to L^*\otimes L$. The following proposition gives the desired intuition of $L$ being a a ``unit-free'' field of numbers and $\Real$ being the patron of ratios between ``unit-free'' numbers.
\begin{prop}[Ratio Maps]\label{RatMaps}
The maps $\lambda$ and $\rho$ are well-defined and for any $a,b,c\in L^{\times}$ we have the following identities for the maps $l:=p_L\circ \lambda$ and $r:=p_L\circ \rho$
\begin{equation*}\label{ratio1}
    l_{ab}\cdot l_{bc}\cdot l_{ca}=1 \qquad l_{ab}\cdot r_{ab}=1 \qquad r_{ab}\cdot r_{bc}\cdot r_{ca}=1.\tag{ratio1}
\end{equation*}
The maps $l$ and $r$ are called the \textbf{ratio maps}.
\end{prop}
\begin{proof}
Since both $L$ and $\Hom_\Vect(L,L)$ are 1-dimensional real vector spaces, it follows that there is a unique linear isomorphism mapping any two non-zero elements in $L$ and the zero element is only mapped by the zero linear map. This can be seen more explicitly with the use of the canonical factor $p_L:\Hom(L,L)\to \Real$, which allows us to define the map $l:=p_L\circ \lambda$ that when applied to two non-zero elements $a,b\in L$ gives the unique non-zero real number $l_{ab}$ acting as proportionality factor: $a=l_{ab}\cdot b$. To prove the 2-out-of-3 identity consider $a=\lambda_{ab}(b)$, $b=\lambda_{bc}(c)$ and $c=\lambda_{ca}(a)$. Combining the three equations we find $\lambda_{bc}\circ\lambda_{ca}(a)=\lambda_{ba}(a)$. Noting that $\lambda_{ab}=\lambda^{-1}_{ba}$ gives the desired result. Similarly for $\rho$ and the reciprocal identity follows by construction.
\end{proof}
Consider now two lines $L_1,L_2\in\Line$ with their corresponding maps $l^1$, $r^1$ and $l^2$, $r^2$. It follows from the definitions above that for any factor $B:L_1\to L_2$ we have $l^1_{ab}=l^2_{B(a)B(b)}$ and $r^1_{ab}=r^2_{B(a)B(b)}$. More generally, we can define functions on the space of factors between the lines $L_1$ and $L_2$ from pairs of line elements:
\begin{align*}\label{ratio2}
    a_1\in L_1, b_2\in L_2^{\times} &\quad \mapsto \quad l^{12}_{a_1b_2}(B):=l^2_{B(a_1)b_2}\\
    b_1\in L_1^{\times}, a_2\in L_2 &\quad \mapsto \quad r^{12}_{b_1a_2}(B):=r^1_{b_1B^{-1}(a_2)}\tag{ratio2}
\end{align*}
for all $B:L_1 \to L_2$ a factor. It follows then by construction, that for any pair of non-zero line elements $b_1\in L_1^{\times}$ and $b_2\in L_2^{\times}$ the following identity holds
\begin{equation*}
    l^{12}_{b_1b_2}\cdot r^{12}_{b_1b_2} =1
\end{equation*}
as functions over the space of factors $\Hom_\Line(L_1,L_2)$. Whenever there is no room for confusion, we will employ the following abuse of fraction notation for non-zero elements:
\begin{equation*}
    l_{ab}=\frac{a}{b}=\frac{1}{r_{ab}} \qquad \qquad l^{12}_{a_1b_2}=\frac{a_1}{b_2}=\frac{1}{r^{12}_{b_1a_2}}
\end{equation*}
which notationally reflect the identities (\ref{ratio1}).\newline

If we write the tensor product of $n$ copies of a line and $m$ copies of its dual as $\otimes^m L^* \otimes^n L$, composing the canonical factors $t_L$, $b_\sigma$ and $p_L$ iteratively we obtain the following canonical factor:
\begin{equation*}
    a_{nm}:\otimes^m L^*  \otimes^n L \to \begin{cases} 
      \otimes^{|n-m|}L & n>m \\
      \Real & n=m \\
      \otimes^{|n-m|}L^* & n<m 
   \end{cases}
\end{equation*}
This result then encourages the introduction of the following notation:
\begin{equation*}
    \begin{cases} 
      L^n:=\otimes^nL & n>0 \\
      L^n:= \Real & n=0 \\
      L^n:=\otimes^nL^* & n<0 
   \end{cases}
\end{equation*}
which is such that given two integers $n,m\in\Int$ and any line $L\in\Line$ the following equations hold
\begin{equation*}
    (L^n)^* = L^{-n} \qquad L^{n}\otimes L^{m} = L^{n+m}.
\end{equation*}
Thus we see how one single line and its dual $L,L^*\in\Line$ generate an abelian group with the tensor product as group multiplication, the patron $\Real\in\Line$ as group identity and the duality autofunctor as inversion. We call this group the \textbf{potentil} or \textbf{power group} of $L\in\Line$ and denote it by $L^\otimes$. It trivially follows from the results above that $(\Real^\otimes,\otimes) \cong 0$ and for $L\neq \Real$,  $(L^\otimes,\otimes) \cong (\Int,+)$ as abelian groups.\newline

If we now consider a finite collection of lines $\{L_i\in\Line\}_{i=1}^k$ we can define its \textbf{potentil} or \textbf{power module} by taking all the possible tensor powers:
\begin{equation*}\label{potentil}
    (L_1\dots L_k)^\otimes := \bigcup_{n_1,\dots,n_k\in\Int} L_1^{n_1}\otimes \cdots \otimes L_k^{n_k}.\tag{potentil}
\end{equation*}
An iterative use of the canonical factors $b_\sigma$, $e_{LL'}$ and $p_L$ allows us to write the following equations
\begin{align*}
    (L_1^{n_1}\otimes \cdots \otimes L_k^{n_k})^* &=L_1^{-n_1}\otimes \cdots \otimes L_k^{-n_k}\\
    (L_1^{n_1}\otimes \cdots \otimes L_k^{n_k}) \otimes (L_1^{m_1}\otimes \cdots \otimes L_k^{m_k})  &= L_1^{n_1+m_1}\otimes \cdots \otimes L_k^{n_k+m_k}
\end{align*}
and thus we see that the potentil of a collection of lines forms an abelian group with the tensor product as multiplication, duality as inversion and the patron as group identity. The components of a potentil $(L_1\dots L_k)^\otimes$ are called the \textbf{base lines} and we have
\begin{equation*}
    (L_1\dots L_k)^\otimes \cong L_1^\otimes \oplus \dots \oplus L_k^\otimes
\end{equation*}
as abelian groups. Furthermore, the potentil $(L_1\dots L_k)^\otimes$ has the structure of a $\Int$-module with the action given by:
\begin{equation*}
    m\cdot (L_1^{n_1}\otimes \cdots \otimes L_k^{n_k}):=\otimes^m(L_1^{n_1})\otimes \cdots \otimes \otimes^m(L_k^{n_k})=L_1^{m\cdot n_1}\otimes \cdots \otimes L_k^{m\cdot n_k}.
\end{equation*}

The identification of the category $\Line$ allows for a natural generalization of the category of real vector spaces where one replaces the field of scalars with a line, thought of informally as a ``unit-free'' field of scalars. The category of \textbf{line vector spaces} or \textbf{L-vector spaces} is defined formally as the product category
\begin{equation*}
    \LVect:= \Vect \times \Line.
\end{equation*}
Our notation for objects in this category will be $V^L:=(V,L)$ with $V\in\Vect$, $L\in\Line$, and similarly for morphisms $\psi^b:V^L\to W^{L'}$ with $\psi\in\Hom_\Vect(V,W)$, $b\in\Hom_\Line(L,L')$. Objects $V^L$ will be called L-vector spaces and morphisms $\psi^b$ will be called linear factors. The tensor structures of both $\Vect$ and $\Line$ induce a natural \textbf{tensor product of L-vector spaces} given by
\begin{align*}
    V^L\otimes W^{L'}:&=(V\otimes W)^{L\otimes L'}\\
    \psi^b\otimes \phi^{b'}:&=(\psi \otimes \phi)^{b\otimes b'}\\
    1_{\otimes\LVect}&=\Real^\Real,
\end{align*}
making $\textsf{Lvect}$ into a symmetric monoidal category. We define a the \textbf{L-dual} of an L-vector space as
\begin{equation*}
    *(V^L):=(V^*\otimes L)^L.
\end{equation*}
We will often use the more compact notation $(V^*\otimes L)^L:=V^{*L}$. This construction is crucially idempotent
\begin{equation*}
    *(*(V^{L}))=(V^{**}\otimes L\otimes L^*)^L \cong (V^{**}\otimes \Real)^L \cong (V^{**})^L \cong V^L,
\end{equation*}
which justifies the chosen name of L-duality, in analogy with conventional duality of vector spaces. Given a linear factor $\psi^b:V^L\to W^{L'}$ we can define its L-dual as
\begin{align*}
\psi^{L'*L}: W^{*L'} & \to V^{*L}\\
\beta^{l'} & \mapsto \alpha^l
\end{align*}
with
\begin{equation*}
    \alpha = b^{-1}\circ \beta \circ \psi \qquad l=b^{-1}(l'),
\end{equation*}
where we regard L-dual vectors as linear maps $\alpha\in V^*\otimes L\cong \Hom_\Vect(V,L)$. This construction clearly makes $*$ into a duality contravariant autofunctor. The direct sum is only defined for L-vector spaces sharing the same line:
\begin{equation*}
    V^L\oplus W^L := (V\oplus W)^L.
\end{equation*}

Subobjects in this category will be identified with linear factors whose vector component is injective, in particular, given any vector subspace $i:U\hookrightarrow V$ we have an inclusion linear factor given by $i^{\Id_L}:U^L\to V^L$. Quotients can then be taken in the obvious way $V^L/U:=(V/U)^L$. There is a natural notion of \textbf{L-annihilator} of a subspace
\begin{equation*}
    U^{0L}:=\{\alpha\in V^{*L} | \quad \alpha(u)=0\in L, \forall u\in U\}\cong U^0\otimes L.
\end{equation*}
It is then a direct result from linear algebra on $\Vect$ that there exists a canonical linear factor isomorphism
\begin{equation*}
    a^{\Id_L}:U^{*L}\to V^{*L}/U^{0L}.
\end{equation*}

\subsection{The Category of Line Bundles}\label{CatLineB}

In this section we gather a few standard results about vector bundles and differential operators in order to define the category of line bundles and study its structure in some detail. For references on smooth vector bundles we direct the reader to any standard text on differential geometry, e.g. \cite{lee2009manifolds}, however the author finds \cite{nestruev2006smooth}, with its emphasis on the algebraic description of sections, particularly illuminating for the matters that concern the present paper.\newline

Let us denote the category of smooth manifolds with smooth maps as $\Man$. Let $\epsilon: E\to M$ be a vector bundle over a smooth manifold $M\in\Man$, throughout this paper we will use the compact notation $E_M$, or simply $E$ at times, to denote the bundle as an object in the category of vector bundles $E_M\in \Vect_\Man$ with bundle maps covering general smooth maps as morphisms. The bundle projection of any other vector bundles that are canonically constructed from the fibre-wise linear structure, such as duals or tensor powers, will be given the same name, e.g. $\epsilon: E^*\to M$ or $\epsilon: E\otimes E \to M$, unless the context demands otherwise. Let $\phi:N\to M$ be a smooth map, the pull-back bundle $\phi^*E_M$ will be also denoted by $E_N$. The $\Cin{M}$-module formed by the sections of $E$ will be denoted by $\Sec{E}$. Unlike the contravariant functor $\text{C}^\infty:\Man\to \Ring$, the association of sections $\Gamma:\Vect_\Man\to \textsf{RMod}$ is not a functor for general vector bundles.\newline

Let two vector bundles $A_M$ and $B_M$ over the same base $M$; we can identify a class of maps between the spaces of sections demanding compatibility with the module structures. These are called \textbf{differential operators of order at most} $k$, for $k$ some integer, and defined as follows:
\begin{equation*}
    \Diff_k(A,B):=\{\Delta : \Sec{A} \to \Sec{B} |\quad c_{f_0}\circ c_{f_1}\circ \dots c_{f_k}\circ (\Delta)=0 \quad \forall f_0,f_1,\dots,f_k\in\Cin{M}\}
\end{equation*}
where we have simplified the notation of nested commutators by using the maps $c_f(\Delta):=[\Delta,f]=\Delta(f\cdot )-f\Delta(\cdot)$. Observe that from this definition it follows immediately that $\Diff_k(A,B)$ is a $\Cin{M}$-module. Also clear from the definition is that differential operators of order $0$ correspond to $\Cin{M}$-linear maps and thus are identified with vector bundle morphisms covering the identity:
\begin{equation*}
    \Diff_0(A,B)\cong \Hom(A,B)\cong \Sec{A^*\otimes B}.
\end{equation*}
We thus explicitly see that differential operators of order $0$ are the sections of some vector bundle. In general, we have $\Diff_{k-1}(A,B)\subseteq \Diff_k(A,B)$, furthermore, given any differential operator order $k$ we can associate a differential operator of order $0$ that captures its behaviour at the highest order. This is the \textbf{$k$-symbol map}:
\begin{align*}
\sigma: \Diff_k (A,B) & \to \Sec{\odot^k\Tan M\otimes A^*\otimes B}\\
\Delta & \mapsto \sigma_\Delta | \sigma_\Delta(df_1,\dots,df_k):=c_{f_1}\circ \cdots \circ c_{f_k}(\Delta).
\end{align*}
The image of a given differential operator under this map is usually called its symbol. We thus write the following short exact sequence:
\begin{equation*}
\begin{tikzcd}
0 \arrow[r] & \Diff_{k-1}(A,B) \arrow[r ] & \Diff_{k}(A,B) \arrow[r, "\sigma"] & \Sec{\odot^k\Tan M\otimes A^*\otimes B}\arrow[r]  & 0
\end{tikzcd}
\end{equation*}
We will denote the space of differential operators from $A$ to itself as $\Diff_k(A,A):=\Diff_k(A)$ and the bundle of fibre-wise endomorphisms as End$(A):=A^*\otimes A$. Composition of differential operators gives $\bigcup_k^ \infty \Diff_k(A) \subseteq\text{Hom}_{\mathbb{R}}(\Sec{A})$ the structure of a filtered algebra since one can readily check:
\begin{align*}\label{diffops}
    \Diff_k(A) \circ \Diff_l(A) &\subseteq \Diff_{k+l}(A)\\
    [\Diff_k(A), \Diff_l(A)] &\subseteq \Diff_{k+l-1}(A).
\end{align*}
The \textbf{$k$-jet bundle} of a vector bundle $A_M$ is defined fibre-wise as the classes of local sections agreeing on derivatives of order up to $k$, we denote it by $\pi:\Jet^kA\to M$. The \textbf{$k$-jet prolongation} $j^k: \Sec{A}\to \Sec{\Jet^kA}$, which simply takes a section to the class of jets corresponding to its derivatives of order at most $k$. This map is sometimes called the \textbf{universal differential operator} since it realizes the isomorphism of $\Cin{M}$-modules:
\begin{equation*}
    \Diff_k(A,B)\cong \Sec{(\Jet^kA)^*\otimes B}
\end{equation*}
via the correspondence
\begin{equation*}
    \Hom(\Jet^k A,B)\ni\Phi \mapsto\Phi \circ j^k\in\Diff_k(A,B).
\end{equation*}
Now taking the $A^*$-adjoint of the symbol sequence above we obtain the \textbf{Spencer sequence}
\begin{equation*}
\begin{tikzcd}
0 \arrow[r] & \odot^k\Cot M\otimes A \arrow[r, "i"] & \Jet^kA \arrow[r, "\pi^k"]  & \Jet^{k-1}A \arrow[r]  & 0.
\end{tikzcd}
\end{equation*}

For a differential operator of degree at most $1$, $\Delta\in\Diff_1(A)$, the symbol is uniquely characterized by a Leibniz-like property that determines the interaction of the operator with the $\Cin{M}$-module structure:
\begin{equation*}
    \Delta(f\cdot a) =f\cdot \Delta (a) + \sigma_\Delta (df)(a)
\end{equation*}
here the symbol is simply a bundle map covering the identity $\sigma_\Delta : \Cot M\to \Hom(A,B)$. Note that the symbol sequence at order $1$ becomes
\begin{equation*}\label{diff1}
\begin{tikzcd}
0 \arrow[r] & \Sec{\text{End}(A)} \arrow[r] & \Diff_{1}(A) \arrow[r, "\sigma"] & \Sec{\Tan M\otimes \text{End}(A)}\arrow[r]  & 0,
\end{tikzcd}\tag{diff1}
\end{equation*}
and by adjoining $A^*$ we get the Spencer sequence
\begin{equation*}\label{spncr1}
\begin{tikzcd}
0 \arrow[r] & \Cot M\otimes A \arrow[r] & \Jet^1A \arrow[r]  & A \arrow[r]  & 0.
\end{tikzcd}\tag{spncr1}
\end{equation*}

The space of sections of the endomorphism bundle is clearly isomorphic to the space of vector bundle morphisms covering the identity as a $\Cin{M}$-module and it contains a submodule isomorphic to $\Cin{M}$ given by the elements of $\Sec{\text{End}(A)}$ corresponding to fibre-wise multiplication by functions. This special submodule of endomorphisms allows us to identify the subspace of differential operators that have symbol with endomorphism part lying in it. These are the \textbf{derivations} of the vector bundle $A$:
\begin{equation*}
    \Dr{A}:=\{D\in \Diff_1(A)|\quad \exists !X_D : D(f\cdot s)=f\cdot D(s) + X_D[f]\cdot s \quad \forall f\in \Cin{M},s\in\Sec{A}\}.
\end{equation*}
The space of derivations $\Dr{A}$ so defined can be easily shown to be a $\Cin{M}$-submodule of $\Diff_1(A)$ that is closed under the commutator. In fact, the derivations of a general vector bundle give a prime example of an integrable Lie algebroid. Note that derivations $\Dr{A}$ can be shown to be isomorphic to the sections of the \textbf{der bundle} of $A$:
\begin{equation*}
    \Der A:=\{a_x:\Sec{A}\to A_x|\quad a_x(f\cdot s)=f(x)a_x(s)+\delta(a_x)(f)(x)s(x)\quad \forall f\in \Cin{M},s\in\Sec{A}\}
\end{equation*}
where $\delta:\Der A\to \Tan M$ is a bundle map covering the identity called the \textbf{anchor} or the \textbf{local symbol}, defined from the symbol map $\sigma$. The \textbf{general linear groupoid} GL$(A)$ is defined as the space of all isomorphisms between fibres and one can show that its Lie algebroid $\mathfrak{gl}(A)$ is canonically isomorphic as a vector bundle to $\Der A$ and that $\Sec{\Der A}\cong \Dr{A}$ gives an isomorphism as Lie algebroids identifying the bracket of sections of $\Der A$ with the commutator of derivations. The der bundle is hence also referred to as the \textbf{derivation algebroid}. When the rank of $A_M$ is $1$, i.e. the fibres $A_x$ are 1-dimensional vector spaces, an elementary linear algebra argument shows that the fibre-wise endomorphisms are simply functions, in other words, the endomorphism bundle of a line bundle is trivial:
\begin{equation*}
    \text{rk}(A)=1 \quad \Rightarrow \quad \text{End}(A)=A^*\otimes A \cong \Real_M
\end{equation*}
where $\Real_M$ denotes the trivial line bundle over $M$. This implies, in turn, that all differential operators of degree at most $1$ are derivations:
\begin{equation*}
     \text{rk}(A)=1 \quad \Rightarrow \quad\Dr{A}=\Diff_1(A) \quad \Leftrightarrow \quad \Jet^1A=(\Der A)^*\otimes A.  
\end{equation*}

We define the \textbf{category of line bundles} $\Line_\Man$ as the subcategory of $\Vect_\Man$ whose objects are rank $1$ vector bundles $\lambda: L \to M$ and whose morphisms are regular, i.e. fibre-wise invertible, bundle morphisms covering general smooth maps
\begin{equation*}
\begin{tikzcd}
L_1 \arrow[r, "B"] \arrow[d, "\lambda_1"'] & L_2 \arrow[d, "\lambda_2"] \\
M_1 \arrow[r, "b"'] & M_2.
\end{tikzcd}
\end{equation*}
In the interest of brevity, we may refer to line bundles $L\in \Line_\Man$ as \textbf{lines} and regular line bundle morphisms $B\in\Hom_{\Line_\Man}(L_1,L_2)$ as \textbf{factors}. This nomenclature aligns with that of section \ref{CatLine} as we can indeed regard the category of vector lines $\Line$ as the subcategory of $\Line_\Man$ of line bundles over a single point. A factor covering a diffeomorphism, i.e. a line bundle isomorphism, is called a \textbf{diffeomorphic factor}. Similarly a factor covering an embedding or submersion is called an \textbf{embedding factor} or \textbf{submersion factor}, respectively. By fixing the base manifold $M$ we can use the tensor product and duality functor of $\Line$ to give fibre-wise analogues for general line bundles over $M$, denoted by $\Line_M$. This makes $(\Line_M,\otimes)$ into a symmetric monoidal category with tensor unit $\Real_M$. The duality assignment $L\in\Line_\Man\mapsto L^*\in\Line_\Man$ is well defined for objects but, in general, it fails to give an assignment of factors covering smooth (possibly non-invertible) maps, thus not being a functor.\newline

Recall that in section \ref{CatLine} we argued that objects of $\Line$, i.e. $1$-dimensional vector spaces, were to be thought of as ``unit-free'' fields of real numbers. Similarly, our definition of the category $\Line_\Man$ is motivated by an attempt to find a notion of ``unit-free'' manifold. Let us clarify this point. One thinks of a manifold $M$ as a smooth space with a ring of functions that has a unit, the constant function $1_M\in\Cin{M}$, with the assignment $\text{C}^\infty:\Man\to \Ring$ being a contravariant functor sending smooth maps $\varphi:M_1\to M_2$ to pull-backs of functions $\varphi^*:\Cin{M_2}\to \Cin{M_1}$. By taking an object in the category of line bundles $L_M\in\Line_\Man$ as our replacement for smooth space and its sections $\Sec{L_M}\in\textsf{RMod}$, $1$-dimensional projective $\Cin{M}$-modules, as our replacement for functions, we find $L_M$ to be our candidate for the notion of \textbf{unit-free smooth manifold}. We readily note that there is no choice of (canonical, global) unit section in $\Sec{L_M}$. In the reminder of this section we shall prove several results that will support the claim  that $L_M$ is indeed a good generalization of a smooth manifold without a canonical choice of unit. In particular, we will show: that there is a natural notion of pull-back of a factor, making the section assignment $\Gamma:\Line_\Man\to \textsf{RMod}$ into a contravariant functor; that there is a natural notion of subobject induced by submanifolds of the base manifold; that there is a categorical product of line bundles generalizing the Cartesian product of manifolds; and that the der and jet bundles generalize the notion of tangent and cotangent bundles.\newline

Let a $B:L_{M_1}\to L_{M_2}$ be a factor covering the smooth map $b:M_1\to M_2$, we can define the \textbf{factor pull-back} on sections as:
\begin{align*}
B^*: \Sec{L_{M_2}} & \to \Sec{L_{M_1}}\\
s_2 & \mapsto B^*s_2
\end{align*}
where
\begin{equation*}
    B^*s_2(x_1):=B^{-1}_{x_1}s_2(b(x_1))
\end{equation*}
for all $x_1\in M_1$, note that the fibre-wise invertibility of $B$ has critically been used for this to be a well-defined map of sections. Linearity of $B$ implies the following interaction of factor pull-backs with the module structures:
\begin{equation*}
    B^*(f\cdot s)=b^*f\cdot B^*s \quad \forall f\in\Cin{M_2}, s\in \Sec{L_{M_2}}.
\end{equation*}
Then, if we consider $\textsf{RMod}$, the $1$-dimensional projective modules over rings with module homomorphisms covering ring homomorphisms (satisfying the equation above), the assignment of sections becomes a contravariant functor
\begin{equation*}
    \Gamma:\Line_\Man \to \textsf{RMod}.
\end{equation*}

Since factors restrict to linear fibre isomorphisms, there is a 1:1 correspondence between monomorphisms in the category $\Line_\Man$ and embeddings of smooth manifolds in the category $\Man$. Indeed, given a submanifold $i:S\hookrightarrow M$ and a line bundle $\lambda:L\to M$, restriction to $S$ gives a factor
\begin{equation*}
\begin{tikzcd}
i^*L \arrow[r, "\iota"] \arrow[d] & L \arrow[d] \\
S \arrow[r,hook, "i"'] & M
\end{tikzcd}
\end{equation*}
where $\iota$ is an embedding factor. The pull-back bundle will often be denoted simply by $L_S$. We may give an algebraic characterization of submanifolds of the base space of a line bundle as follows: for an embedded submanifold $i:S\hookrightarrow M$ we define its \textbf{vanishing submodule} as
\begin{equation*}
    \Gamma_S:=\Ker{\iota^*}=\{s\in\Sec{L_M}|\quad s(x)=0\in L_x \quad \forall x\in S\}.
\end{equation*}
This is the unit-free analogue to the characterization of submanifolds with their vanishing ideals. In fact, the two notions are closely related since (depending on the embedding $i$, perhaps only locally) we have
\begin{equation*}
    \Gamma_S=I_S\cdot \Sec{L_M},
\end{equation*}
where $I_S$ is the multiplicative ideal of functions vanishing on the submanifold $I_S:=\Ker{i^*}$. This ideal gives a natural isomorphism of rings $\Cin{S}\cong \Cin{M}/I_S$ which, in turn, gives the natural isomorphism of ring-modules
\begin{equation*}
    \Sec{L_S}\cong \Sec{L}/\Gamma_S.
\end{equation*}

Let $G$ be a Lie group with lie algebra $\mathfrak{g}$ and a line bundle $\lambda:L\to M$, we say that \textbf{$G$ acts on $L$} and denote $G\Acts L$ when there is a smooth map $\Phi:G\times L\to L$ such that $\Phi_g:L\to L$ is a diffeomorphic factor for all $g\in G$ and the usual axioms of a group action are satisfied
\begin{equation*}
    \Phi_g\circ \Phi_h = \Phi_{gh}\qquad \Phi_{e}=\Id_L \qquad \forall g,h\in G.
\end{equation*}
We call the map $\Phi$ a \textbf{line bundle $G$-action}. It follows from construction that any such action induces a standard group action of $G$ on the base space $\phi:G\times M\to M$, which give the base smooth maps of the factors given by the line bundle action and is thus called the base action. The infinitesimal counterpart of a line bundle $G$-action is a morphism of Lie algebras
\begin{equation*}
    \Psi:\mathfrak{g}\to \Dr{L}.
\end{equation*}
Note that $G$ acts on both the domain and codomain of the infinitesimal action by the adjoint action and by push-forward respectively; the fact that $\Psi$ is defined as the infinitesimal counterpart of the action $\Phi$ manifest as $G$-equivariance in the following sense
\begin{equation*}
    \Psi\circ \text{Ad}_g = (\Phi_g)_*\circ \Psi \qquad \forall g\in G.
\end{equation*}
Denoting by $\psi:\mathfrak{g}\to \Sec{\Tan M}$ the infinitesimal counterpart of $\phi$, the infinitesimal line bundle action satisfies 
\begin{equation*}
    \Psi(\xi)[f\cdot s]=f\cdot \Psi(\xi)[s]+\psi(\xi)[f]\cdot s \qquad\forall f\in\Cin{M},s\in\Sec{L},\xi\in\mathfrak{g}.
\end{equation*}
The \textbf{orbits} of a line bundle $G$-action $\Phi$ can be simply defined as the images of all group elements acting on a single fibre, and thus they are naturally regarded as the line bundle restricted to the orbits of the base action $\phi$. In analogy with the case of smooth actions, we denote the set of orbits by $L/G$. Any notion defined for usual group actions on smooth manifolds extends to a corresponding notion for line bundle actions simply requiring the base action to satisfy the corresponding conditions. In particular, a \textbf{free and proper} line bundle $G$-action gives a well-defined submersion factor
\begin{equation*}
\begin{tikzcd}
L \arrow[r, "\zeta"] \arrow[d] & L/G \arrow[d] \\
M \arrow[r,two heads, "z"'] & M/G
\end{tikzcd}
\end{equation*}
where the line bundle structure on the set of orbits $L/G$ is induced from the fact that the orbit space of the base action $M/G$ is a smooth manifold and from the fact that all pairs of fibres over the same orbit are mapped isomorphically by some group element. Under these conditions, analogously to the ring isomorphism between functions on the orbit space and invariant functions $\Cin{M/G}\cong \Cin{M}^G$, we find the following isomorphism of modules
\begin{equation*}
    \Sec{L/G}\cong \Sec{L}^G:=\{s\in\Sec{L}| \quad \Phi_g^*s=s \quad \forall g\in G\}.
\end{equation*}

Let us now define the analogue of the Cartesian product in the category $\Line_\Man$. Consider two line bundles $\lambda_i:L_i\to M_i$, $i=1,2$, we will use $L_i$ and $L_{M_i}$ indistinctly. We begin by defining the set of all linear invertible maps between fibres:
\begin{equation*}
    M_1 \dtimes M_2:=\{B_{x_1x_2}:L_{x_1}\to L_{x_2}, \text{ linear isomorphism }, (x_1,x_2)\in M_1\times M_2\},
\end{equation*}
we call this set the \textbf{base product} of the line bundles. Let us denote by $p_i:M_1\dtimes M_2 \to M_i$ the obvious projections. This set is easily shown to be a smooth manifold by taking trivializations $L|_{U_i}\cong U_i\times \Real$, $i=1,2$ that give charts of the form $U_1\times U_2\times \Real^\times$ for the open neighbourhoods $(p_1\times p_2)^{-1}(U_1 \times U_2)\subseteq M_1\dtimes M_2$. In fact, the natural $\Real^\times$-action given by fibre-wise multiplication makes $M_1\dtimes M_2$ into a principal bundle
\begin{equation*}
    \begin{tikzcd}[column sep=0.1em]
        \Real^\times & \Acts & M_1\dtimes M_2  \arrow[d,"p_1\times p_2"]\\
         & & M_1\times M_2.
    \end{tikzcd}
\end{equation*}
We then define the \textbf{line product} of the line bundles as $L_1\utimes L_2 :=p_1^*L_1$, which is a line bundle over the basic product $\lambda_{12}:L_1\utimes L_2\to M_1\dtimes M_2$. Note that our notation for the base product $M_1 \dtimes M_2$ does not explicitly keep track of the line bundles $\lambda_i:L_i\to M_i$, in the interest of simplicity of notation, throughout this paper, care has been taken to always provide sufficient context for an unambiguous interpretation of this notation.\newline

Despite the apparent asymmetry of the definition, note that we can define the following factors
\begin{align*}
    P_1(B_{x_1x_2},l_{x_1})&:= l_1\in L_{x_1}\\
    P_2(B_{x_1x_2},l_{x_1})&:= B_{x_1x_2}(l_{x_1})\in L_{x_2}
\end{align*}
where $(B_{x_1x_2},l_{x_1})\in p_1^*L_1$, thus giving the following commutative diagram
\begin{equation*}\label{lproduct}
\begin{tikzcd}
L_1 \arrow[d, "\lambda_1"'] & L_1\utimes L_2 \arrow[l,"P_1"']\arrow[d, "\lambda_{12}"]\arrow[r,"P_2"] & L_2\arrow[d,"\lambda_2"] \\
M_1 & M_1 \dtimes M_2\arrow[l,"p_1"]\arrow[r,"p_2"'] & M_2
\end{tikzcd}\tag{lproduct}
\end{equation*}
where $P_1$ and $P_2$ are submersion factors. It is clear from this definition that taking the pull-back bundle $p_2^*L_2$ instead of $p_1^*L_1$ as our definition for line product will give the line bundle that we have denoted by $L_2\utimes L_1$. By construction of the base product $M_1\dtimes M_2$ we can construct the following smooth map:
\begin{align*}
c_{12}: M_1\dtimes M_2 & \to M_2\dtimes M_1\\
B_{x_1x_2} & \mapsto B_{x_1x_2}^{-1}
\end{align*}
which is clearly invertible, $c_{12}^{-1}=c_{21}$, and that induces a factor on the line products $C_{12}(B_{x_1x_2},l_{x_1})=(B_{x_1x_2}^{-1},B_{x_1x_2}(l_{x_1}))$. We have thus found a canonical factor covering a diffeomorphism
\begin{equation*}
\begin{tikzcd}
L_1\utimes L_2 \arrow[r,"C_{12}"]\arrow[d, "\lambda_{12}"] & L_2\utimes L_1\arrow[d, "\lambda_{21}"] \\
M_1 \dtimes M_2\arrow[r,"c_{12}"'] & M_2 \dtimes M_1
\end{tikzcd}
\end{equation*}
hence proving that the definition of line product is indeed symmetrical, i.e the two products are canonically isomorphic as line bundles $L_1\utimes L_2\cong L_2\utimes L_1$. The following proposition establishes the line product as the unit-free version of the Cartesian product.
\begin{prop}[Line Product]\label{LineProd}
The line product construction $\utimes:\Line_\Man \times \Line_\Man \to \Line_\Man$ is a categorical product in the category of line bundles $\Line_\Man$.
\end{prop}
\begin{proof}
For $\utimes$ to be a categorical product in $\Line_\Man$ first we need the two canonical projection morphisms, these are given by the factors $P_1$ and $P_2$ defined above
\begin{equation*}
\begin{tikzcd}
L_1 & L_1\utimes L_2 \arrow[l,"P_1"']\arrow[r,"P_2"] & L_2.
\end{tikzcd}
\end{equation*}
Now we must check the universal property that, given two morphisms $B_i:L\to L_i$, $i=1,2$, we get a unique morphism $B:L\to L_1 \utimes L_2$ such that the following diagram commutes
\begin{equation*}
\begin{tikzcd}
 & L\arrow[dl,"B_1"'] \arrow[dr,"B_2"]\arrow[d,"B"] & \\
L_1 & L_1\utimes L_2 \arrow[l,"P_1"]\arrow[r,"P_2"'] & L_2.
\end{tikzcd}
\end{equation*}
We can indeed define the \textbf{line product of factors} as
\begin{align*}
B_1\utimes B_2: L & \to L_1\utimes L_2\\
l_x & \mapsto (B_2|_x\circ B_1|_{x}^{-1},B_1(l_x))
\end{align*}
noting that, by assumption, $B_2|_x\circ B_1|_{x}^{-1}:L_{b_1(x)}\to L_{b_2(x)}\in M_1\dtimes M_2$ since factors are fibre-wise invertible. We readily check that setting $B=B_1\utimes B_2$ makes the above diagram commutative, thus completing the proof.
\end{proof}
The construction of the line product as a pull-back bundle over the base product indicates that the sections $\Sec{L_1\utimes L_2}$ are spanned by the sections of each factor, $\Sec{L_1}$, $\Sec{L_2}$, and the functions on the base product. More precisely, we have the following isomorphisms of $\Cin{M_1\dtimes M_2}$-modules:
\begin{equation*}
    \Cin{M_1\dtimes M_2}\cdot P_1^*\Sec{L_1}\cong \Sec{L_1 \utimes L_2} \cong \Cin{M_1\dtimes M_2}\cdot P_2^*\Sec{L_2},
\end{equation*}
where for the second isomorphism we have used the canonical factor $C_{12}$ to establish the isomorphisms $\Cin{M_1\dtimes M_2}\cong \Cin{M_2\dtimes M_1}$ and $\Sec{L_1 \utimes L_2}\cong \Sec{L_2 \utimes L_1}$ as $\Real$-vector spaces. We can summarize this by writing the image of the line product commutative diagram (\ref{lproduct}) under the section functor
\begin{equation*}
\begin{tikzcd}[row sep=tiny]
\Sec{L_1} \arrow[r,"P_1^*"] & \Sec{L_1\utimes L_2}  & \Sec{L_2}\arrow[l,"P_2^*"'] \\
\bullet & \bullet & \bullet \\
\Cin{M_1}\arrow[r,"p_1^*"'] & \Cin{M_1 \dtimes M_2} & \Cin{M_2}\arrow[l,"p_2^*"]
\end{tikzcd}
\end{equation*}
where $\bullet$ denotes ring-module structure. We clearly see that pull-backs of functions on the Cartesian product of base manifolds $M_1\times M_2$ form a subring of $\Cin{M_1\dtimes M_2}$ but a quick computation introducing trivializations shows that differentials of these do not span the cotangent bundle of the base product everywhere. However, the proposition below shows that it is possible to identify a subspace of spanning functions on $M_1 \dtimes M_2$ defined from local sections of $L_1$ and $L_2$. Consider two trivializations $L|_{U_i}\cong U_i\times \Real$, $i=1,2$ so that the spaces of non-vanishing sections over the trivializing neighbourhoods, denoted by $\Sec{L^\times_i}$, are non-empty. Recall that the ratio maps introduced in section \ref{CatLine} allowed us to define functions over invertible maps between lines, we can use these fibre-wise together with the choice of two local sections to define the \textbf{ratio functions} on $M_1\dtimes M_2$ as follows: let the local sections $s_1\in\Sec{L_1}$ and $s_2\in\Sec{L_2^\times}$, then the function $\tfrac{s_1}{s_2}\in \Cin{M_1\dtimes M_2}$ is defined by
\begin{equation*}
    \frac{s_1}{s_2}(B_{x_1x_2}):=\frac{B_{x_1x_2}(s_1(x_1))}{s_2(x_2)}
\end{equation*}
where we have made use of the map $l^{12}$ in (\ref{ratio2}) in fraction notation. Similarly, using the map $r^{12}$ in (\ref{ratio2}) we can define $\tfrac{s_2}{s_1}\in \Cin{M_1\dtimes M_2}$ for some local sections $s_1\in\Sec{L_1^\times}$ and $s_2\in\Sec{L_2}$. If we consider two non-vanishing local sections $c_1\in\Sec{L_1^\times}$ and $c_2\in\Sec{L_2^\times}$, it is a direct consequence of (\ref{ratio1}) that the relation
\begin{equation*}
    \frac{c_1}{c_2}\cdot \frac{c_2}{c_1}=1.
\end{equation*}
holds on the open set $(p_1\times p_2)^{-1}(U_1 \times U_2)$ where $U_i\subseteq M_i$ are the neighbourhoods where the local sections are defined.
\begin{prop}[Spanning Functions of $M_1\dtimes M_2$]\label{SpanFuncProd}
Let $B_{x_1x_2}\in M_1\dtimes M_2$ any point in the base product of two line bundles $L_{M_1}$, $L_{M_2}$, consider trivializing neighbourhoods $x_i\in U_i\subseteq M_1$ so that $\Sec{L_i^\times}$ are non-empty; then the cotangent bundle at $B_{x_1x_2}$ is spanned by the differentials of all possible ratio functions, i.e.
\begin{equation*}
\normalfont
    \Cot(M_1\dtimes M_2)=d\frac{\Sec{L_1}}{\Sec{L_2^\times}}
\end{equation*}
when restricted to the open set $(p_1\times p_2)^{-1}(U_1 \times U_2)$.
\end{prop}
\begin{proof}
First note that it is a general fact from basic vector bundle geometry that such open neighbourhoods $U_i$ can be chosen for any points $x_i\in M_i$. The trivializations induce the following diffeomorphism
\begin{equation*}
    (p_1\times p_2)^{-1}(U_1 \times U_2) \cong U_1\times U_2\times \Real^\times
\end{equation*}
and thus we can assign coordinates $(y_1,y_2,b)$, where $y_i\in\Real^{\text{dim}M_i}$ and $b>0$, to any element $B_{x_1x_2}\in (p_1\times p_2)^{-1}(U_1 \times U_2)$. Local sections $s_i\in\Sec{L_i}$ are given by  functions $f_i:U_i\to \Real$ under the trivialization. It follows directly from the definition that ratio functions have the following coordinate expressions
\begin{equation*}
    \frac{s_1}{s_2}(y_1,y_2,b)=f_1(y_1) f_2^{-1}(y_2)b
\end{equation*}
where $f_1$, $f_2$ are the trivialized expressions for the sections $s_1\in\Sec{L_1}$, $s_2\in\Sec{L_2^\times}$. A simple computation shows that
\begin{equation*}
    d_{(y_1,y_2,b)}\frac{s_1}{s_2}=f^{-1}_2b\cdot d_{y_1}f_1-f_1f_2^{-2}b\cdot d_{y_2}f_2+f_1f_2^{-1}\cdot d_bb
\end{equation*}
Since $f_1(y_1)$ is allowed to vanish in general, these differentials span all the cotangent vectors at any coordinate point $(y_1,y_2,b)$.
\end{proof}

One of the first canonical examples of vector bundle in elementary differential geometry is the tangent bundle of a smooth manifold $\Tan M$. This is a vector bundle whose sections are identified with vector fields which, in turn, give derivations of $\Real$-valued functions, $\Dr{\Cin{M}}\cong \Sec{\Tan M}$, where the isomorphism arises from the construction of local directional derivatives. Furthermore, differentials of smooth maps give vector bundle morphisms that respect the Lie bracket of vector fields. This can be summarized by identifying the tangent construction as a functor sending smooth manifolds to Lie algebroids with identity anchor
\begin{equation*}
    \Tan:\Man \to \Lie_\Man\subset \Vect_\Man.
\end{equation*}
Our proposal to regard a line bundle $L$ as a unit-free manifold entails thinking of sections $\Sec{L}$ as the generalization of $\Real$-valued functions. Simply by looking at the definition of \textbf{der bundle} given above, one immediately notices that $\Der L$ is a bundle whose fibres realize the unit-free version of local directional derivatives. In analogy with the usual notation of the action a vector field $X$ on a function $f$ as a directional derivative $X[f]$, given a line bundle $L_M$ we define the action of a section of the der bundle, $a\in \Sec{\Der L}$, as a derivation on sections by setting
\begin{equation*}
    a[s](x):=a(x)(s)\in L_x, \quad \forall x\in M.
\end{equation*}
We then write the Lie bracket of derivations as a bracket on sections of $\Der L$ explicitly
\begin{equation*}
    [a,a'][s]:=a[a'[s]]-a'[a[s]]
\end{equation*}
for all $a,a'\in \Sec{\Der L}$ and $s\in\Sec{L}$. It is then straightforward to check that $(\Der L,\delta,[,] )$ is a Lie algebroid where $\delta$ denotes the anchor given by the local symbol map.\newline

Let a factor $B:L_1\to L_2$ covering a smooth map $b:M_1\to M_2$, we define its \textbf{der map} at a point $x_1\in M_1$ as follows:
\begin{align*}
\Der_{x_1} B: \Der_{x_1} L_1 & \to \Der_{b(x_1)} L_2\\
a_{x_1} & \mapsto \Der_{x_1}a_{x_1} 
\end{align*}
where
\begin{equation*}
    \Der_{x_1}a_{x_1}(s_2):=B_{x_1}(a_{x_1}(B^*s_2)) \quad \forall s_2\in\Sec{L_2}.
\end{equation*}
Given two derivations, $a_1\in\Sec{\Der L_1}$ and $a_2\in\Sec{\Der L_2}$ we say that they are $B$\textbf{-related}, and denote $a_1\sim_B a_2$, if the following diagram commutes
\begin{equation*}
    \begin{tikzcd}
    \Der L_1 \arrow[r, "\Der B"] & \Der L_2 \\
    M_1 \arrow[r, "b"']\arrow[u,"a_1"] & M_2.\arrow[u,"a_2"]
\end{tikzcd}
\end{equation*}
When the base map of the factor $b$ is a diffeomorphism we can define the \textbf{der push-forward} of derivations by
\begin{align*}
B_*: \Sec{\Der L_1} & \to \Sec{\Der L_2}\\
a & \mapsto \Der B \circ a \circ b^{-1}.
\end{align*}
With our identification of derivations with sections of the der bundle and the definition of pull-back of sections, we can readily check the following identity:
\begin{equation*}
    B_*a[s]=(B^{-1})^*a[B^*s] \qquad \forall s\in \Sec{L_2},
\end{equation*}
which gives an alternative definition of der push-forward only in terms of pull-backs of diffeomorphic factors. The following proposition establishes the der bundle as the adequate generalization of unit-free tangent bundle.

\begin{prop}[The Der Functor]\label{DerFunct}
The der map of a factor $B:L_1 \to L_2$ gives a well-defined vector bundle morphism
\begin{equation*}
\normalfont
\begin{tikzcd}
\Der L_1 \arrow[r, "\Der B"] \arrow[d, "\lambda_1"'] & \Der L_2 \arrow[d, "\lambda_2"] \\
M_1 \arrow[r, "b"'] & M_2
\end{tikzcd}
\end{equation*}
which is a Lie algebroid morphism, i.e.
\begin{equation*}
\normalfont
    \Tan b\circ \delta_1 = \delta_2\circ \Der B \qquad  \qquad a_1\sim_B a_2 ,\quad a'_1\sim_B a'_2 \quad \Rightarrow \quad [a_1,a'_1]_1\sim_B [a_2,a'_2]_2
\end{equation*}
for $\normalfont a_i,a_i'\in\Sec{\Der L_i}$. Furthermore, for any other factor $F:L_2 \to L_3$ and $\normalfont \Id_L:L\to L$ the identity factor, we have
\normalfont
\begin{align*}
    \Der (F\circ B) & = \Der F \circ \Der B\\
    \Der (\Id_L) & = \Id_{\Der L}.
\end{align*}
\end{prop}
\begin{proof}
Note first that  $B^*(f\cdot s)=b^*f\cdot B^*s$ for $f\in\Cin{M_2}$ and $s\in\Sec{L_2}$, then applying the definition of point-wise der map one obtains
\begin{equation*}
    \Der_xBa_x(f\cdot s)=f(b(x))\cdot \Der_x B a_x(s)+\delta_1(a_x)(b^*f)(x)\cdot s(b(x)).
\end{equation*}
Simply using the definition of tangent map $\Tan b$ we can rewrite the second term to find
\begin{equation*}
    \Der_xBa_x(f\cdot s)=f(b(x))\cdot \Der_x B a_x(s)+\Tan_xb\circ \delta_1(a_x)(f)(b(x))\cdot s(b(x))
\end{equation*}
which shows that indeed $\Der_xBa_x\in \Der_{b(x)}L_2$, making $\Der B$ a well-defined vector bundle morphism. We can then use the anchor $\delta_2$ to rewrite the LHS
\begin{equation*}
    \Der_xBa_x(f\cdot s)=f(b(x))\cdot \Der_x B a_x(s)+\delta_2(\Der_xB a_x)(f)(b(x))\cdot s(b(x))
\end{equation*}
and thus we obtain
\begin{equation*}
    T_xb\circ \delta_1(a_x)=\delta_2(\Der_xB a_x),
\end{equation*}
which must hold for any $a_x\in\Der_xL_1$, thus giving the desired compatibility condition with the anchors. The $B$-relatedness condition follows directly from the definition of der map by noting that
\begin{equation*}
    a_2(b(x))(a'_2[s])=(\Der_x B a_1(x))(a'_2[s])=B_x(a_1(x)(B^*a_2[s]))
\end{equation*}
for all $s\in\Sec{L_2}$. The two functorial identities follow from contravariance of pull-backs $(F\circ B)^*=B^*\circ F^*$ and the trivial fact $\Id_L^*=\Id_{\Sec{L}}$.
\end{proof}

In section \ref{CatLine}  we defined the category $\LVect$ and we argued that it adequately implements the notion of unit-free vector spaces. We define the \textbf{category of L-vector bundles} $\LVect_\Man$ in a similar fashion to how one defines the category of vector bundles $\Vect_\Man$ as the class of smooth fibrations of objects in $\Vect$ with morphisms the smooth maps that fibre-wise restrict to morphisms of $\Vect$. Thus, the objects in the category of L-vector bundles are pairs of vector bundles and line bundles over the same base manifold $(E_M,L_M)$, and the morphisms are pairs of vector bundle morphisms and factors covering the same smooth map. The following notations will be used indistinctly to denote a L-vector bundle: $E^L$, $E_M^L$ and $L$ may be omitted altogether as superscript if it is understood from the context. Similar notational conventions will apply for L-vector bundle morphisms. Fixing a base manifold $M$, we have the subcategory of L-vector bundles over a single base $\LVect_M$, there we can define \textbf{L-direct sums}, \textbf{L-tensor products}, and \textbf{L-duals} by using the fibre-wise L-linear structure and following the definitions given at the end of section \ref{CatLine}. Pull-back bundles are defined in the obvious way by taking the simultaneous pull-back of both the vector bundle and the line bundle over the same base. It is clear from our discussion on der bundles above that these represent a first class of examples of L-vector bundles and that der maps of factors are L-vector bundle morphisms.\newline

One of the key properties of the tangent functor $\Tan$ is that it allows to regard tangent bundles of submanifolds as subbundles. More concretely, if $i:S\hookrightarrow M$ is a submanifold we can regard its tangent bundle $\Tan S$ as a subbbundle of $\Tan M$ via the injective bundle map $\Tan i:\Tan S\hookrightarrow \Tan M$. Vector fields of the submanifold $S$ are identified with vector fields of $M$ that restrict tangentially or, equivalently, those vector fields on $M$ which preserve the vanishing ideal of the submanifold when acting as derivations modulo vector fields identically vanishing on the manifold, i.e  there is a natural isomorphism
\begin{equation*}
    \Sec{\Tan S}\cong \{X\in\Sec{\Tan M}|\quad X[I_S]\subseteq I_S\}/\Gamma_{0S}(\Tan M).
\end{equation*}
The following two propositions give the corresponding results for der bundles.
\begin{prop}[Der Bundle of a Submanifold]\label{DerSubMan}
Let $\lambda:L\to M$ be a line bundle and $i:S\hookrightarrow M$ an embedded submanifold of the base. The canonical embedding factor $\iota : L_S\hookrightarrow L$ gives an injective L-vector bundle morphism
\begin{equation*}
\normalfont
    \Der \iota : \Der L_S \hookrightarrow \Der L.
\end{equation*}
By a slight abuse of notation we will write $\Der L_S\subseteq \Der L$ in the same way that we write $\Tan S\subseteq \Tan M$, then we have
\begin{equation*}
    \delta_L(\Der L_S)=\Tan S,
\end{equation*}
where $\delta_L:\Der L\to \Tan M$ is the anchor of the der bundle.
\end{prop}
\begin{proof}
That $\Der \iota$ is a L-vector bundle morphism follows by construction since $L_S:=i^*L$, so the line bundle morphism is simply the fibre-wise identity map. Then, injectivity of $\Der \iota$ follows simply from injectivity of $i:S\hookrightarrow M$. Note that using the full notation the second identity in the proposition reads
\begin{equation*}
    \delta_L(\Der \iota (\Der L_S))=\Tan i(TS)
\end{equation*}
which is clearly a direct consequence of $\Der \iota$ being a Lie algebroid morphism, in particular compatible with the anchors, and the anchors being surjective so that $\delta_L(\Der L)=\Tan M$ and $\delta_{L_S}(\Der L_S)=\Tan S$.
\end{proof}
\begin{prop}[Derivations of a Submanifold]\label{DerivSubMan}
Let $\lambda:L\to M$ be a line bundle, $i:S\hookrightarrow M$ an embedded submanifold of the base and denote by $\Gamma_S\subseteq\Sec{L}$ its vanishing submodule. We define the derivations that tangentially restrict to $S$ as
\begin{equation*}
\normalfont
    \text{Der}_S(L):=\{D\in\Dr{L}|\quad D[\Gamma_S]\subseteq \Gamma_S\}
\end{equation*}
and the derivations that vanish on $S$ as
\begin{equation*}
\normalfont
    \text{Der}_{0S}(L):=\{D\in\Dr{L}|\quad D[\Sec{L}]\subseteq \Gamma_S\},
\end{equation*}
then there is a natural isomorphism of $\Cin{M}$-modules and Lie algebras
\begin{equation*}
\normalfont
    \Dr{L_S}\cong \text{Der}_S(L)/\text{Der}_{0S}(L). 
\end{equation*}
\end{prop}
\begin{proof}
The isomorphism as modules follows directly from Proposition \ref{DerSubMan} using the correspondence between sections of the der bundle and derivations. The Lie algebra isomorphism is then a consequence of the following simple facts
\begin{align*}
    [\text{Der}_S(L),\text{Der}_S(L)] & \subseteq  \text{Der}_S(L)\\
    [\text{Der}_S(L),\text{Der}_{0S}(L)] & \subseteq  \text{Der}_{0S}(L)\\
    [\text{Der}_{0S}(L),\text{Der}_{0S}(L)] & \subseteq  \text{Der}_{0S}(L),
\end{align*}
easily derived from the definitions above, thus showing that $\text{Der}_S(L)\subseteq \Dr{L}$ is a Lie subalgebra and $\text{Der}_{0S}(L)\subseteq\text{Der}_S(L)$ is a Lie ideal making the subquotient $\text{Der}_S(L)/\text{Der}_{0S}(L)$ into a Lie algebra reduction.
\end{proof}

In the presence of a free and proper group action $G \Acts M$ so that there is a surjective submersion $q:M\to M/G$, the tangent functor gives the equivalent description for point-wise tangent spaces: $\Tan_[x](M/G)\cong \Tan_xM/\mathfrak{g}$, with $\mathfrak{g}$ the Lie algebra of $G$. Furthermore, vector fields on the orbit space are identified with the $G$-invariant vector fields that tangentially restrict to the orbits $\Sec{\Tan (M/G)}\cong\Sec{\Tan M}^G/\Sec{\Tan M}^{0G}$. The following proposition gives an analogous result for the der functor.

\begin{prop}[Der Bundle of a Group Action Quotient]\label{DerGAct}
Let $\Phi:G\times L\to L$ be a free and proper line bundle $G$-action with infinitesimal counter part $\normalfont \Psi:\mathfrak{g}\to \Dr{L}$ and let us denote by $\zeta: L\to L/G$ the submersion factor given by taking the quotient onto the space of orbit line bundles. Defining the $G$-invariant derivations as
\begin{equation*}
\normalfont
    \Dr{L}^G:=\{D\in\Dr{L}|\quad (\Phi_g)_*D=D \quad \forall g\in G\}=\{D\in\Dr{L}|\quad D[\Sec{L}^G]\subseteq \Sec{L}^G\}
\end{equation*}
and the derivations that tangentially restrict to the orbits as
\begin{equation*}
\normalfont
     \Dr{L}^{0G}:=\{D\in\Dr{L}|\quad D[\Sec{L}^G]=0\}
\end{equation*}
we find the following isomorphism of modules and Lie algebras
\begin{equation*}
\normalfont
    \Dr{L/G}\cong \Dr{L}^G/\Dr{L}^{0G}.
\end{equation*}
When $G$ is connected, this isomorphism becomes
\begin{equation*}
\normalfont
    \Dr{L/G}\cong \Dr{L}^{\mathfrak{g}}/\Psi(\mathfrak{g}).
\end{equation*}
where the $\mathfrak{g}$-invariant derivations are defined as
\begin{equation*}
\normalfont
    \Dr{L}^{\mathfrak{g}}:=\{D\in\Dr{L}|\quad [\Psi(\xi),D]=0 \quad \forall \xi\in \mathfrak{g}\}.
\end{equation*}
When applied point-wise, this last isomorphism gives a fibre-wise isomorphism of der spaces over the base orbit space $M/G$:
\begin{equation*}
\normalfont
    \Der_{[x]}(L/G)\cong \Der_xL/\mathfrak{g}.
\end{equation*}
\end{prop}
\begin{proof}
The module isomorphism $\Dr{L/G}\cong \Dr{L}^G/\Dr{L}^{0G}$ is a direct consequence of the ring isomorphism $\Cin{M/G}\cong \Cin{M}^G$ and the module isomorphism $\Sec{L/G}\cong\Sec{L}^G$. The isomorphism as Lie algebras follows from the fact that, by construction, we have
\begin{align*}
    [\Dr{L}^G,\Dr{L}^G] & \subseteq  \Dr{L}^G\\
    [\Dr{L}^G, \Dr{L}^{0G}] & \subseteq   \Dr{L}^{0G}\\
    [ \Dr{L}^{0G}, \Dr{L}^{0G}] & \subseteq   \Dr{L}^{0G},
\end{align*}
thus showing that $\Dr{L}^G\subseteq \Dr{L}$ is a Lie subalgebra and $\Dr{L}^{0G}\subseteq\Dr{L}^G$ is a Lie ideal making the subquotient $\Dr{L}^G/\Dr{L}^{0G}$ into a Lie algebra reduction. When $G$ is connected, its action is uniquely specified by the infinitesimal counterpart, which can be equivalently regarded as a Lie algebroid morphism $\Psi:(\mathfrak{g}_M,\psi,[,])\to (\Der L,\delta,[,])$, where $\mathfrak{g}_M=M\times \mathfrak{g}$ here denotes the action Lie algebroid with anchor given by the infinitesimal action on the base manifold $\psi:\mathfrak{g}\to \Sec{\Tan M}$. It is clear by construction that for the action of a connected $G$ we have $\Dr{L}^{0G}=\Psi(\mathfrak{g})$ and that $G$-invariance under push-forward, $(\Phi_g)_*D=D$, becomes vanishing commutator with the infinitesimal action, $[\Psi(\xi),D]=0$. Then the second isomorphism follows. Since the action is free and proper, the map $\Psi$ will be injective as a vector bundle morphism covering the identity map on $M$, then applying the second isomorphism point-wise and injectivity of the infinitesimal action, so that $\Psi(\mathfrak{g})_x\cong \mathfrak{g}$, we find the last desired result.
\end{proof}

Another crucial property of the tangent functor is that it sends Cartesian products to direct sums, i.e. if $M_1$ and $M_2$ are smooth manifolds and we denote by $\Proj_i:M_1\times M_2\to M_i$ the canonical projections, we have the following vector bundle isomorphism
\begin{equation*}
    \Tan (M_1\times M_2)\cong \Proj_1^*\Tan M_1 \oplus \Proj_2^* \Tan M_2.
\end{equation*}
Furthermore, this isomorphism allows to lift vector fields on each of the factors $X_i\in\Sec{\Tan M_i}$ to vector fields on the Cartesian product $X_1\oplus 0,0\oplus X_2\in\Sec{\Tan (M_1\times M_2)}$. The two propositions that follow below establish similar results for the der functor.
\begin{prop}[Der Bundle of the Line Product]\label{DerLineProd}
Let two line bundles $L_1, L_2\in \Line_\Man$ and take their line product $L_1\utimes L_2$, then there is a canonical isomorphism of L-vector bundles covering the identity map on the base product:
\begin{equation*}
\normalfont
    \Der (L_1\utimes L_2)\cong p_1^*\Der L_1 \oplus p_2^*\Der L_2.
\end{equation*}
\end{prop}
\begin{proof}
Note first that both vector bundles are, by construction, vector bundles over the same base manifold $M_1\dtimes M_2$. The first term $p_1^*\Der L_1$ is clearly an L-vector bundle with line $p_1^*L_1$ and although the second term is an L-vector bundle with line $p_2^*L_2$ we can use the canonical factor $C_{12}$ to isomorphically regard $p_2^*\Der L_2$ as an L-vector bundle with line $p_1^*L_1$. Then the direct sum as L-vector bundles is well-defined and giving an L-vector bundle isomorphism now reduces to finding a vector bundle isomorphism. A quick check introducing trivializations (see the end of this section for details) shows that both vector bundles have the same rank and so it is enough to find a fibre-wise surjective map $\Phi$ between the vector bundles. We can write this map explicitly as
\begin{align*}
\Phi_{B_{x_1x_2}}:\Der_{B_{x_1x_2}}(L_1\utimes L_2) & \to \Der_{x_1}L_1\oplus \Der_{x_2}L_2\\
a & \mapsto \Der_{B_{x_1x_2}}P_1(a) \oplus \Der_{B_{x_1x_2}}P_2(a)
\end{align*}
where $P_1$ and $P_2$ are the line product projection factors as in (\ref{lproduct}). Surjectivity of this map follows directly from the definition of the der map and the fact that projection factors are surjective.
\end{proof}
\begin{prop}[Derivations of the Line Product]\label{DerivLineProd}
Let two line bundles $L_1, L_2\in \Line_\Man$ and take their line product $L_1\utimes L_2$, then we find the derivations of each factor as submodules of the derivations of the product
\begin{equation*}
\normalfont
\begin{tikzcd}
\Dr{L_1}\arrow[r,hook, "k_1"] & \Dr{L_1\utimes L_2} & \Dr{L_2}\arrow[l,hook',"k_2"'].
\end{tikzcd}
\end{equation*}
The maps $k_i$ are Lie algebra morphisms making $\normalfont \Dr{L_i}\subseteq\Dr{L_1\utimes L_2}$ into Lie subalgebras which, furthermore, satisfy
\begin{equation*}
\normalfont
    [\Dr{L_1},\Dr{L_1}]\subseteq \Dr{L_1} \qquad [\Dr{L_1},\Dr{L_2}]=0 \qquad [\Dr{L_2},\Dr{L_2}]\subseteq\Dr{L_2}.
\end{equation*}
\end{prop}
\begin{proof}
This can be proved using the isomorphism in Proposition \ref{DerLineProd} above and the fact that sections of the der bundle are identified with derivations. However, we give an independent proof that highlights that module structure of the sections of the line product. Recall that sections of the line product $\Sec{L_1\utimes L_2}$ are spanned by the pull-backs $P^*_i\Sec{L_i}$ over the functions on the base product $\Cin{M_1\dtimes M_2}$. This means that a derivation $D$ is characterised by its action on projection pull-backs and the action of its symbol $X$ on spanning functions of the base product, which are the ratio functions defined in Proposition \ref{SpanFuncProd}. With this in mind, given derivations $D_i\in\Dr{L_i}$ we give the derivations on the line product $k_i(D_i)=\overline{D}_i\in\Dr{L_1\utimes L_2}$ determined uniquely by the conditions
\begin{align*}
    \overline{D}_1(P^*_1s_1)&=P^*D_1(s_1) &\overline{D}_1(P^*_2s_2)&=0\\
    \overline{D}_2(P^*_1s_1)&=0 &\overline{D}_2(P^*_2s_2)&=P^*_2D_2(s_2)
\end{align*}
and with symbols $\overline{X}_i\in\Sec{T(M_1\dtimes M_2)}$ defined on ratio functions by
\begin{align*}
    \overline{X}_1[\tfrac{s_1}{u_2}]=\tfrac{D_1(s_1)}{u_2}\\
    \overline{X}_2[\tfrac{s_2}{u_1}]=\tfrac{D_2(s_2)}{u_1}
\end{align*}
for all sections $s_i\in\Sec{L_i}$ and local non-vanishing sections $u_i\in\Sec{L_i^\times}$. The fact that $k_i$ are Lie algebra morphisms and that $[\overline{D}_1,\overline{D}_2]=0$ follows directly from the defining conditions above.
\end{proof}

In summary, we clearly see that the der functor for line bundles
\begin{equation*}
    \Der : \Line_\Man \to \Lie_\Man\subset \LVect
\end{equation*}
plays a categorical role entirely analogous to that of the the tangent functor for smooth manifolds.\newline

For a conventional manifold one defines the cotangent bundle as the vector bundle of point-wise differentials of functions, which is indeed dual to the tangent bundle $\Cot M\cong (\Tan M)^*$. Differentials of functions give sections of this bundle and one defines the exterior algebra of differential forms from the exterior derivative $d:\Cin{M}\to \Sec{\Cot M}$, which also induces the seminal identity of Cartan calculus on $M$: $X[f]=i_Xdf=\mathcal{L}_Xf$. When we consider a section of a line bundle $s\in\Sec{L}$ as the generalization of a unit-free function, the analogue of a point-wise differential is the 1-jet value $j_xs\in \Jet^1L$, thus hinting to the bundle of 1-jets as the analogue of the cotangent bundle for unit-free manifolds. It follows from our remarks about vector bundles of rank $1$ above that the \textbf{jet bundle} $\Jet^1 L$ is indeed the L-dual to the der bundle
\begin{equation*}
    \Jet^1 L=(\Der L)^*\otimes L =: (\Der L)^{*L}.
\end{equation*}
The action of a section of the der bundle $a\in\Sec{\Der L}$ on a section $s\in\Sec{L}$ as a local derivation can now be rewritten as
\begin{equation*}
    a[s]=j^1s(a)
\end{equation*}
where the jet prolongation $j^1:\Sec{L}\to \Sec{\Jet^1 L}$ gives the unit-free generalization of the exterior derivative. In fact, this map extends to a graded differential on the \textbf{der complex} of the line bundle $(\Sec{\wedge^\bullet(\Der L)^*\otimes L},d_L)$. Recall that any Lie algebroid with a representation on a vector bundle carries a natural Cartan calculus, in the case of the der bundle of a line bundle this is automatically the case with the identity map as the representation. These remarks confirm that the jet bundle of a line bundle is indeed the unit-free analogue of the cotangent bundle. The geometry of jet bundles will be extensively studied in section \ref{CanHamJ}.\newline

The symbol (\ref{diff1}) and Spencer (\ref{spncr1}) short exact sequences find a straightforward reformulation in the category of L-vector bundles. Note that, for any line bundle $\lambda:L\to M$, the anchor map $\delta:\Der L\to \Tan M$ is surjective by construction and elements of its kernel correspond to bundle endomorphisms induced from fibre-wise scalar multiplication, $\Ker{\delta}\cong \Real_M$. Then, regarding $\Der L$ and $\Tan M$ as L-vector bundles and $\delta$ as a L-vector bundle morphism in a trivial way, the symbol sequence is simply the short exact sequence of L-vector bundles induced by the fact that $\delta$ is surjective
\begin{equation*}
\begin{tikzcd}
0 \arrow[r] & \Real_M \arrow[r] & \Der L \arrow[r, "\delta"] & \Tan M \arrow[r] & 0.
\end{tikzcd}
\end{equation*}
This will be sometimes called the \textbf{der sequence} of the line bundle $L$. The Spencer sequence then corresponds precisely to the L-dual of the short exact sequence above
\begin{equation*}
\begin{tikzcd}
0 &  L \arrow[l] & \Jet^1L \arrow[l]& (\Tan M)^{*L} \arrow[l,"i"'] & 0 \arrow[l]
\end{tikzcd}
\end{equation*}
where $i=\delta^{L*L}$ is injective. We will refer to this sequence as the \textbf{jet sequence} of the line bundle $L$. We will use the notation $\Tan^{*L} M:= (\Tan M)^{*L}$ in analogy with the usual notation for cotangent bundles.\newline

We end this section by giving a summary of the results proved above for the particular case of trivial line bundles. Note that our discussion will also apply to general line bundles when restricted to trivializing neighbourhoods. Sections of a trivial line bundle $\Real_M$ are isomorphic to the smooth functions of the base $\Sec{\Real_M}\cong \Cin{M}$ with the module structure being simply point-wise multiplication. This implies that there is now a natural inclusion $\Dr{\Cin{M}}\subseteq\Dr{\Real_M}$ making the der short exact sequence split and thus giving an isomorphism of modules
\begin{equation*}
    \Dr{\Real_M} \cong \Sec{\Tan M}\oplus \Cin{M}.
\end{equation*}
The action of a derivation on a function $s\in \Sec{\Real_M}$, regarded as a section of the trivial bundle $\Real_M$, is given by
\begin{equation*}
    (X\oplus f)[s]=X[s]+fs.
\end{equation*}
From this it is easy to see that the der bundle is
\begin{equation*}
    \Der \Real_M \cong \Tan M \oplus \Real_M
\end{equation*}
with anchor $\delta=\Proj_1$ and Lie bracket bracket
\begin{equation*}
    [X\oplus f,Y\oplus g]=[X,Y]\oplus X[g]-Y[f].
\end{equation*}
Note how this Lie bracket is entirely determined by the fact that vector fields are the Lie algebra of derivations on functions. Taking $\Real_M$-duals corresponds to taking ordinary duals, therefore the jet bundle is
\begin{equation*}
    \Jet^1 \Real_M\cong \Cot M\oplus \Real_M.
\end{equation*}
The jet prolongation map then becomes
\begin{align*}
j^1: \Cin{M} & \to \Sec{\Cot M}\oplus \Cin{M}\\
s & \mapsto ds\oplus s,
\end{align*}
which indeed only carries the information of the ordinary exterior differential. The base product of two trivial line bundles $\Real_M$ and $\Real_N$ is
\begin{equation*}
    M\dtimes N\cong M\times N\times \Real^\times
\end{equation*}
and the line product is again a trivial line bundle
\begin{equation*}
    \Real_M\utimes \Real_N \cong \Real_{M\dtimes N}.
\end{equation*}
A factor between trivial line bundles $B:\Real_M\to \Real_N$ is given by a pair $B=(b,\beta)$ with $b:M\to N$ a smooth map and $\beta\in\Cin{M}$ a nowhere-vanishing function. We have explicitly
\begin{align*}
(b,\beta): \Real_M & \to \Real_N\\
(x,l) & \mapsto (b(x),\beta(x)l).
\end{align*}
Pull-backs then become
\begin{equation*}
    (b,\beta)^*s=\tfrac{1}{\beta}\cdot b^*s
\end{equation*}
for all $s\in\Cin{N}$. A simple computation shows that the der map of a factor $(b,\beta)$ gives a map of the form
\begin{align*}
\Der (b,\beta): \Tan M\oplus \Real_M & \to \Tan N\oplus \Real_N\\
v_x\oplus a_x & \mapsto \Tan_x b(v_x)\oplus a_x-\tfrac{1}{\beta(x)}d_x\beta(v_x).
\end{align*}
For a diffeomorphic factor $(b,\beta):\Real_M\to \Real_M$, i.e. when $b$ is a diffeomorphism, the der push-forward of a derivation is given by
\begin{equation*}
    (b,\beta)_*(X\oplus f)=b_* X \oplus b_*f +\beta\cdot X[\tfrac{1}{\beta}].
\end{equation*}

\subsection{Summary of Poisson and Jacobi Geometry} \label{JacGeo}

In preparation for our discussion on geometric mechanics we summarize the basic definitions and results of Poisson and Jacobi geometry. An excellent modern reference on the theory of Poisson manifolds is \cite{fernandes2014lectures}. We adopt L. Vitagliano's point of view that line bundles are the natural category where Jacobi structures are defined, see e.g. \cite{vitagliano2014deformations}, \cite{tortorella2017deformations} or \cite{vitagliano2015dirac}. Most of the results we prove here are the non-trivial counterparts of the coisotropic calculus for trivial Jacobi manifolds of \cite{ibanez1997coisotropic}.\newline

A \textbf{Poisson algebra} is a triple $(A,\cdot,\{,\})$ with $A$ a $\mathbb{R}$-vector space, $(A,\cdot)$ a commutative (unital, associative) algebra and $(A,\{,\})$ a Lie algebra such that the following \textbf{Leibniz rule} holds
\begin{equation*}
    \{a,b\cdot c\}=\{a,b\}\cdot c + b\cdot \{a,c\}.
\end{equation*}
Equivalently, this condition can be rephrased as
\begin{equation*}
    \text{ad}_{\{,\}}:A\to \Dr{A,\cdot}
\end{equation*}
making the adjoint map of the Lie algebra a homomorphism of Lie algebras from $A$ to the derivations of both bilinear products on $A$, $\Dr{A,\{,\}}\cap \Dr{A,\cdot}$.The algebra endomorphisms given by the adjoint map of the Lie bracket are usually called \textbf{Hamiltonian derivations} and are denoted by $X_a:=\{a,-\}$. A linear map $\psi:A\to B$ is called a \textbf{Poisson homomorphism} if $\psi:(A,\cdot)\to (B,\cdot)$ is a commutative algebra homomorphism and $\psi:(A,\{,\})\to (B,\{,\})$ is a Lie algebra homomorphism. A subspace $I\subseteq A$ is called a \textbf{coistrope} if it is a multiplicative ideal and a Lie subalgebra.\newline

A \textbf{Poisson manifold} is a smooth manifold $P$ whose commutative algebra of smooth functions has the structure of a Poisson algebra $(\Cin{P},\cdot,\{,\})$, we also use the term \textbf{Poisson structure} to refer to this datum. A \textbf{Poisson map} is a smooth map between Poisson manifolds whose pull-back on functions gives a morphism of Poisson algebras. Hamiltonian derivations become \textbf{Hamiltonian vector fields} in this context and are denoted by
\begin{equation*}
    f\in\Cin{P}\mapsto X_f=\{f,-\}\in\Sec{TP}.
\end{equation*}
It follows directly from this definition that the assignment of Hamiltonian vector fields is a $\mathbb{R}$-linear Lie algebra homomorphism,
\begin{equation*}
    X_{\{f,g\}}=[X_f,X_g].
\end{equation*}
The observation that derivations on functions are isomorphic to vector fields allows for a more geometric definition of a Poisson manifold as a pair $(P,\pi)$ with $\pi\in\Sec{\wedge^2TM}$ a bivector field, called the \textbf{Poisson bivector}, that satisfies an integrability condition\footnote{In the Schouten algebra of multivector fields extending the Lie algebra of vector fields $(\Sec{\wedge^\bullet TM},[,])$, a Poisson bivector must satisfy $[\pi,\pi]=0$.} which is equivalent to the Jacobi identity. The bracket on smooth functions then becomes
\begin{equation*}
    \{f,g\}_\pi:=\pi(df,dg).
\end{equation*}
Note that this bracket satisfies the Leibinz rule automatically by construction. The bivector gives a musical map
\begin{equation*}
    \pi^\sharp:T^*P\to TP
\end{equation*}
which is a vector bundle morphism covering the identity. The Hamiltonian vector field of a function $f\in\Cin{P}$ can now be given explicitly as
\begin{equation*}
    X_f=\pi^\sharp(df).
\end{equation*}
The map $X_{-}=\pi^\sharp\circ d:\Cin{P}\to\Sec{\Tan P}$ is usually called the \textbf{Hamiltonian map}. An important class of Poisson manifold are the \textbf{fibre-wise linear Poisson manifolds} defined as Poisson structures on the total space of a vector bundle $\epsilon:E\to M$ whose bracket is compatible with the sub-module of fibre-wise constant functions $\epsilon^*(\Cin{M})$ and fibre-wise linear functions $l(\Sec{E^*})$ in the following sense:
\begin{align*}
    \{l(\Sec{E^*}),l(\Sec{E^*})\} & \subseteq  l(\Sec{E^*})\\
    \{l(\Sec{E^*}),\epsilon^*(\Cin{M})\} & \subseteq  \epsilon^*(\Cin{M})\\
    \{\epsilon^*(\Cin{M}),\epsilon^*(\Cin{M})\} & = 0.
\end{align*}

A submanifold $i:C\hookrightarrow P$ is called a \textbf{coisotropic submanifold} of the Poisson manifold $(P,\pi)$ if $T_xC\subset (T_xP,\pi_x)$ is a coisotropic subspace for all $x\in C$, that is
\begin{equation*}
    \pi_x^\sharp (T_xC^0)\subset T_xC.
\end{equation*}
Denoting by $I_C:=\Ker{i^*}\subset \Cin{P}$ the vanishing ideal of the submanifold $C$ we formulate the following result giving alternative equivalent definitions of coisotropic submanifolds.
\begin{prop}
Let $i:C\hookrightarrow P$ be a closed submanifold of the Poisson manifold $(P,\pi)$, then the following are equivalent:
\begin{enumerate}
    \item $C$ is coisotropic,
    \item $I_C$ is a coisotrope in the Poisson algebra $(\Cin{P},\cdot,\{,\}_\pi)$,
    \item for all $c\in I_C$ the Hamiltonian vector field $X_c$ is tangent to $C$: $X_c|_C\in \Sec{TC}$.
\end{enumerate}
\end{prop}
Let two Poisson manifolds $(P_1,\pi_1)$ and $(P_2,\pi_2)$, we define the \textbf{product Poisson manifold} simply from the Cartesian product as $(P_1\times P_2,\pi_1+\pi_2)$, where the sum of bivectors is defined via the canonical isomorphism $\Tan(P_1\times P_2)\cong \Proj_1^*\Tan P_1\oplus \Proj_2^*\Tan P_2$. Equivalently, the product Poisson manifold structure on $P_1\times P_2$ is the Poisson bracket $\{,\}_{12}$ determined uniquely by the condition that the canonical projections
\begin{equation*}
\begin{tikzcd}
P_1 & P_1\times P_2 \arrow[l,"\Proj_1"']\arrow[r,"\Proj_2"] & P_2.
\end{tikzcd}
\end{equation*}
are Poisson maps, that is
\begin{equation*}
    \Proj_1^*\{f_1,g_1\}_1=\{\Proj_1^*f_1,\Proj_1^*g_1\}_{12} \qquad \Proj_1^*\{f_2,g_2\}_2=\{\Proj_2^*f_2,\Proj_2^*g_2\}_{12}
\end{equation*}
for all $f_1,g_1\in\Cin{P_1}$ and $f_2,g_2\in\Cin{P_2}$. For a given Poisson manifold $(P,\pi)$ let us denote the \textbf{opposite} Poisson manifold by $\overline{P}:=(P,-\pi)$. The proposition below gives a characterization of Poisson maps as \textbf{coisotropic relations}, defined as coisotropic submanifolds of the product Poisson manifold.
\begin{prop}
Let two Poisson manifolds $(P_1,\pi_1)$ and $(P_2,\pi_2)$, and a smooth map $\varphi:P_1 \to P_2$, then $\varphi$ is a Poisson map iff its graph
\begin{equation*}
\normalfont
    \Grph{\varphi}\subseteq P_1\times \overline{P_2}
\end{equation*}
is a coisotropic submanifold in the product Poisson manifold.
\end{prop}

It follows from our discussion above that any coisotropic submanifold $C$ carries an involutive tangent distribution given by the images of the Hamiltonian vector fields of the (locally generating) elements of the vanishing ideal $I_C$. Let us denote this distribution by $X_{I_C}\subset TC$ and note that $[X_{I_C},X_{I_C}]\subset X_{I_C}$ from the fact that $I_C$ is a coisotrope. This means that there will be a (singular) foliation in $C$, denoted by $\mathcal{X}_C$, integrating the tangent distribution of Hamiltonian vector fields of vanishing functions. The following result gives the standard construction of what is known as \textbf{coisotropic reduction}.

\begin{prop}
Let $(P,\pi)$ be a Poisson manifold and $i:C\hookrightarrow P$ a closed coisotropic submanifold. Assume that the tangent distribution $X_{I_C}$ integrates to a regular foliation with smooth leaf space $C/\mathcal{X}_C$ in such a way that there is a surjective submersion
\begin{equation*}
    q:C\to P':=C/\mathcal{X}_C.
\end{equation*}
Then, the manifold $P'$ inherits a Poisson bracket on functions $(\Cin{P'},\{,\}')$ uniquely determined by the condition
\begin{equation*}
    q^*\{f,g\}'=i^*\{F,G\} \quad \forall f,g\in\Cin{P'}
\end{equation*}
and for all $F,G\in\Cin{P}$ leaf-wise constant extensions, i.e functions of the ambient Poisson manifold satisfying
\begin{equation*}
    q^*f=i^*F, \quad q^*g=i^*G.
\end{equation*}
\end{prop}
A particular case of coistropic reduction is given by a \textbf{Hamiltonian group action}: this is the situation when a Lie group acts on a Poisson manifold $G\Acts (P,\pi)$ via Poisson maps and with infinitesimal generators $\psi:\mathfrak{g}\to \Sec{\Tan P}$ given by Hamiltonian vector fields via a (co)moment map $\overline{\mu}:\mathfrak{g}\to \Cin{M}$ with the defininfg conditions:
\begin{equation*}
    \psi(\xi)=X_{\overline{\mu}(\xi)}, \qquad \{\overline{\mu}(\xi),\overline{\mu}(\zeta)\}_\pi=\overline{\mu}([\xi,\zeta]) \qquad \forall \xi,\zeta\in \mathfrak{g}.
\end{equation*}
The map functional dual map $\mu:P\to \mathfrak{g}^*$ is called the \textbf{moment map}. When $0\in\mathfrak{g}^*$ is a regular value, $\mu^{-1}(0)\subseteq P$ is easily checked to be a coisotropic submanifold. If the Poisson action restricts to a free and proper action on $\mu^{-1}(0)$ then the orbit space $\mu^{-1}(0)/G$ is the coisotropic reduction of $(P,\pi)$. This is called \textbf{Hamiltonian reduction}.\newline

A \textbf{symplectic manifold} is a Poisson manifold with non-degenerate bivector, thus equivalently determined by a non-degenerate 2-form, $(S,\omega$). The integrability condition for the non-degenerate bivector is equivalent to $d\omega=0$. All the notions introduced above for Poisson manifolds apply to symplectic manifolds but the non-degeneracy requirement makes symplectic manifolds a great deal more rigid than general Poisson manifolds. Firstly, it is a direct consequence of elementary linear algebra that symplectic manifolds must be even dimensional and that \textbf{symplectomorphisms}, smooth maps preserving the 2-forms, are necessarily (local) diffeomorphisms. The musical map now has an inverse
\begin{equation*}
    \omega^\flat :\Tan S\to \Cot S
\end{equation*}
and the Hamiltonian vector fields satisfy the following defining identities
\begin{equation*}
    \omega^\flat(X_f)=df, \qquad \omega(X_f,X_g)=\{f,g\}_\omega \qquad \forall f,g\in\Cin{S}
\end{equation*}
In addition to coisotropic submanifolds, \textbf{isotropic submanifolds} can now be defined as submanifolds with vanishing restriction of the 2-form. Of particular relevance for our discussion in this paper are \textbf{Lagrangian submanifolds} which are submanifolds that are both isotropic and coisotropic, or equivalently, coisotropic submanifolds of dimension half of that of the ambient symplectic manifold. The notions of coisotropic reduction and Hamiltonian reduction apply verbatim to the case of symplectic manifolds.\newline

In light of our discussion in section \ref{CatLineB}, Jacobi geometry is naturally motivated as the unit-free analogue to Poisson geometry. A \textbf{Jacobi manifold} or \textbf{Jacobi structure} is defined as a line bundle $\lambda:L\to M$, indeed an object in the category $\Line_\Man$, whose space of sections carries a local Lie algebra structure $(\Sec{L},\{,\})$. A local Lie bracket is such that each of its arguments acts as a differential operator of degree at most $1$, therefore, since for a line bundle we have $\Diff_1(L)=\Dr{L}$, the locality condition is tantamount to the adjoint map of the Lie bracket mapping into derivations
\begin{equation*}
    \text{ad}_{\{,\}}:\Sec{L}\to \Dr{L}.
\end{equation*}
Note that this is entirely analogous to the Leibniz condition of a Poisson bracket where the smooth functions are replaced by the sections of the line bundle and vector fields are replaced by derivations. Since the symbol of any derivation defines a vector field, the analogue of the Hamiltonian map for a Jacobi manifold becomes a pair of $\Real$-linear maps
\begin{align*}
\Delta: \Sec{L} & \to \Dr{L}\\
u & \mapsto \Delta_u:=\{u,-\},
\end{align*}
\begin{align*}
X: \Sec{L} & \to \Sec{\Tan M}\\
u & \mapsto X_u:=\delta_*(\Delta_u).
\end{align*}
For a section $u\in\Sec{L}$, $\Delta_u$ is called its \textbf{Hamiltonian derivation} and $X_u$ its \textbf{Hamiltonian vector field}. As for any derivation of a line bundle, the Hamiltonian vector field map is uniquely determined by the following Leibniz identity
\begin{equation*}
    \Delta_u(f\cdot v)=f\cdot\Delta_u(v)+X_u[f]\cdot v,
\end{equation*}
which must hold for all sections $u,v\in\Sec{L}$ and smooth functions $f\in\Cin{M}$. The following proposition establishes these as the Jacobi analogue of the Hamiltonian map.
\begin{prop}[Hamiltonian Maps of a Jacobi Structure]\label{HamMapsJac}
The Hamiltonian maps $\Delta$ and $X$ are Lie algebra morphisms, i.e.
\begin{equation*}
    \Delta_{\{u,v\}}=[\Delta_u,\Delta_v] \qquad X_{\{u,v\}}=[X_u,X_v]
\end{equation*}
for all $u,v\in\Sec{L}$.
\end{prop}
\begin{proof}
The Lie algebra morphism condition for $\Delta$ follows directly from the Jacobi identity of the bracket. To show that $X$ is a Lie algebra morphism we consider the triple bracket $\{u,\{v,f\cdot w\}\}$. We can expand the expression by using the Jacobi identity of the Lie bracket first and then applying the Leibniz property repeatedly to obtain
\begin{equation*}
   \{u,\{v,f\cdot w\}\}=f\cdot\{u,\{v,w\}\}+X_{u}[f]\cdot \{v,w\} + X_{v}[f]\cdot \{u,w\} + X_{\{u,v\}}[f]\cdot w + X_{v}X_{u}[f]\cdot w,
\end{equation*}
or we can expand only using the Leibinz property obtaining
\begin{equation*}
   \{u,\{v,f\cdot w\}\}=f\cdot\{u,\{v,w\}\}+X_{u}[f]\cdot \{v,w\} + X_{v}[f]\cdot \{u,w\} + X_{u}X_{v}[f]\cdot w.
\end{equation*}
Since both expressions must agree, it follows that
\begin{equation*}
   X_{\{u,v\}}[f]\cdot w + X_{v}X_{u}[f]\cdot w=X_{u}X_{v}[f]\cdot w
\end{equation*}
which must hold for all sections $u,v,w\in\Sec{A}$ and functions $f\in\Cin{M}$, thus implying the desired Lie algebra homomorphism condition.
\end{proof}
The fact that the Lie bracket $\{,\}$ is a derivation on each argument allows us to write
\begin{equation*}
    \{u,v\}=\Tilde{\Lambda}(j^1u,j^1v) \quad \text{ with }\quad  \Tilde{\Lambda}\in \Sec{\wedge^2(\Jet^1L)^*\otimes L}.
\end{equation*}
We thus see that the bilinear form $\Tilde{\Lambda}$ is the analogue of the Poisson bivector and we appropriately call it the \textbf{Jacobi biderivation}. Noting that $(\Jet^1L)^*\otimes L=((\Der L)^{*L})^{*L}\cong \Der L$, the biderivation induces a musical map
\begin{equation*}
    \Tilde{\Lambda}^\sharp:\Jet^1 L \to \Der L
\end{equation*}
which allows us to rewrite the Hamiltonian derivation of a section $u\in\Sec{L}$ as
\begin{equation*}
    \Delta_u = \Tilde{\Lambda}^\sharp (j^1u).
\end{equation*}
This musical map connects the, otherwise disconnected, symbol and Spencer short exact sequences of vector bundles in the following way
\begin{equation*}
\begin{tikzcd}
0 \arrow[r] &  \Real_M \arrow[r] & \Der L \arrow[r, "\delta"]  & \Tan M \arrow[r]  & 0 \\
0 &  L \arrow[l] & \Jet^1L  \arrow[u, "\Tilde{\Lambda}^\sharp"] \arrow[l]& \Tan^{*L} M \arrow[u,"\Lambda^\sharp"] \arrow[l,"i"] & 0 \arrow[l] 
\end{tikzcd}
\end{equation*}
where we have defined the bundle map $\Lambda^\sharp$ so that the diagram commutes. By construction, we see that this bundle map is indeed a musical map for the bilinear form defined as the fibre-wise pull-back of the Jacobi biderivation via the injective map of the Spencer sequence: $\Lambda=i^*\Tilde{\Lambda}\in\Sec{\wedge^2(\Tan^{*L} M)\otimes L}$. We call this form the \textbf{Jacobi L-bivector}. These bilinear forms should be interpreted as the symbols of the Hamiltonian maps when regarded as differential operators themselves since they indeed control the interaction of the Lie bracket with the $\Cin{M}$-module structure via the following identities
\begin{align*}
    \Delta_{f\cdot u}&=f\cdot \Delta_u+\Tilde{\Lambda}^\sharp(i(df\otimes u))\\
    X_{f\cdot u}&=f\cdot X_u+\Lambda^\sharp(df\otimes u).
\end{align*}
where $f\in\Cin{M}$ and $u,v\in\Sec{L}$.\newline

In some situations it will be convenient to have a characterization of a Jacobi structure in terms of the Lie bracket and the Hamiltonian maps seen as $\Real$-linear maps satisfying some compatibility conditions with the module structure of sections. The next proposition gives this characterization.
\begin{prop}[Extension by Symbol]\label{ExtBySymb}
Let $\lambda:L\to M$ be a line bundle and let $\Sigma\subseteq \Sec{L}$ be a subspace of spanning sections such that $\Sec{L}\cong \Cin{M}\cdot \Sigma$. The datum of a Jacobi structure on $L$, i.e. a local Lie algebra $(\Sec{L},\{,\})$, is equivalent to a triple $((\Sigma,[,]),X,\Lambda)$ where
\begin{itemize}
\normalfont
    \item $(\Sigma,[,])$ is a $\Real$-linear Lie bracket,
    \item $X:\Sec{L}\to \Sec{\Tan M}$ is a $\Real$-linear map and,
    \item $\Lambda\in\Sec{\wedge^2\Tan^{*L} M\otimes L}$ inducing a vector bundle morphism $\Lambda^\sharp:\Tan^{*L} M\to \Tan M$
\end{itemize}
satisfying the compatibility conditions
\begin{align*}
    [X_s,X_{s'}]&=X_{[s,s']}\\
    X_{f\cdot s}&=f\cdot X_s+\Lambda^\sharp(df\otimes s) \\
    [X_s,\Lambda^\sharp(df\otimes s')]&=\Lambda^\sharp(dX_{s}[f]\otimes s' + df\otimes [s,s'])
\end{align*}
for all $f\in\Cin{M}$ and $s,s'\in\Sigma$. The equivalence is realized by imposing
\begin{align*}
    \{s,s'\}&:=[s,s'] \\
    \{s,f\cdot s'\}&:=f\cdot \{s,s'\} + X_{s}[f]\cdot s'
\end{align*}
which uniquely determine a local Lie bracket $(\Sec{L},\{,\})$.
\end{prop}
\begin{proof}
Since $\Sigma$ is a subspace of spanning sections, a Jacobi bracket will be completely determined by specifying brackets of the form $\{f\cdot s,g\cdot s'\}$ for all $f,g\in\Cin{M}$ and $s,s'\in\Sigma$. It follows by direct computation that the Hamiltonian vector field $X$ and the L-bivector $\Lambda$ of a Jacobi structure satisfy the compatibility conditions in the proposition and, in virtue of these, the spanning brackets expand as
\begin{equation*}
    \{f\cdot s,g\cdot s'\}=fg\cdot\{s,s'\}+fX_s[s]\cdot s' - gX_{s'}[f]\cdot s +\Lambda^\sharp(df\otimes s)[g]\cdot s'.
\end{equation*}
Note that antisymmetry of $\Lambda$ implies that
\begin{equation*}
    \Lambda^\sharp(df\otimes s)[g]\cdot s'=-\Lambda^\sharp(dg\otimes s')[f]\cdot s
\end{equation*}
Now, by setting $\{s,s'\}=[s,s']$ it is clear that the triple $((\Sigma,[,]),X,\Lambda)$ determines an antisymmetric bracket $(\Sec{L},\{,\})$, thus we are only left with showing that it does indeed satisfy the Jacobi identity. Assuming the Jacobi identity for the $\Real$-linear bracket $[,]$, this is shown to be equivalent to the second compatibility condition
\begin{equation*}
    [X_u,\Lambda^\sharp(df\otimes v)]=\Lambda^\sharp(dX_u[f]\otimes v + df\otimes \{u,v\}).
\end{equation*}
The only non-trivial check corresponds to the triple bracket $\{w,\{f\cdot u,g\cdot v\}\}$ whose Jacobi identity is
\begin{equation*}
    \{w,\{f\cdot u,g\cdot v\}\}=\{\{w,f\cdot u\},g\cdot v\}+\{f\cdot u,\{w,g\cdot v\}\}.
\end{equation*}
We can expand both sides using the defining identities of $X$ and $\Lambda$ and, simplifying whenever possible with the Jacobi identity of the bracket, we obtain the following identity
\begin{align*}
& X_w[g]X_v[f]\cdot u - gX_{\{v,w\}}[f]\cdot u -X_w[f]X_u[g]\cdot v - fX_{\{w,u\}}[g]\cdot v = \\
& \Lambda^\sharp(df\otimes u)[g]\cdot \{w,v\} +X_w[\Lambda^\sharp(df\otimes u)[g]]\cdot v -[X_w,X_{f\cdot u}][g]\cdot v + [X_w,X_{g\cdot v}][f]\cdot u
\end{align*}
Then, further expanding by the defining identities of $X$ and crucially using the $\Cin{M}$-linearity of $\Lambda^\sharp$, after simplifications we get
\begin{equation*}
    X_w[\Lambda^\sharp(df\otimes u)[g]]\cdot v -\Lambda^\sharp(df\otimes \{w,u\})[g]\cdot v = \Lambda^\sharp(dX_w[f]\otimes u)[g]\cdot v - \Lambda^\sharp(dX_w[g]\otimes v)[f]\cdot u.
\end{equation*}
We can use the antisymmetry of $\Lambda$ to rewrite the last term and make the identity into $\Cin{M}$-linear combination of terms in $v$, thus obtaining the desired result
\begin{equation*}
    [X_w,\Lambda^\sharp(df\otimes u)][g]\cdot v = \Lambda^\sharp(dX_w[f]\otimes u + df\otimes \{w,u\})[g]\cdot v.
\end{equation*}
\end{proof}

Analogously to the linear Poisson structures defined on vector bundles we define the corresponding notion for Jacobi structures on L-vector bundles. Let $\epsilon:E^L\to M$ be a L-vector bundle and consider the pull-back line bundle over the total space $L_E:=\epsilon^*L$. We identify the fibre-wise constant sections as $\epsilon^*:\Sec{L}\hookrightarrow \Sec{L_E}$ and the fibre-wise linear sections as $l:\Sec{E^{*L}}\hookrightarrow \Sec{L_E}$, which form a submodule of spanning sections for $\Sec{L_E}$. A \textbf{fibre-wise linear Jacobi structure} is defined as a Jacobi structure $(\Sec{L_E},\{,\})$ that is completely determined by its restriction to spanning sections in the following way
\begin{align*}
    \{l(\Sec{E^{*L}}),l(\Sec{E^{*L}})\} & \subseteq  l(\Sec{E^{*L}})\\
    \{l(\Sec{E^{*L}}),\epsilon^*(\Sec{L}\} & \subseteq  \epsilon^*(\Sec{L})\\
    \{\epsilon^*(\Sec{L}),\epsilon^*(\Sec{L})\} & = 0.
\end{align*}

Let $i:S\hookrightarrow M$ be an embedded submanifold of a line bundle $\lambda:L\to M$ and let us denote by $L_S:=i^*L$ the induced line bundle over $S$ and by $\Gamma_S\subseteq \Sec{L}$ the submodule of vanishing sections. Note that at each point $x\in M$ the jet space $\Jet^1_x L$ carries a $L_x$-valued bilinear form given by the Jacobi biderivation $\Tilde{\Lambda}_x$ and so the notion of isotropic subspace of a jet space (one where the bilinear form restricts to zero) is well-defined point-wise. We can then define coisotropic submanifolds of Jacobi manifolds in an entirely analogous way to the coisotropic submanifolds of Poisson manifolds. We say that $S$ is \textbf{coisotropic} if the L-annihilator of its der bundle $(\Der L_S)^{0L}\subset \Jet^1L$ is an isotropic subbundle with respect to the bilinear form $\Tilde{\Lambda}$. Equivalently, $S$ is coisotropic if
\begin{equation*}
    \Tilde{\Lambda}^\sharp_x((\Der_xL_S)^{0L})\subseteq \Der_xL_S \qquad \forall x\in S.
\end{equation*}
The following proposition gives several equivalent characterizations of coisotropic submanifolds.
\begin{prop}[Coisotropic Submanifolds of a Jacobi Manifold]\label{CoisoJacMan}
Let $i:S\hookrightarrow M$ be an embedded submanifold of a Jacobi structure $(\Sec{L},\{,\})$, then the following conditions are equivalent
\begin{enumerate}
\normalfont
    \item $S$ is a coisotropic submanifold;
    \item $\Lambda^\sharp_x(\Tan_x S^0\otimes L_x)\subseteq \Tan_x S \quad \forall x\in S$;
    \item the vanishing sections $\Gamma_S$ form a Lie subalgebra, $\{\Gamma_S,\Gamma_S\}\subseteq \Gamma_S$;
    \item the Hamiltonian vector fields of vanishing sections are tangent to $S$, $X_s\in\Sec{\Tan S}$ for $s\in\Gamma_S$.
\end{enumerate}
\end{prop}
\begin{proof}
Equivalence between 1. and 2. follows from the definition of the bilinear forms, $\Lambda=i^*\Tilde{\Lambda}$, and the fact that
\begin{equation*}
    i(\alpha\otimes l)\in (\Der_xL_S)^{0L} \quad \Leftrightarrow \quad i(\alpha\otimes l)(a)=0 \quad \forall a\in \Der_xL_S,
\end{equation*}
which, since $i$ is the $L$-dual to the anchor $\delta$, is tantamount to demanding
\begin{equation*}
    \alpha(\delta(a))=0 \quad \forall a\in \Der_xL_S.
\end{equation*}
Recall that $\delta(\Der L_S)=\Tan S$, then 
\begin{equation*}
    i(\alpha\otimes l)\in (\Der_xL_S)^{0L} \quad \Leftrightarrow \quad \alpha\otimes l \in \Tan_x S^0\otimes L_x.
\end{equation*}
To show equivalence between 2. and 3. we can use a local trivialization to write $\Ker{i^*}=\Gamma_S=I_S\cdot \Sec{L}$, then it follows from the definition of the Jacobi L-bivector that
\begin{equation*}
    i^*\{f\cdot s,g\cdot s\}=\Lambda^\sharp(df\otimes s)[g]\cdot i^*s=0
\end{equation*}
for $f,g\in I_S$ and any section $s\in\Sec{L}$. Finally, equivalence between 3. and 4. follows by considering $f\in I_S$, $s\in\Sec{L}$, $s'\in\Gamma_S$ and writing
\begin{equation*}
    i^*\{s',f\cdot s\}=i^*f\cdot i^*\{s',s\} + i^*(X_{s'}[f])\cdot i^*s= i^*(X_{s'}[f])\cdot i^*s
\end{equation*}
thus implying $X_{s'}[f]\in I_S$, completing the equivalence.
\end{proof}

Let $(\Sec{L_1},\{,\}_1)$ and $(\Sec{L_2},\{,\}_2)$ be two Jacobi structures, a \textbf{Jacobi map} is a factor, i.e. morphism in the category of line bundles $\Line_\Man$, $B:L_1\to L_2$ such that its pull-back on sections $B^*:(\Sec{L_2},\{,\}_2)\to (\Sec{L_1},\{,\}_1)$ is a Lie algebra morphism. The next proposition shows that the categorical product in $\Line_\Man$ allows us to define a product of Jacobi manifolds.
\begin{prop}[Product of Jacobi Manifolds]\label{ProdJacMan}
Let Jacobi structures $(\Sec{L_1},\{,\}_1)$ and $(\Sec{L_2},\{,\}_2)$, then there exists a unique Jacobi structure in the line product $(\Sec{L_1\utimes L_2},\{,\}_{12})$ such that the canonical projection factors
\begin{equation*}
\begin{tikzcd}
L_1 & L_1\utimes L_2 \arrow[l,"P_1"']\arrow[r,"P_2"] & L_2.
\end{tikzcd}
\end{equation*}
are Jacobi maps. The Jacobi structure $(\Sec{L_1\utimes L_2},\{,\}_{12})$ is called the \textbf{product Jacobi structure}.
\end{prop}
\begin{proof}
Recall from the definition of line product as a pull-back bundle that $\Sec{\L_1\utimes L_2}=\Cin{M_1\dtimes M_2} \cdot p_1^*\Sec{L_1}\cong \Cin{M_1\dtimes M_2} \cdot p_2^*\Sec{L_2}$, therefore it will suffice to determine the bracket on pull-backs of sections of each factor and then use the extension by symbol formula for some appropriately chosen Hamiltonian vector field $X^{12}:\Sec{\L_1\utimes L_2}\to \Sec{\Tan M_1\dtimes M_2}$ and L-bivector musical map $\Lambda^{12}:\Sec{\Tan ^{*L} (M_1\dtimes M_2)\otimes L_1\utimes L_2}\to \Sec{\Tan (M_1\dtimes M_2)}$. The bracket on the line product is then defined as a $\Real$-linear bracket given by
\begin{equation*}
    \{P_1^*s_1,P_1^*s_1'\}_{12}:=P_1^*\{s_1,s_1'\}_1 \qquad \{P_2^*s_2,P_2^*s_2'\}_{12}:=P_2^*\{s_2,s_2'\}_2\qquad \{P_1^*s_1,P_2^*s_2\}_{12}:=0
\end{equation*}
for all $s_i,s_i'\in\Sec{L_i}$, $i=1,2$. It suffices to define the action of the Hamiltonian vector field and the L-bivector musical map locally on spanning functions of the base product as introduced in Proposition \ref{SpanFuncProd}. By direct computation we readily check that the following defining conditions
\begin{align*}
    X^{12}_{P_1^*s_1}[\tfrac{a}{b}]&=\tfrac{\{s_1,a\}_1}{b}   & X^{12}_{P_2^*s_2}[\tfrac{b}{a}]&=\tfrac{\{s_2,b\}_2}{a}\\
    \Lambda^{12}(d\tfrac{a}{b}\otimes P_1^*s_1)[\tfrac{a'}{b'}]&=\tfrac{\{a,a'\}_1}{b}\tfrac{s_1}{b'} &\Lambda^{12}(d\tfrac{b}{a}\otimes P_2^*s_2)[\tfrac{b'}{a'}]&=\tfrac{\{b,b'\}_2}{a}\tfrac{s_2}{a'}
\end{align*}
for all $s_i\in\Sec{L_i}$, $a,a'\in\Sec{L_1^\times}$ and $b,b'\in\Sec{L_2^\times}$, determine $X^{12}$ and $\Lambda^{12}$ so that the compatibility conditions of Proposition \ref{ExtBySymb} are satisfied and thus define a Jacobi structure on the line product $L_1\utimes L_2$. The above formulas may appear asymmetric for the local non-vanishing sections $a\in\Sec{L_1^\times}$ and $b\in\Sec{L_2^\times}$, however recall that our definition of local ratio functions is such that $\tfrac{a}{b}\tfrac{b}{a}=1$ and so, by construction, we can relate both sides of the definition by the relations
\begin{equation*}
    d\tfrac{a}{b}=-(\tfrac{a}{b})^2d\tfrac{b}{a} \qquad P_1^*a=\tfrac{a}{b}P_2^*b.
\end{equation*}
\end{proof}

Consider a factor between line bundles
\begin{equation*}
\begin{tikzcd}
L_1 \arrow[r, "B"] \arrow[d, "\lambda_1"'] & L_2 \arrow[d, "\lambda_2"] \\
M_1 \arrow[r, "b"'] & M_2.
\end{tikzcd}
\end{equation*}
The construction of the base product of two line bundles $M_1\dtimes M_2$ as the space of fibre isomorphisms allows to identify the factor $B$ with a submanifold that we may regard as the line bundle analogue of a graph:
\begin{equation*}
    \LGrph{B}:=\{C_{x_1x_2}\in M_1\dtimes M_2|\quad x_2=b(x_1), \quad C_{x_1x_2}=B_{x_1}\},
\end{equation*}
we call this set the \textbf{L-graph} of the factor $B$. Given a Jacobi structure $(\Sec{L},\{,\})$ we define the \textbf{opposite Jacobi structure} as $\overline{L}:=(\Sec{L},-\{,\})$. The proposition below shows that we can regard Jacobi maps as a particular case of \textbf{coistropic relations}, which are defined, in general, as coisotropic submanifolds of the product Jacobi manifold $L_1\utimes \overline{L}_2$ and are regarded as the general notion of morphism in the category of Jacobi manifolds.
\begin{prop}[Jacobi Maps as Coisotropic Relations]\label{JacMapCoisoRel}
Let two Jacobi structures $(\Sec{L_{M_1}},\{,\}_1)$ and $(\Sec{L_{M_2}},\{,\}_2)$ and a factor $B:L_{M_1}\to L_{M_2}$, then $B$ is a Jacobi map iff its L-graph
\begin{equation*}
\normalfont
    \LGrph{B}\subseteq M_1\dtimes M_2
\end{equation*}
is a coisotropic submanifold of the product Jacobi structure $L_{M_1}\utimes \overline{L}_{M_2}$.
\end{prop}
\begin{proof}
We will use characterization 3. in Propostion \ref{CoisoJacMan} to identify a coisotropic submanifold with a Lie subalgebra of the Jacobi bracket. Note that, by construction of $\LGrph{B}$, the vanishing sections $\Gamma_{\LGrph{B}}\subseteq \Sec{L_1\utimes \overline{L}_2}$ are generated by those of the form
\begin{equation*}
    P_1^*B^*s-P_2^*s\qquad \forall s\in\Sec{L_2}.
\end{equation*}
We can directly compute the brackets using the defining relations of the product Jacobi structure
\begin{equation*}
    \{P_1^*B^*s-P_2^*s,P_1^*B^*s'-P_2^*s'\}_{12}=P_1^*\{B^*s,B^*s'\}_1-P_2^*\{s,s'\}_2
\end{equation*}
which must hold for all $s.s'\in\Sec{L_2}$. Then it follows that
\begin{equation*}
    \{P_1^*B^*s-P_2^*s,P_1^*B^*s'-P_2^*s'\}_{12}\in \Gamma_{\LGrph{B}}\quad \Leftrightarrow \quad \{B^*s,B^*s'\}_1 = B^*\{s,s'\}_2
\end{equation*}
as desired.
\end{proof}

In direct analogy with reduction schemes of Poisson manifolds we find the notion of \textbf{Jacobi reduction}. Let a line bundle $\lambda:L\to M$ with a Jacobi structure $(\Sec{L},\{,\})$ and consider a submanifold $i:C\hookrightarrow M$ with its corresponding embedding factor $\iota: L_C\hookrightarrow L$. For another line bundle $\lambda':L'\to M'$, assume there exists a submersion factor $\pi:L_S\twoheadrightarrow L'$ with corresponding surjective submersion $p:C\twoheadrightarrow M'$. Then, we say that a Jacobi structure $(\Sec{L'},\{,\}')$ \textbf{Jacobi reduces} from the Jacobi structure $(\Sec{L},\{\,\})$ via $\pi:L_C\twoheadrightarrow L'$ when for all pairs of sections $s_1,s_2\in\Sec{L'}$ the identity
\begin{equation*}
    \pi^*\{s_1,s_2\}'=\iota^*\{S_1,S_2\}
\end{equation*}
holds for all choices of extensions $S_1,S_2$, i.e all choices of sections $S_1,S_2\in\Sec{L}$ satisfying 
\begin{equation*}
    \pi^*s_1=\iota^*S_1 \qquad \pi^*s_2=\iota^*S_2.
\end{equation*}
As a particular case of Jacobi reduction we find \textbf{Hamiltonian Jacobi reduction}: this is the situation when the reduction data $\pi:L_C\twoheadrightarrow L'$ is entirely determined by the presence of a Lie group acting via Jacobi maps $G\Acts L$, with infinitesimal action $\Psi:\mathfrak{g}\to \Dr{L}$, and a (co)moment map $\overline{\mu}:\mathfrak{g}\to \Sec{L}$ satisfying the defining conditions
\begin{equation*}
    \Psi(\xi)=\Delta_{\overline{\mu}(\xi)},\qquad \{\overline{\mu}(\xi),\overline{\mu}(\zeta)\}=\overline{\mu}([\xi,\zeta]) \qquad \forall \xi,\zeta\in\mathfrak{g}.
\end{equation*}

A Jacobi structure $(\Sec{L}, \{,\})$ is called \textbf{non-degenerate} when its Jacobi biderivation $\Tilde{\Lambda}$ induces a musical isomorphism of L-vector bundles
\begin{equation*}
        \begin{tikzcd}
        \Jet^1 L\arrow[r, "\Tilde{\Lambda}^\sharp",yshift=0.7ex] & \Der L \arrow[l,"\Tilde{\Lambda}^\flat",yshift=-0.7ex]
        \end{tikzcd}
\end{equation*}
Recall that symplectic manifolds were shown to be non-degenerate Poisson manifolds, it turns out that there is an entirely analogous connection between contact manifolds and non-degenerate Jacobi manifolds. In what follows we make this claim precise. Let us first give the standard definition of contact manifold as a pair $(M,\text{H}\subset \Tan M)$ where $\text{H}$ is a maximally non-integrable hyperplane distribution, $\text{rk}(\text{H})=\dimm M-1$. Let us clarify what we mean by maximally non-integrable by noting first that a hyperplane distribution on the tangent bundle is equivalent to the datum of a (generically non-trivial) line bundle $\lambda:L\to M$ and a non-vanishing $L$-valued 1-form $\theta:\Tan M\to L$. The equivalence is realized by setting
\begin{equation*}
    \text{H}=\Ker{\theta}, \quad \text{ which then gives } \quad \Tan M/\text{H}\cong L.
\end{equation*}
Let us denote the $\Cin{M}$-submodule of vector fields tangent to the hyperplane distribution by $\Sec{\text{H}}$. We can define the following antisymmetric map for vector fields tangent to the hyperplane distribution
\begin{align*}
\omega: \Sec{\wedge^2 \text{H}} & \to \Sec{L}\\
(X,Y) & \mapsto \theta([X,Y]).
\end{align*}
The kernel of this map clearly measures the degree to which $\text{H}$ is integrable as a tangent distribution noting that, in particular, when $[\Sec{\text{H}},\Sec{\text{H}}]\subseteq\Sec{\text{H}}$ the map $\omega$ is identically zero. It follows by construction that $\omega$ is in fact $\Cin{M}$-bilinear and thus defines a bilinear form $\omega_\text{H}:\wedge^2 \text{H}\to L$ called the \textbf{curvature form} of the hyperplane distribution $\text{H}$. A hyperplane distribution $\text{H}$ is called \textbf{maximally non-integrable} when its curvature form $\omega_\text{H}$ is non-degenerate, i.e. when the musical map $\omega^\flat_\text{H}:\text{H}\to \text{H}^{*L}$ has vanishing kernel. Such a hyperplane distribution $\text{H}\subset \Tan M$ is called a \textbf{contact structure} on $M$ and we refer to the pair $(M,\text{H})$ as a \textbf{contact manifold}. We mention in passing that only requiring the kernel of $\omega_\text{H}^\flat$ to be regular, thus giving an integrable tangent distribution, is the definition of $(M,\text{H})$ as a \textbf{precontact manifold}. Simple dimension counting applied to any tangent space of a contact manifold reveals that a manifold supporting a contact hyperplane is necessarily odd-dimensional. A \textbf{contactomorphism}, defined as a smooth map $\varphi:(M_1,\text{H}_1)\to(M_2,\text{H}_2)$ whose tangent maps the hyperplane distributions isomorphically $\Tan \varphi (\text{H}_1)=\text{H}_2$, is necessarily a diffeomorphism from the fact that hyperplane distributions have codimension $1$ everywhere. This condition is equivalent to the tangent map $\Tan \varphi$ inducing a well-defined morphism of line bundles $\Phi:L_1\cong \Tan M_1/\text{H}_1\to L_2\cong \Tan M_2/\text{H}_2$.

\begin{prop}[Non-Degenerate Jacobi Manifolds as Contact Manifolds]\label{JacManContMan}
Let $M$ be any smooth manifold, then the datum of a contact structure $(M,\text{H}\subset \Tan M)$ is equivalent to the datum of a non-degenerate Jacobi structure $(\Sec{L_M},\{,\})$.
\end{prop}
\begin{proof}
Let us assume $(M,\text{H}\subset \Tan M)$ is a contact structure first and show that it defines a Hamiltonian vector field and L-bivector, which, in virtue of Proposition \ref{ExtBySymb}, give a Jacobi structure that is non-degenerate. The curvature form is non-degenerate and thus has an inverse $\omega^\sharp:\text{H}^{*L}\to \text{H}\subset \Tan M$. Define $L:=\Tan M/\text{H}$ with canonical projection defining the contact form $\theta\in\Sec{\Tan^{*L} M}$, which in turn gives an $\Real$-linear map $X:\Sec{L}\to\Sec{\Tan M}$ defined by the condition $\theta(X_u)=u$. Let us define the L-bivector musical map
\begin{equation*}
    \Lambda^\sharp(df\otimes u):=\omega^\sharp(df|_\text{H}\otimes u).
\end{equation*}
It follows from the definition of curvature form and the Jacobi identity of Lie bracket of vector fields that $X$ and $\Lambda$ so defined satisfy the extension by symbol conditions. The point-wise decomposition of tangent spaces as $T_xM\cong \text{H}\oplus \Real\cdot X_u(x)$ for some locally non-vanishing section $u$ then shows that the Jacobi biderivation associated to this Jacobi structure is indeed non-degenerate. Conversely, let us assume $\lambda:L\to M$ is a line bundle whose sections carry a non-degenerate Jacobi structure $(\Sec{L},\{,\})$ with L-bivector musical map $\Lambda^\sharp:\Tan^{*L}M\to \Tan M$. A point-wise dimension count shows that $\text{H}=\Lambda^\sharp(\Tan^{*L}M)\subset \Tan M$ is a hyperplane distribution and by direct computation we show that it is, in fact, contact and precisely the converse construction to the previous case.
\end{proof}

We conclude this section by giving a brief description of Jacobi manifolds with trivial line bundle. It follows from the results presented at the end of section \ref{CatLineB} for a trivial line bundle $\Real_M$ that a Jacobi structure corresponds to a Lie bracket on functions $(\Cin{M},\{,\})$ defined by a bivector field $\pi\in\Sec{\wedge^2 \Tan M}$ and a vector field $R\in\Sec{\Tan M}$ in the following way:
\begin{equation*}
    \{f,g\}=\pi(df,dg)+fR[g]-gR[f]
\end{equation*}
for all $f,g\in\Cin{M}$. The Jacobi identity of this bracket is tantamount to the following integrability condition
\begin{equation*}
    [R,\pi]=0\qquad [\pi,\pi]+2R\wedge\pi =0,
\end{equation*}
where $[,]$ denotes the Schouten bracket of multivector fields. When these identities hold $(\pi,R)$ is called a \textbf{Jacobi pair} and $(M,\pi,R)$ is known as a \textbf{Lichnerowicz manifold}, which is indeed simply a Jacobi manifold with trivialized line bundle. Recall that for a trivial line bundle we have the isomorphism $\Der \Real_M\cong \Tan M\oplus \Real_M$, then the Hamiltonian derivation and Hamiltonian vector field of a function $f\in\Cin{M}$ take the following concrete form
\begin{equation*}
    \Delta_f= \pi^\sharp(df)+fR\oplus -R[f] \qquad X_f=\pi^\sharp(df)+fR.
\end{equation*}
A \textbf{Lichnerowicz map} is a smooth map between Lichnerowicz manifolds whose pull-back on functions is a Lie algebra morphisms for the Jacobi brackets, a condition, in fact, equivalent to the bivector and vector fields being related by the tangent map of the smooth map. Lichnerowicz maps are special cases of Jacobi maps of trivial line bundles where the conformal factor is set to the constant function $1$.\newline

Suppose that the Jacobi structure of a Lichnerowicz manifold $(M,\pi,R)$ is non-degenerate so that the manifold is endowed with a contact structure $(M,\text{H})$. Since the line bundle is now globally trivial, the hyperplane distribution will be given by the kernel of some non-vanishing ordinary $1$-form 
\begin{equation*}
    \text{H}=\Ker{\theta}\qquad \theta\in\Omega^1(M)
\end{equation*}
called the \textbf{contact form}. In this case $(M,\theta)$ is called an \textbf{exact contact manifold}. It follows by construction that the curvature form of the hyperplane distribution is given by the restriction of the exterior derivative of the contact form
\begin{equation*}
    \omega_\text{H}=-d\theta|_\text{H}=:\omega|_\text{H},
\end{equation*}
and the fact that it is non-degenerate on $\text{H}$ is tantamount to the bilinear form
\begin{equation*}
    \eta:=\theta \otimes \theta + \omega \in \Sec{\otimes^2\Cot M}
\end{equation*}
being non-degenerate, thus giving a musical isomorphism
\begin{equation*}
        \begin{tikzcd}
        \Tan M \arrow[r, "\eta^\flat",yshift=0.7ex] & \Cot M \arrow[l,"\eta^\sharp",yshift=-0.7ex]
        \end{tikzcd}.
\end{equation*}
It then follows that the datum of an exact contact structure $(M,\theta)$ is equivalent to the datum of a non-degenerate Lichnerowicz manifold $(M,\pi,R)$, this is easily seen from the fact that the presence of a musical isomorphism allows for the natural assignment:
\begin{equation*}
    \theta = \eta^\flat R \qquad \omega=\eta^{\flat\flat}\pi
\end{equation*}
where
\begin{equation*}
    \eta^{\flat\flat}\pi(v,w):=\pi(\eta^\flat v, \eta^\flat w) \qquad \forall v,w\in\Tan M.
\end{equation*}

As a closing remark for this section, note that a Lichnerowicz manifold $(M,\pi,R)$ with $R=0$ makes $(M,\pi)$ precisely into a Poisson manifold. Thus we can regard Poisson manifolds as a special subclass of Jacobi manifolds defined on trivial line bundles and characterised by the fact that the Hamiltonian vector field of a function only depends on the differential of the function but not on its point-wise values. It is now a simple check to recover all the notions introduced for Poisson manifolds at the beginning of this section as their Jacobi analogues by imposing the two aforementioned conditions. This now retroactively justifies our choice of common terminology for both Poisson and Jacobi manifolds.

\section{A Mathematical Model for Measurands} \label{ModMes}

With the goal of developing a mathematical formalism that allows for a unit-free formulation of physics, we set out to axiomatize the fundamental notion of ``physical quantity''. Since our proposal will not be based in any particular exiting definition of physical quantity found in the mathematical physics literature, we will use the less common term ``measurand'' to refer to this concept. As it will become apparent in the sections below, this choice of terminology is motivated from our emphasis on the property that is to be measured, before any choice of unit is necessary, rather than the measurement outcome itself, which carries a magnitude and a unit.  The common usage of this term in metrology justifies our choice. Indeed, in the International Vocabulary of Metrology \cite{bimp2012metrology} one finds the following definition:\newline

\emph{2.3 (2.6)\textbf{Measurand}: quantity intended to be measured. NOTE 1 The specification of a measurand requires knowledge of the kind of quantity, description of the state of the phenomenon, body, or substance carrying the quantity, including any relevant component, and the chemical entities involved.}

\subsection{Measuring Canes and Cobblestones} \label{CanCob}

In order to introduce the basic notions of metrology, let us depict a scenario of very primitive science that, however simple and mundane as it may seem, illustrates quite clearly the main formal properties of physical quantities and the process of measurement. Suppose that one is outside in the countryside wandering near a riverbed where plenty of canes and flat cobblestones can be found. These objects, canes and flat cobblestones, are the matter of scientific inquiry in this simple scenario as one sets out to develop an elementary theory of length, area and weight. Let us denote a cane as $C\in \mathcal{C}$ and a flat cobblestone as $F\in\mathcal{F}$, where the calligraphic letters denote the somewhat abstract set of all existing objects within reach. The collection of all these objects $\mathcal{D}:=\mathcal{C}\cup \mathcal{F}$ will be called the \textbf{domain of observation} or simply the \textbf{domain}. One is interested in quantitatively comparable characteristics of objects within the domain of observation; these shall be called \textbf{measurands}. What we mean by ``quantitatively comparable'' here is that one should be able to give a numerical correspondence to a practical construction involving a pair of objects in the domain of observation. The examples of measurands below will make this notion precise. For these we shall only assume basic manipulations that can be practically made with canes and cobblestones together with the inherent human ability to count and keep a consistent memory record.
\begin{itemize}
    \item \textbf{Length.} Take two canes $C_1,C_2\in\mathcal{C}$; lay them on the ground side-by-side with near ends matching; by marking where far ends fall on the ground, move the near end of the shorter cane to the marked position; repeat this process until the far end of the shorter cane exceeds the far end of the longer cane or the two ends meet\footnote{The ``meet'' condition essentially amounts to the accuracy of the measurement process being described here. Assuming that there is a minimum accuracy is not only more realistic but is, in fact, operationally necessary for the measurement algorithm described to terminate in a finite number of steps. We will not deal with the notion of measurement uncertainty in our discussion and simply note that a more realistic numerical output will be two (instead of one) rational numbers corresponding to the lower and upper bounds of the ``meet'' condition, essentially given in this example by what the naked human eye can discern and the shape of the canes' ends.}; if the ends didn't meet, move the longer cane to its far end; repeat until the far ends meet; count the number of markings made for far ends of each cane $n_1,n_2\in\mathbb{N}$ after the process has stopped. The length comparison of cane $C_1$ to $C_2$ is given by the rational number:
    \begin{equation*}
        l(C_1,C_2):=\frac{n_1}{n_2}\in\Rat.
    \end{equation*}
    The property that makes this construction into the desired ``quantitative comparison'' is the fact that its transitivity is precisely modelled by the multiplicative structure of $\Rat$. To see this, consider a third cane $C_3\in\mathcal{C}$ and perform the length comparisons with $C_1$ and $C_2$ to verify the following experimental fact:
    \begin{equation*} \label{2oo3}
        l(C_1,C_2)\cdot l(C_2,C_3)\cdot l(C_3,C_1)=1. \tag{2\textbackslash 3}
    \end{equation*}
    We will refer to this transitivity condition as the \textbf{2-out-of-3 property}. In fact, it is this property what guarantees the consistency of  the formal assignment of a fraction (seen as an equivalence class of solutions to the basic multiplicative equation of integers) to the pair of naturals resulting from the length comparison algorithm described above. We can symbolically summarize the length comparison procedure as an assignment of the form $l:\mathcal{C}\times \mathcal{C}\to \Rat$ satisfying the 2-out-of-3 property.
    \item \textbf{Area.} Via a procedure that now takes flat cobblestones and compares their surface area by superimposing them in a similar fashion to the algorithm of the comparison of lengths, we characterize the comparison of area via a map $a:\mathcal{F}\times \mathcal{F}\to \Rat$ satisfying the 2-out-of-3 property.
    \item \textbf{Weight.} With the simple materials at hand one can construct a very primitive model of scale that will allow for the balancing of pairs of objects on an arm of varying span. An entirely analogue construction of quantitative comparison thus follows but now for any object in our observation domain $\mathcal{D}=\mathcal{C}\cup \mathcal{F}$. We shall call this comparison of weight and denote it by $w:\mathcal{D}\times \mathcal{D}\to \Rat$.
\end{itemize}
For the reminder of this section we make a few remarks about these concrete cases of Length, Area and Weight; nevertheless, the formal constructions that follow shall illustrate general points about arbitrary measurands.\newline

The measurand ``length of canes'' is conceptually thought of as the unifying property of all canes (or any straight object, in general) that allows for the operational system described above, including subdomain of applicability, algorithmic description of the procedure and numerical assignment. In particular, we see that a measurand is inseparably tied to the concrete class of physical objects on which it is defined. If we fix an object to set it as a reference or standard, say we take a cane $C_0\in\mathcal{C}$, we arrive at the notion of \textbf{unit of measurement}: all canes are assigned a rational number given by their length comparison with the unit $C_0$. In other words, a choice of a unit $C_0\in\mathcal{C}$ gives a map defined by
\begin{align*}
l_{C_0}: \mathcal{C} & \to \Rat\\
C & \mapsto l(C_0,C).
\end{align*}
Simply having a map of this form is not particularly useful but the 2-out-of-3 property of $l$ (\ref{2oo3}) allows to identify each arbitrary comparison of length with a quotient of the rational values of the map. More explicitly, having fixed a unit of length so we have a map  $l_{C_0}: \mathcal{C} \to \Rat$, the length comparison of any arbitrary pair of canes $C_1,C_2\in\mathcal{C}$ is given by
\begin{equation*}
        l(C_1,C_2)=\frac{l_{C_0}(C_1)}{l_{C_0}(C_2)}.
\end{equation*}
The length comparison $l(C_1,C_2)$ was defined from the experimental procedure with no reference whatsoever to the notion of unit, we thus expect the above equation to be independent of the choice of unit. Indeed, again by the 2-out-of-3 property (\ref{2oo3}), choosing a different unit $\overline{C}_0$ will introduce the factor $l(\overline{C}_0,C_0)$ in both numerator and denominator of the fraction, hence giving the same value for the length comparison. It is then easy to see that this leads to the usual notion of ``intrinsic length of a cane'' as the abstract entity that only becomes concrete once a unit is chosen. The language often employed to refer to measurement outcomes, that we will enforce throughout our discussion, is now fully justified: a cane $C$ is said to have length $l_{C_0}(C)\in\Rat$ in units of $C_0$.\newline

From the elementary measurement process described above for Length it is clear that given two canes $C_1,C_2\in\mathcal{C}$ we can make the far end of one meet with the near end of the other to virtually (i.e. to all metrological effects) create a new cane $C_1\parallel C_2$, we call it the \textbf{combined} cane or we say that $C_1\parallel C_2$ is the \textbf{combination} of the canes $C_1$ and $C_2$. It is then an experimental fact that for any other cane $C_3$ we have:
\begin{equation*}
    l(C_1\parallel C_2,C_3)=l(C_1,C_3)+l(C_2,C_3).
\end{equation*}
Choosing another pair of canes $\overline{C}_1$ and $\overline{C}_2$ whose comparisons with the previous pair, $l(C_1,\overline{C}_1)=q_1$ and $(C_2,\overline{C}_2)=q_2$, are known, it is also an experimental fact that
\begin{equation*} \label{comblin}
    l(\overline{C}_1\parallel \overline{C}_2,C_3)=q_1\cdot l(C_1,C_3)+q_2\cdot l(C_2,C_3),\tag{comb}
\end{equation*}
and thus we see that the combination construction satisfies a $\Rat$-linearity property. Under the choice of a unit $C_0$, this implies, in particular, an additivity property for the lengths of the canes in units of $C_0$:
\begin{equation*}
    l_{C_0}(C_1\parallel C_2)=l_{C_0}(C_1)+l_{C_0}(C_2).
\end{equation*}
It is not hard to see that any sensible definition of the combination construction for objects comparable by area or weight will satisfy a similar linearity property.\newline

Next we formalize the intuition that areas are products of lengths. Suppose that one finds a corner-shaped piece of wood or stone to act as a standard wedge $\vee$ (a form of rudimentary protractor that only measures one non-flat angle) so it is possible to consistently array two canes in two independent planar directions. Taking any two canes $C_1$ and $C_2$, with the use of the standard wedge one virtually (i.e. to all metrological effects) forms a planar shape directly comparable in area to flat cobblestones, we denote it by $C_1 \vee C_2$ and call it the \textbf{tile} of $C_1$ and $C_2$. We can thus compare the area of a generic flat cobblestone $F\in\mathcal{F}$ to the area of a tile via $a(F,C_1\vee C_2)\in\Rat$, and we call that value the area of $F$ in units of $C_1$by$C_2$. It is clear form the measurement procedures for length and area comparison that the following are experimental facts for any flat cobblestone $F\in\mathcal{F}$ and canes $C_1,C_2,C_3\in\mathcal{C}$:
\begin{align*}\label{tile1}
    a(F,C_1\vee C_2)&=a(F,C_2\vee C_1), \\ 
    a(F,C_1\vee (C_2\parallel C_3))&=a(F,C_1\vee C_2)+a(F,C_1\vee C_3). \tag{tile1}
\end{align*}
Consider now a different tile $\overline{C}_1\vee \overline{C}_2$ whose canes have length comparisons $l(C_1,\overline{C}_1)=q_1$ and $l(C_2,\overline{C}_2)=q_2$ with the canes of the tile $C_1\vee C_2$. It is then an experimental fact that, for any flat cobblestone $F\in\mathcal{F}$:
\begin{equation*}\label{tile2}
    a(F,C_1\vee C_2)=q_1\cdot q_2\cdot a(F,\overline{C}_1\vee \overline{C}_2), \tag{tile2}
\end{equation*}
which gives the usual identification of areas with products of lengths.\newline

Lastly, we briefly remark on a simple implementation of a theory of density enabled by the fact that standard scale described in Weight above allows to give quantitative comparisons between any objects. In particular, if we consider canes $\mathcal{C}$, they can be compared by length or by weight. By simultaneously comparing both characteristics it is possible to construct a new practical procedure that defines a new measurand, what we may call Density, and which is experimentally shown to be given by:
\begin{align*}
d: \mathcal{C}\times \mathcal{C} & \to \Rat\\
(C_1,C_2) & \mapsto \frac{w(C_1,C_2)}{l(C_1,C_2)}
\end{align*}
Such a measurand will be recognized as an entity that measures pairs of canes forming what we call a \textbf{ratio} by weight and length, denoted by $C_1|^w_lC_2$. To make this notion more concrete, we fix a unit of weight $C_0^w$ and a unit of length $C_0^l$, the density of an arbitrary cane $C$ will be specified by the quotient of the comparison in weight $w(C_0^w,C)$ by the comparison in length $l(C_0^l,C)$, then we say that this is the value of the density of $C$ measured in units of $C_0^w$per$C_0^l$. The upshot of the construction of Density is that it provides a natural example of a ``quantitative comparison of quantitative comparisons'', thus defining what we call a \textbf{derived measurand}. These play a fundamental role in practical science, as physical laws are typically formulated in terms of derived measurands. In this example, Density is the common characteristic of canes that allows for a prediction of a comparison of weights with the knowledge of a comparison of length or \emph{vice versa}. 

\subsection{Measurand Spaces as Potentils in the Category of Lines} \label{MeasPot}

In section \ref{CanCob} we explored the empirico-formal structure of some simple measurands in a fair amount of detail so as to motivate the purely mathematical definition that will be given in this section. Before we proceed, let us make two remarks about the use of rational numbers as the outcome of quantitative comparisons.\newline

Firstly, invoking what we could call the \textbf{Principle of Refinement}, we note that in real physics, particularly in the historical conventions of classical mechanics, one assumes that smaller units always exist so that, as far as the minimum accuracy of a measurement process allows, comparisons can be refined indefinitely. One possible way in which this principle could be implemented abstractly in our formal treatment of Length of section \ref{CanCob} is by assuming that the outcome of a length comparison is no longer a single rational value $l(C_1,C_2)\in\Rat$ but, in general, a sequence of rationals $\{l_i(C_1,C_2)\}_{i=1}^\infty$ with
\begin{equation*}
   \lim_{i,j\to\infty} |l_i(C_1,C_2)- l_j(C_1,C_2)|=0.
\end{equation*}
In other words, $\{l_i(C_1,C_2)\}_{i=1}^\infty$ is a Cauchy sequence in $\Rat$. Since from an operational point of view a length comparison needs to yield a single well-defined numerical value, the Cauchy sequence should be forced to be convergent to a single value that we would denote by $l(C_1,C_2)$. This is nothing but the usual construction of $\Real$ by completion of $\Rat$ as a metric space. Quantitative comparison outcomes will henceforth be considered real values and we shall update all the notions introduced in section \ref{CanCob} by replacing any instance of $\Rat$ by $\Real$. It is worth noting, however, that this change is only strictly necessary for our discussion in section \ref{DimHam} where we work within the formalism of differential geometry, which obviously makes extensive use of the formal construction of the real numbers.  Although all the mathematical structure necessary to formalize the definitions of this section is already present in $\Rat$, for the sake of consistency with the rest of the paper, we will only work over the field of reals $\Real$ from now on.\newline

Secondly, we have not given any physical (i.e. operationally constructive) meaning to the zero in the field of numbers, and therefore neither have we given any meaning to negative values as outcomes of quantitative comparisons. What this is really telling us is that only the affine structure of $\Rat$ (or $\Real$, as we will consider from now on) is relevant for a consistent definition of measurand. Although we could proceed by simply assuming the affine structure, we will keep a ``rooted'' approach, with the zero as a valid possible outcome of a measurement, and only explicitly use the linear structure of $\Rat$. Since all vector spaces are, in particular, affine spaces, this is an equivalent and convenient choice that will prove particularly useful in section \ref{DimHam} when we deal with line bundles over smooth manifolds as models for certain measurands in generalized classical mechanics.\newline

Following our discussion in section \ref{CanCob}, a measurand should be given abstractly by a set $\mathcal{M}$, whose elements are identified with the physical objects sharing some property that can be quantitatively compared, and a comparison map $\mathcal{P}:\mathcal{M}\times \mathcal{M}\to \Real$ satisfying the 2-out-of-3 property (\ref{2oo3}). In light of our discussion in section \ref{CatLine} we are compelled to give the following definition: a \textbf{measurand} is identified abstractly with a 1-dimensional real vector space $L\in\Line$. Physical objects sharing the same quantitatively comparable quality are identified with non-zero elements of $L$. For the reminder of this section, unless otherwise stated, by denoting $a\in L$ we assume $a$ to be non-zero. The \textbf{comparison map} is given canonically with the choice of $L$ by the ratio map $l:L\times L^{\times}\to \Real$, which was shown in (\ref{ratio1}) to satisfy the 2-out-of-3 identity
\begin{equation*}
    l_{ab}\cdot l_{bc}\cdot l_{ca}=1
\end{equation*}
as desired.\newline

Recall from section \ref{CanCob} that the comparison map could be equivalently reconstructed via the choice of a unit. The constructions above suggest that a unit should correspond to a choice of non-zero element to act as basis of the vector space. Indeed, with the current abstract definition of measurand, this approach can be formulated elegantly in terms of the category of lines. A \textbf{unit of measurement} $u$ for the measurand $L$ is identified with a factor in the category of lines
\begin{equation*}
    u:L\to \Real.
\end{equation*}

The reader may have noticed that only the multiplicative structure of lines, seen as real vector spaces, has been used so far for the mathematization of the basic properties of measurands, however maps have been required to be linear so as to make them morphisms in $\Line$. The additive structure of the line $L$ identifying a measurand is in fact empirically motivated by the combination construction $\parallel$ introduced in section \ref{CanCob}. There we saw how the comparison map, and, consequently, the unit map satisfy a linearity property (\ref{comblin}) with respect to the $\parallel$ operation. An obvious way to encode this property is to endow the set of physical objects themselves with a $\Real$-vector space structure, identifying $\parallel$ with usual vector addition $+$, and requiring maps to be linear. This then justifies the use of lines and factors in the category $\Line$ to abstractly represent measurands and the choice of units.\newline

With the proposed abstract definition of measurand in mind, let us revisit the examples of Length and Area. We identify the length and the area measurand with lines $L\in\Line$ and $A\in\Line$ respectively. The tile construction $\vee$ of section \ref{CanCob} gave an assignment of pairs of lengths to areas satisfying equations (\ref{tile1}) and (\ref{tile2}), a closer look at them reveals that they are simply a bilinearity property of the $\vee$ construction. If one takes elements of the length measurand $b_1,b_2\in L$ representing two canes $C_1$ and $C_2$, respectively, it is a simple exercise to check that making the association $b_1\otimes b_2 \sim C_1 \vee C_2$ and using the definition of the ratio map for the tensor product line leads precisely to equations (\ref{tile1}) and (\ref{tile2}). We thus see that this has effectively created a new measurand out of $L$ by taking its tensor product $L\otimes L$. If the area measurand $A$ was defined independently, as in our example with canes and cobblestones, the fact that tiles of canes are comparable by area to flat cobblestones (ensured by the presence of a standard wedge $\vee$) will be represented by the existence of a canonical factor $w_\vee:L\otimes L\to A$.\newline

In a similar fashion, the notion of ratio of measurands discussed at the end of section \ref{CanCob} gives an example of a measurand that is characterising the association of comparisons of two other measurands. Let us consider the Weight, Length and Density measurands as lines $W,L,D\in\Line$, respectively. In our discussion of section \ref{CanCob} we observed that quantitative comparisons of Density correspond to associations of comparisons of Length with comparisons of Weight. It is therefore a simple check to see that setting $D:=\Hom(L,W)\cong L^*\otimes W$ will recover the desired properties of what was called a derived measurand.\newline

In summary, we see that our example of a simple theory of Length, Area and Weight of canes and cobblestones in section \ref{CanCob} can be mathematically formulated in terms of two basic measurands: ``length of a straight object'' $L\in\Line$ and ``weight of an object'' $W\in\Line$. Other measurands of empirical relevance, such as Area $A$ and Density $D$, are given as tensor products of the lines characterizing the basic measurands:
\begin{equation*}
    A=L\otimes L, \qquad D=L^*\otimes W.
\end{equation*}

These considerations lead us to the following general definition: having set a \textbf{domain of observation}, a notion that will not be formalized here but that conceptually corresponds to specifying the class of physical objects that will be subject to scientific enquiry, a choice of quantitatively comparable characteristics between discernible phenomena will correspond to a choice of a finite family of lines $\{L_1,\dots, L_k\}$ in the category $\Line$, these will be called the \textbf{base measurands} of the domain or physical theory. All other measurands that may be of scientific relevance will be elements of the \textbf{measurand space} $M$, defined as the potentil of the family of base measurands in the category of lines $(L_1\dots L_k)^\otimes$, see the definition (\ref{potentil}).\newline

It follows structurally that any measurand space $M=(L_1\dots L_k)^\otimes$ comes naturally with the following surjective map:
\begin{align*}
t: M & \to \Int^k\\
a\in L_1^{n_1}\otimes \cdots \otimes L_k^{n_k} & \mapsto (n_1,\dots,n_k), 
\end{align*}
which is called the \textbf{type map} or \textbf{dimension map}. The element $a\in M$ is then said to be of type or dimension $t(a)\in\Int^k$.

\subsection{Typed Field structures in Measurand Spaces} \label{MeasDim}

So far in this section we have presented a first-principles approach to the mathematization of the notion of measurand which accounts for the metrological character of physical quantities as commonly understood in practical science. As philosophically satisfactory as this may be, if we were unable to recover the structure of dimensioned numbers seen in section \ref{DimNum} in a natural way, this would have been a futile effort. In this section we show that the mathematical definition of measurand space proposed in section \ref{MeasPot} indeed captures the structure of dimensioned numbers, which is, after all, the relevant structure for real-life  applications.\newline

Let a measurand space $M$ given by the potentil of the base lines $\{L_1,\dots, L_k\}\subset\Line$, that is
\begin{equation*}
    M=(L_1\dots L_k)^\otimes.
\end{equation*}
Recall that the choice of units on a measurand $L\in\Line$ is a choice of non-zero element $e\in L$ or, equivalently, a non-zero map $u:L\to\Real$. A choice of unit on $L$ induces a canonical choice of unit on $L^*$, i.e. the unique non-zero element $e^*$ such that $e^*(e)=1$, and we will denote this dual unit as a map by $u^{-1}:L^*\to\Real$, the notation $u^{-1}$, not to be confused with the mapping inverse of $u$, shall become useful below. Given another measurand $\overline{L}$ with choice of unit $\overline{u}:\overline{L}\to \Real$, the tensor product $L\otimes \overline{L}$ has also a canonical choice of unit $u\otimes \overline{u}:L\otimes \overline{L}\to \Real$. In particular, we will use the notation $u^n$ to denote the $n$-th tensor power of $u$ or $u^{-1}$ in the same way the notation $L^n$ was used in section \ref{CatLine} to denote the $n$-th tensor power of $L$ or $L^*$ depending on the sign of $n$. We then see that a choice of \textbf{base units} of $M$, that is, a choice of unit for each base line
\begin{equation*}
    \{u_i:L_i\to\Real\}_{i=1}^k,
\end{equation*}
induces a canonical choice of units on each measurand of $M$. Indeed, we have the following family of induced units for all measurands:
\begin{equation*}
    u_{n_1\cdots n_k}:L_1^{n_1}\otimes \cdots \otimes L_k^{n_k} \to\Real
\end{equation*}
where

\begin{equation*}
    u_{n_1\cdots n_k}:=u_1^{n_1}\otimes \cdots \otimes u_k^{n_k}
\end{equation*}
for all $n_1,\dots,n_k\in \Int$.\newline

We are now in the position to prove two results that will link the notion of dimensioned numbers of section \ref{DimNum} and our definition of measurand space.
\begin{prop}[Units in Measurand Spaces induce Typed Field Structures]\label{UnitMeasTypedField}
Let $M=(L_1\dots L_k)^\otimes$ be a measurand space. A choice of unit for each base line $\{u_i:L_i\to\Real\}_{i=1}^k$ endows $M$ with the structure of a $k$-dimensioned real field. In other words, a measurand space together with a choice of units has the structure of an abelian $\Real$-field with $\Int^k$ types. The typed field structure will be denoted by
\begin{equation*}
    (M, \cdot , +,t:M\to \Int^k,m:M\to \Real).
\end{equation*}
\end{prop}
\begin{proof}
Recall from the definition of typed field of section \ref{DimNum} that we first need to find two maps $t:M\to \Int^k$ and $m:M\to \Real$. The map $t$ is given by the type map of section \ref{MeasPot} and we define $m$ from the choice of units by setting
\begin{equation*}
    m|_{L_1^{n_1}\otimes \cdots \otimes L_k^{n_k}}:= u_{n_1\cdots n_k}.
\end{equation*}
Additionally, $M$ should carry a multiplicative binary operation and a partially-defined additive operation compatible with the maps $t$ and $m$. These are given by the tensor product $\otimes$ and vector addition $+$, respectively, the latter being only naturally defined on each measurand of the whole potentil forming the measurand space $M$, as it is taken to be a disjoint union of vector spaces. The compatibility of $t$ is tantamount to the equations derived for potentils at the end of section \ref{CatLine}, where we saw that $(M,\otimes)$ formed an abelian group isomorphic to $\Int^k$, the map $t$ is indeed realizing this isomorphism canonically. The compatibility of $m$ means that it should restrict to preimages of $t$ as an isomorphism of fields. Indeed, by construction we have:
\begin{equation*}
    t^{-1}(n_1,\dots,n_k)=L_1^{n_1}\otimes \cdots \otimes L_k^{n_k},
\end{equation*}
and since $m$ is defined to be the unit as a map, the fact that units are non-zero maps implies that $m$ is the desired isomorphism of 1-dimensional real vector spaces, which is of course an isomorphism of fields.
\end{proof}
\begin{prop}[Changes of Units induce Type Field Morphisms]\label{ChangeUnitMeas}
Let $M=(L_1\dots L_k)^\otimes$ be a measurand space and $\{u_i:L_i\to\Real\}_{i=1}^k$ and $\{u_i':L_i\to\Real\}_{i=1}^k$ two choices of base units, then the $k$-dimensioned field structures induced by these, $(M,t,m)$ and $(M,t,m')$, respectively, are related by a type-preserving isomorphism of typed fields.
\end{prop}
\begin{proof}
This follows simply from the fact that each change in base unit can be characterised by a non-zero real number $\alpha_i$ such that $u_i'=\alpha_i\cdot u_i$. The typed field isomorphism $\phi$ is then given on $t$-fibres as
\begin{equation*}
    \phi|_{L_1^{n_1}\otimes \cdots \otimes L_k^{n_k}}= \alpha_1^{n_1} \cdots \alpha_k^{n_k}.
\end{equation*}
\end{proof}

\section{Hamiltonian Mechanics} \label{HamMec}

Before we propose our generalization of the Hamiltonian phase space formalism, we should first make precise what is meant by ``Hamiltonian mechanics'', both in terms of its mathematical formulation and its physical content. In this section we review the conceptual foundations of the theory and give a categorical facelift to the usual formulation of canonical Hamiltonian mechanics.

\subsection{Conceptual Revision of the Hamiltonian Phase Space Formalism}\label{ConHam}

The available literature on Hamiltonian mechanics is vast and it comes in many different flavours: from the more rooted in historical formulations and common practice among physicists in \cite{goldstein2002classical} or \cite{landau1960classical}; to the more formal and familiar among geometers in \cite{abraham1978foundations} or \cite{arnold2013mathematical}; including the more mathematically sophisticated treatments connecting with the topic of quantization \cite{landsman2012mathematical} or \cite{puta2012hamiltonian}. However, despite the large volume of high quality texts available, it is hard to find a single source offering a systematic account of the physical principles underlying the choices of mathematical structures used in the modern embodiment of Hamiltonian mechanics. This section is aimed at bridging this gap. Far from giving a complete reconstruction of conventional Hamiltonian mechanics from first principles, we will list a series of physical principles that can be justifiably identified with particular choices of mathematical axioms. In doing so, we will also set some basic physical terminology to precise mathematical definitions and notation that will be used in subsequent sections.\newline

In order to fix some ideas and begin our discussion of the physical principles underlying Hamiltonian mechanics, let us consider a general picture of a working physicist that is interested in temporal series of measurements (represented by the tape of binary values) as shown in figure \ref{ModPhysObs}. In that diagram, metrological operations, such as the ones discussed in section \ref{CanCob}, are summarized as the bottom right arrow and the mathematical formulation of Hamiltonian mechanics would be an example of the ``Model'' blob. We hope the reader sees in this picture a rough sketch of the abstract working of physics as a formal natural science. We will be dealing with the ``Model'' side in our discussion almost exclusively but it will be helpful to have this picture in mind throughout the reminder of this section.\newline

\begin{figure}[h]
\centering
\includegraphics[scale=0.4]{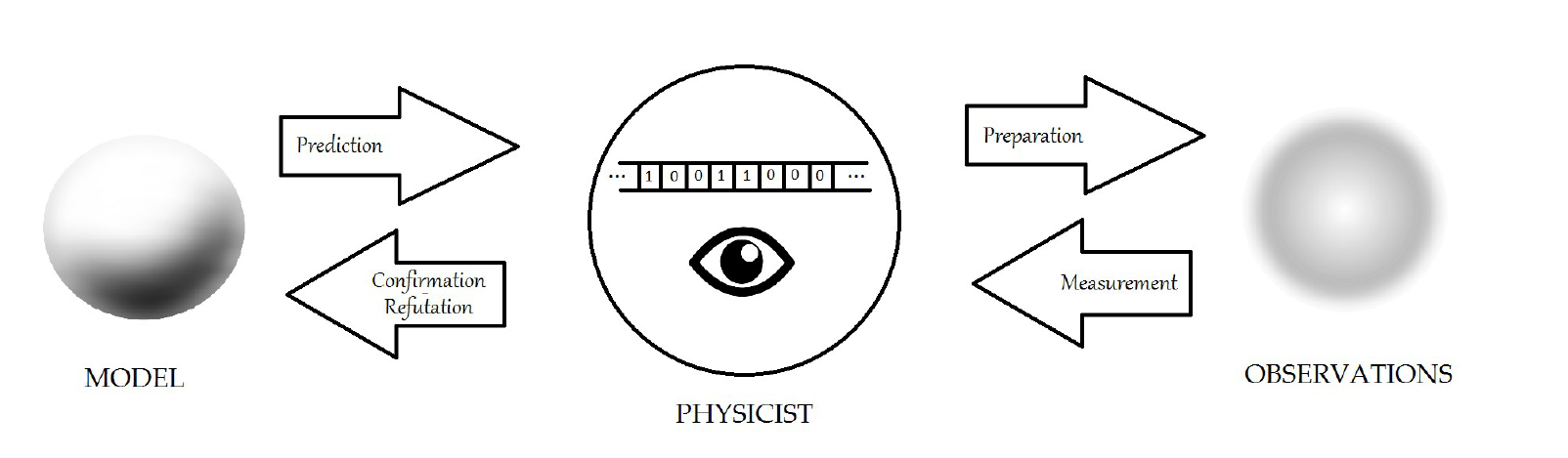}
\caption{A general paradigm of the scientific activity.}
\label{ModPhysObs}
\end{figure}

In a practical experimental situation, the temporal series of measurements will be a discrete sequence of rational numbers $\{m_i\in\Rat\}$, each value being the outcome of a quantitative comparison as described in section \ref{CanCob} e.g. positions of a particle at different times. Following the Principle of Refinement of section \ref{MeasPot}, the discrete temporal series of measurements will be replaced by a smooth dependence of real values $m(t)\in\Real$ with a real \textbf{time} parameter $t\in I\subset \Real$. An orientation of the observer's time interval $I$ is chosen so that the \textbf{past} $\{t<t_0\}$ and \textbf{future} $\{t>t_0\}$ of a given \textbf{present} $t_0\in I$ are defined. The task of a theory of mechanics is to find a mathematical model that takes as inputs the mathematical parameters corresponding to a given experimental preparation and produces as output a temporal series $m(t)$. What follows is a list of the main physical intuitions behind the formal content of a dynamical theory in Hamiltonian phase spaces. Note that the Principle of Refinement will be implicitly invoked every time we use smooth manifolds or smooth maps to represent physical notions.
\begin{itemize}
    \item \textbf{Principle of Realism.} A physical system exists independent of the observer studying it. Discernible configurations of the system, sets of equivalent experimental preparations and outcomes, are called \textbf{states} $s$ and are identified with points in a smooth manifold $P$, called the \textbf{phase space}.
    \item \textbf{Principle of Characterization.} Properties or characteristics of a system are smooth assignments of measurement outcome values to each point of the phase space. In the case of conventional classical mechanics this gives the usual definition of \textbf{observable} as a real-valued function $f\in \Obs{P}:=\Cin{P}$. Note that observables $\Obs{P}$ form a ring with the usual operations of point-wise addition and multiplication.
    \item \textbf{Principle of Kinematics.} The observer studying the system exists simultaneously with the system. In the same way that the observer's memory state is mapped uniquely into the time interval used to array the experimental temporal series, the physical system is thought to be in a single state corresponding to each of the observer's time parameter values. A \textbf{motion} is defined as a smooth curve in phase space $c:I\to P$ parameterized by the observer's time $t\in I\subset \Real$. Phase spaces are assumed to be path-connected so that any state is  connected to any other state by a motion, at least virtually, not necessarily physically. A family of motions $\{c(t)\}$ is called an \textbf{evolution} $\mathcal{E}$ on the phase space $P$ if all the states are included in the path of some curve. More precisely, we call $\mathcal{E}=\{c(t)\}$ an evolution if
    \begin{equation*}
        \forall s_0\in P \quad \exists c(t)\in \mathcal{E} \quad | \quad c(t_0)=s_0 \text{ for some } t_0\in I.
    \end{equation*}
    \item \textbf{Principle of Observation.} The observed time series of measurement outcomes are the result of observables taking values\footnote{We will not deal with non-trivial implementations of measurement outcomes, such as the collapse mechanism in the usual Hilbert space formulation of quantum mechanics, but we note that the general principles of the Hamiltonian phase space formalism described in this section, excluding the Principle of Observation, apply \emph{verbatim} to quantum mechanics.} along a particular motion. More concretely, given a motion $c(t)$ and an observable $f\in\Obs{P}$, the predicted temporal series is simply given by
    \begin{equation*}
        m(t)=(f\circ c)(t).
    \end{equation*}
    \item \textbf{Principle of Reproducibility.} Similar experimental preparations should give similar observation outcomes. An evolution $\mathcal{E}$ implementing this principle will satisfy the following property for all pairs of motions $c(t),c'(t)\in\mathcal{E}$:
    \begin{equation*}
        t_0\in I,\quad c(t_0)=c'(t_0) \quad \Rightarrow \quad c(t)=c'(t) \quad \forall t\in I,
    \end{equation*}
    in other words, $\mathcal{E}$ is a family of non-intersecting curves parameterized by $t\in I$.
    \item \textbf{Principle of Dynamics.} Future states of a physical system are completely determined\footnote{We are careful not to call this determinism since it is the abstract state of a system, not the measurement outcomes, what are assumed to evolve deterministically. Even when other measurement paradigms are in place, such as a collapse mechanism that forces the definition of a pre-measurement state and post-measurement state, if the theory relies on ordinary differential equations for the modelling of time evolution, the principle of dynamics will be implicitly used, at least to some degree.} by any given present state, at least locally. Enforcing this condition on motions for arbitrary small time intervals leads to the condition that a motion must be uniquely specified by the values of its tangent everywhere. In other words, a motion satisfying the Principle of Dynamics must be an integral curve of some smooth vector field, thus justifying the chosen name for this principle. Indeed, existence and uniqueness of ODEs implies that a smooth vector field $X\in\Sec{\Tan P}$ gives (at least locally) an evolution defined by the family of its integral curves $\mathcal{E}^X$ automatically satisfying the Principle of Reproducibility and the Principle of Dynamics. Arguing from a different angle, conjunction of the Principle of Observation with the Principle of Dynamics leads, via the local notion of directional derivative, to predicted temporal series being the integration of some derivation on observables. We thus conclude that an evolution on a phase space $P$ satisfying all the principles stated above is given by a choice of \textbf{dynamics} $X\in\Dyn{P}$, which is equivalently understood as a vector field on the phase space or a derivation of observables
    \begin{equation*}
        \Sec{\Tan P}\cong : \Dyn{P} :\cong \Dr{\Cin{P}}.
    \end{equation*}
    \item \textbf{Principle of Conservation.} Conserved quantities are a fundamental building block of experimental mechanics: one can only study time-dependent phenomena effectively when enough variables can be assumed to be constants to the effects of the experiment at hand. We could promote this to the more general and abstract requirement that any observable $f\in\Obs{P}$ has an associated evolution $\mathcal{E}^f$, given by some vector field $X_f$, along which the predicted time series are constant. We call such an assignment a \textbf{Hamiltonian map}
    \begin{equation*}
        \eta : \Obs{P}\to \Dyn{P},
    \end{equation*}
    which is required to satisfy
    \begin{equation*}
        \eta(f)[f]=0 \quad \forall f\in\Obs{P}.
    \end{equation*}
    In conventional Hamiltonian mechanics this is assumed to be given by the slightly stronger structure of a Poisson bracket on observables $(\Obs{P},\{,\})$, thus making the phase space into a Poisson manifold $(P,\pi)$. The Hamiltonian map is then $\eta:=\pi^\sharp \circ d$ and, following from the Jacobi identity of the Poisson bracket, it is a Lie algebra homomorphism.
    \item \textbf{Principle of Reductionism.} The theoretical description of a physical system specified as a subsystem of a larger system must be completely determined by the theoretical description of the larger system and the information of how the smaller system sits inside. This principle is implemented by demanding that knowledge about submanifolds in a phase space allows us to construct new phase spaces. In Hamiltonian mechanics this corresponds to the identification of constraints with coisotropic submanifolds and the procedure of Poisson reduction.
    \item \textbf{Principle of Combination.} The theoretical description of a system formed as a combination of other two systems must be completely determined by the theoretical descriptions of each of the parts and the information of how they interact. This is implemented by demanding that there is a \textbf{combination product} construction for phase spaces
    \begin{equation*}
        \between :  (P_1,P_2)\mapsto P_1 \between P_2
    \end{equation*}
    that acts as a categorical product in the corresponding category of phase spaces. In the case of conventional Hamiltonian mechanics, this is simply given by the usual definition of the Cartesian product of Poisson manifolds.
    \item \textbf{Principle of Symmetry.} A theoretical description of a system containing states that are physically indistinguishable should contain all the information to form the faithful theoretical description of the system. Physically indistinguishable states are commonly regarded to be orbits of some Lie group action on the phase space, thus an implementation of this principle will require that from the information of a Lie group action preserving some existing structure on a phase space, a new phase space is constructed whose states are classes of physically-indistinguishable states. In Hamiltonian mechanics this is implemented via the theory of Poisson group actions and equivariant moment maps.
\end{itemize}

\subsection{Canonical Hamiltonian Mechanics as the cotangent functor $\Cot$}\label{CanHamT}

In this section we shall see how the mathematical structure motivated by the principles introduced in section \ref{ConHam} appears naturally in standard classical Hamiltonian mechanics. In the interest of compactness and clarity of exposition, we are to give a categorical facelift to the usual formulation of the canonical formalism. Far from a comprehensive account, what follows should serve as a summary of essential structural features.\newline

Classical mechanics is erected upon the fundamental notion of \textbf{configuration space}: the set of static states of a physical system, such as the possible spatial positions of moving particles or the possible shapes of a vibrating membrane. Invoking the Principle of Refinement once more, our definition of a configuration space will be simply that of a smooth manifold whose points $q\in Q$ are identified with the different static states of a given physical system. We then propose the definition of \textbf{the category of configuration spaces} $\Conf$ simply to be the category of smooth real manifolds. We now give physical interpretation to the natural categorical structure present in $\Conf$:
\begin{itemize}
    \item The measurable static properties of a physical system with configuration space $Q\in\Conf$, what we call \textbf{static observables}, are simply the smooth real-valued functions $\Obs{Q}:=\Cin{Q}$. More generally, we see that the assignment of observables to configuration spaces is given by a contravariant functor $\text{Obs}:\Conf\to\Ring$, which we regard as the categorical version of the Principle or Characterization for configuration spaces.
    \item A subsystem is characterised by restricting possible positions of a larger system, that is, by an inclusion of an embedded submanifold $i:S\to Q$. We then see that monomorphisms in the category $\Conf$ implement the Principle of Reductionism for configuration spaces.
    \item Physically-indistinguishable static states in a configuration space $Q\in\Conf$ are related by equivalence relations $\sim$ that have quotients faithfully characterizing the physical system, that is, there is a surjective submersion $p:Q\to Q/\sim$. In particular, a free and proper action of a Lie group $G\Acts Q$ gives an example implementing the Principle of Symmetry via $p:Q\to Q/G$. These are epimorphisms in the category of configuration spaces $\Conf$.
    \item Given two configuration spaces $Q_1,Q_2\in\Conf$ representing the possible positions of two physical systems, the combined system resulting from considering them together will have static states given by all the possible pairs of static states in each of the two systems. The categorical manifestation of the Principle of Combination for configuration spaces is then simply the existence of the Cartesian product of smooth manifolds $Q_1\times Q_2\in\Conf$ as a categorical product.
    \item A temporal series of static states will be called a \textbf{path} of the physical system. Paths in a configuration space $Q\in\Conf$ will be given by smooth curves $r:I\to Q$ parameterized by the observer's time parameter $t\in I\subset \Real$. Following the discussion on the Principle of Dynamics of section \ref{ConHam}, dynamics on the static states of a configuration space $Q\in\Conf$ are identified as
    \begin{equation*}
        \Sec{\Tan Q}\cong : \Dyn{Q} :\cong \Dr{\Cin{Q}}.
    \end{equation*}
\end{itemize}

The basic goal of the canonical formalism is to provide a systematic assignment of a ``space of dynamical states'' to a given configuration space in such a way that it acts as the phase space of the system and the principles of section \ref{ConHam} are naturally implemented. This will be described below as a functor from the category of configuration spaces $\Conf$ to a suitable category of phase spaces $\mathcal{P}$. This formal correspondence will be called the Hamiltonian functor. For the target category we will choose (a subcategory of) Weinstein's symplectic category $\Symp$, whose objects are symplectic manifolds and the morphisms are canonical relations, i.e. Lagrangian submanifolds of the Cartesian product of symplectic manifolds. We direct the reader to \cite{weinstein2009symplectic} for a detailed discussion on the subtleties of this definition of symplectic category.\newline

In order to motivate the definition of the Hamiltonian functor, let us recall a few facts about the symplectic geometry of cotangent bundles:
\begin{itemize}
    \item The cotangent bundle of any smooth manifold $Q$ carries a canonical symplectic structure
    \begin{equation*}
        \begin{tikzcd}
        \Cot Q \arrow[d, "\pi_Q"]\\
        Q
        \end{tikzcd}
        \qquad \qquad \omega_Q:=-d\theta_Q, \quad \text{ with }\quad  \theta_Q|_{\alpha_q}(v):=\alpha_q(\Tan_{\alpha_q}\pi_Q (v)) \quad \forall v\in \Tan_{\alpha_q}(\Cot Q).
    \end{equation*}
    A symplectic structure is a non-degenerate case of Poisson structure thus inducing a Poisson algebra on functions
    \begin{equation*}
        (\Cin{\Cot Q},\cdot,\{,\}_Q).
    \end{equation*}
    \item There are natural inclusions of vector fields and functions on $Q$ as linear and basic functions on $\Cot Q$, respectively:
    \begin{equation*}
        l_Q:\Sec{\Tan Q}\to \Cin{\Cot Q}, \qquad \pi_Q^*:\Cin{Q}\to \Cin{\Cot Q}.
    \end{equation*}
    These functions are, in fact, enough to generate all the Poisson brackets since $(\Cin{\Cot Q},\cdot,\{,\}_Q)$ is an example of a (fibre-wise) linear Poisson structure,
    \begin{align*}
        \{l_Q(X),l_Q(Y)\}_Q & =  l_Q([X,Y])\\
        \{l_Q(X),\pi_Q^*f\}_Q & =  \pi_Q^* X [f]\\
        \{\pi_Q^*f,\pi_Q^*g\}_Q & = 0.
    \end{align*}
    \item The cotangent bundle of a Cartesian product is canonically symplectomorphic to the Cartesian product of symplectic manifolds, which gives the desired combination product as a categorical product
    \begin{equation*}
        (\Cot (Q_1\times Q_2),\omega_{Q_1\times Q_2})\stackrel{\Symp}{\cong} (\Cot Q_1\times \Cot Q_2, \Proj_1^*\omega_{Q_1} \oplus \Proj_2^*\omega_{Q_2}).
    \end{equation*}
    \item A smooth map $\varphi:Q_1\to Q_2$ induces a canonical relation $\Cot \varphi\subset \Cot Q_1 \times \overline{\Cot Q_2}$ called the cotangent lift of $\varphi$  and defined by
    \begin{equation*}
        \Cot \varphi:=\{(\alpha_q,\beta_p)| \quad \varphi(q)=p, \quad \alpha_q=(\Tan_q\varphi)^*\beta_p\}.
    \end{equation*}
    Here $\overline{\Cot Q_2}$ denotes $(\Cot Q_2,-\omega_{Q_2})$. When $\varphi:Q_1\to Q_2$ is a diffeomorphism, its cotangent lift becomes the graph of a symplectomorphism $\Cot \varphi :\Cot Q_2\to \Cot Q_1$. Composition of cotangent lifts is always strongly transversal in the sense of \cite{weinstein2009symplectic} and thus the cotangent bundle construction can be seen as a contravariant functor
    \begin{equation*}
        \Cot : \Conf \to \Symp.
    \end{equation*}
    \item An embedded submanifold $i:S\to Q$ gives the coisotropic submanifold $\Cot Q|_S\subset (\Cot Q,\omega_Q)$ inducing the coisotropic reduction diagram
    \begin{equation*}
    \begin{tikzcd}
        \Cot Q|_S \arrow[r, hookrightarrow] \arrow[d, twoheadrightarrow] & (\Cot Q,\omega_Q) \\
        (\Cot S,\omega_S).
    \end{tikzcd}
    \end{equation*}
    Note that the surjective submersion results from quotienting by the foliation given by the conormal bundle of $S$, in other words
    \begin{equation*}
        \Cot S\cong \Cot Q|_S/(\Tan S)^0
    \end{equation*}
    as vector bundles over $S$.
    \item Given a free and proper group action $\phi: G\times Q\to Q$, with infinitesimal action $\psi:\mathfrak{g}\to \Sec{\Tan Q}$ and smooth orbit space $\Tilde{Q}:=Q/G$, the cotangent lift gives an action by symplectomorphisms $\Cot \phi : G\times \Cot Q\to \Cot Q$ with equivariant (co)moment map given by
    \begin{equation*}
        \overline{\mu}:=l\circ \psi : \mathfrak{g}\to \Cin{\Cot Q},
    \end{equation*}
    this induces the symplectic reduction diagram
    \begin{equation*}
        \begin{tikzcd}
        \mu^{-1}(0) \arrow[r, hookrightarrow] \arrow[d, twoheadrightarrow] & (\Cot Q,\omega_Q) \\
        (\Cot \Tilde{Q},\omega_{\Tilde{Q}}).
        \end{tikzcd}
    \end{equation*}
\end{itemize}

In light of these results, we are compelled to define the \textbf{category of phase spaces} simply as the image of the category of configuration spaces under the cotangent functor, $\mathcal{P}:=\Cot (\Conf)$. This choice, instead of simply setting $\mathcal{P}:=\Symp$, is made so that composition of canonical relations is always well-defined, otherwise transversality issues may occur. We are now in the position to argue that the structural content of standard Hamiltonian mechanics can be summarized as the \textbf{Hamiltonian functor}, mapping categorical structures of configuration spaces to those of phase spaces as shown below:
\begin{center}
\begin{tabular}{ c c c }
Category of Configuration Spaces & Hamiltonian Functor  & Category of Phase Spaces \\
\hline
 $\Conf$ & $\begin{tikzcd}\phantom{A} \arrow[r,"\mathbb{H}_\Conf"] & \phantom{B} \end{tikzcd}$ & $\mathcal{P}$ \\ 
 $Q$ & $\begin{tikzcd} \phantom{Q} \arrow[r, "\Cot ",mapsto] & \phantom{Q} \end{tikzcd}$ & $(\Cot Q,\omega_Q)$ \\ 
 $\Obs{Q}$ & $\begin{tikzcd} \phantom{Q} \arrow[r,hookrightarrow, "\pi_Q^*"] & \phantom{Q}  \end{tikzcd}$ & $\Obs{\Cot Q}$ \\
 $\Dyn{Q}$ & $\begin{tikzcd} \phantom{Q} \arrow[r,hookrightarrow, "l_Q"] & \phantom{Q}  \end{tikzcd}$ & $\Obs{\Cot Q}$ \\
 $Q_1\times Q_2$ & $\begin{tikzcd} \phantom{Q} \arrow[r,mapsto, "\Cot "] & \phantom{Q}  \end{tikzcd}$ & $\Cot Q_1\times \Cot Q_2$ \\
 $\varphi:Q_1\to Q_2$ & $\begin{tikzcd} \phantom{Q} \arrow[r,mapsto, "\Cot "] & \phantom{Q}  \end{tikzcd}$ & $\Cot \varphi\subset \Cot Q_1\times \overline{\Cot Q_2}$ Lagrangian\\
 $i:S\to Q$ & $\begin{tikzcd} \phantom{Q} \arrow[r,mapsto, "\Cot"] & \phantom{Q}  \end{tikzcd}$ & $\Cot Q|_S\subset \Cot Q$ coisotropic \\
 $G\Acts Q$ & $\begin{tikzcd} \phantom{Q} \arrow[r,mapsto, "\Cot"] & \phantom{Q}  \end{tikzcd}$ & $G\Acts \Cot Q$ Hamiltonian action \\
\end{tabular}
\end{center}
We invite the reader to compare the list of principles of section \ref{ConHam} with the results presented in this section in order to confirm that the Hamiltonian functor indeed captures the desired physical insights of the phase space formalism. We note, in particular, that both the static observables and the dynamics over static states become part of the observables under the Hamiltonian functor. This corresponds to the usual physical intuition that position and velocities are the fundamental properties of a system to be studied in a theory of classical mechanics. The upshot of our categorical formulation is that, once a choice of the category of configuration spaces is found to correctly implement the desired physical principles attributed to static states, the Hamiltonian functor into phase spaces is given without any further choices. It is in this way that the historical name of \textbf{canonical Hamiltonian mechanics} finds a renewed meaning.\newline

Let us remark at this point that our definition of the category of phase spaces as $\mathcal{P}=\Cot (\Conf)$ has been chosen to emphasize the essential structural aspects of mechanics but the reader is reminded of the numerous examples of classical systems whose dynamical phase space is modelled more generally by Poisson manifolds. Examples of these include constrained systems that give coisotropic reductions of cotangent bundles, which are generically non-symplectomorphic to cotangent bundles, and duals of Lie algebras of symmetries, e.g. the angular velocity of a rigid body is seen as an element of $\mathfrak{so}(3)^*$. For a more comprehensive discussion we could set the category of Poisson manifolds with coisotropic relations as the category of phase spaces, $\mathcal{P}=\Poiss$, but this is beyond the scope of the present work. \newline

Once a physical system with configuration space $Q$ is canonically assigned a phase space $\Cot Q$ as above, the only remaining task left for the physicist is to determine which dynamics will produce predictions that match experimental measurements. The Poisson structure present in $\Cot Q$ reduces this problem to the choice of an observable $h\in\Obs{\Cot Q}$ which, under the Hamiltonian map, gives a choice of conservative dynamics $\{h,-\}_Q\in\Dyn{\Cot Q}$. This distinguished observable is often called the \textbf{energy} of the system. It generates the time evolution of the system and is, itself, a fundamental conserved quanity. Mathematically, this is a trivial fact by construction since
\begin{equation*}
    \{h,-\}_Q[h]=\{h,h\}_Q=0
\end{equation*}
by antisymmetry of the Lie bracket. Given two systems with a choice of energy $(Q_1,h_1)$ and $(Q_2,h_2)$, where $h_i\in\Cin{\Cot Q_i}$, their combined product has a canonical choice of energy given by the sum of pull-backs via the canonical projections $h_1+h_2:=\Proj_1^*h_1+\Proj_2^*h_2\in\Cin{\Cot Q_1\times \Cot Q_2}$. This gives an extra line of assignments to the Hamiltonian functor:
\begin{center}
\begin{tabular}{ c c c }
Category of Configuration Spaces & Hamiltonian Functor  & Category of Phase Spaces \\
\hline
  $(Q_1,h_1)\times (Q_2,h_2)$ & $\begin{tikzcd} \phantom{Q} \arrow[r,mapsto, "\Cot "] & \phantom{Q}  \end{tikzcd}$ & $h_1+h_2 \in\Cin{\Cot Q_1\times \Cot Q_2}$ \\
\end{tabular}
\end{center}

The phase space formalism described so far in this section is general enough to account for a vast class of mechanical systems, however, this generality comes at a price: the Hamiltonian functor above fails to select a preferred choice of energy observable for a given configuration space. Turning to one of the earliest examples of mechanics we find inspiration to redefine the category of configuration spaces in order to account for some extra physical intuitions. In Newtonian mechanics, configuration spaces are submanifolds of Euclidean space and Cartesian products thereof, thus always carrying a Riemannian metric, encoding the physical notion of distance and angle; often also with a choice of potential, which is a function on the configuration space. This motivates us to refine our definition of configuration space and define the \textbf{category of Newtonian configuration spaces} $\Conf_{\text{Newtonian}}$ whose objects are triples $(Q,g,V)$, where $Q\in\Conf$, $g\in \Sec{\odot ^2 \Cot Q}$ Riemannian metric and $V\in\Cin{Q}$, and whose morphisms $\varphi:(Q_1,g_1,V_1)\to (Q_2,g_2,V_2)$ are smooth maps $\varphi:Q_1\to Q_2$ such that $g_1-\varphi^*g_2$ is positive semi-definite and $V_1=\varphi^*V_2$. When $\varphi$ is a diffeomorphism, a morphism in this category is, in particular, an isometry between $Q_1$ and $Q_2$. Note that a metric defines a quadratic form on tangent vectors $K_g:\Tan Q\to \Real$ given simply by $K_g(v):=\tfrac{1}{2}g(v,v)$, then, given the equivalence between quadratic forms and positive-definite bilinear forms, we see that the datum of a Newtonian configuration space $(Q,g,V)$ is equivalent to a choice of function on static states $V\in\Cin{Q}$ and a choice of quadratic function on velocities $K\in\Cin{\Tan Q}$, indeed the two fundamental notions of classical mechanics. The Cartesian product of configuration spaces gets updated to a categorical product in $\Conf_{\text{Newtonian}}$ by setting
\begin{equation*}
    (Q_1,g_1,V_1) \times (Q_2,g_2,V_2) := (Q_1\times Q_2,g_1+g_2,V_1+V_2)
\end{equation*}
where
\begin{align*}
    g_1+g_2 &:= \Proj^*_1g_1 \oplus \Proj_2^*g_2\in \Sec{\odot ^2 (\Tan Q_1\oplus \Tan Q_2))}\cong\Sec{\odot ^2 \Tan (Q_1\times Q_2)}, \\
    V_1+V_2 &:= \Proj^*_1V_1+\Proj_2^*V_2\in \Cin{Q_1\times Q_2}.
\end{align*}
A metric $g$ on $Q$ gives the usual musical isomorphism
\begin{equation*}
        \begin{tikzcd}
        \Tan Q \arrow[r, "\flat",yshift=0.7ex] & \Cot Q \arrow[l,"\sharp",yshift=-0.7ex]
        \end{tikzcd},
\end{equation*}
which can be used to regard the quadratic function $K$ identified with the metric $g$ as a quadratic function on the cotangent bundle by pull-back: $\sharp^*K_g\in\Cin{\Cot Q}$. Once a configuration space $(Q,g,V)$ is fixed, we now see how the phase space formalism for Newtonian configuration spaces does provide a canonical choice of energy by
\begin{equation*}
    E_{g,V}:=\sharp^*K_g+\pi_Q^*V,
\end{equation*}
what we call the \textbf{Newtonian energy}. This name is further justified by the fact that a direct computation shows that solving the Hamiltonian dynamics of this observable is equivalent to solving Newton's equations on a Riemannian manifold background $(Q,g)$ and with force field $F=-\sharp(dV)$. For two Newtonian configuration spaces $(Q_1,g_1,V_1)$ and $(Q_2,g_2,V_2)$, the categorical product construction above gives the following additivity property of Newtonian energy:
\begin{equation*}
    E_{g_1+g_2,V_1+V_2}=E_{g_1,V_1}+E_{g_2,V_2}.
\end{equation*}

Newtonian configuration spaces $\Conf_{\text{Newtonian}}$ form a subcategory of $\Conf$, thus the Hamiltonian functor must operate in a similar way but with the updated categorical product and the added assignment of the Newtonian energy as a canonical choice of observable to generate the dynamical evolution. We shall then update the Hamiltonian functor given above for general configuration spaces by adding the extra correspondence of Newtonian energy:
\begin{center}
\begin{tabular}{ c c c }
Category of Configuration Spaces & Hamiltonian Functor  & Category of Phase Spaces \\
\hline
 $\Conf_{\text{Newtonian}}$ & $\begin{tikzcd}\phantom{A} \arrow[r,"\mathbb{H}_{\Conf_{\text{Newtonian}}}"] & \phantom{B} \end{tikzcd}$ & $\mathcal{P}$ \\ 
 $(Q,g,V)$ & $\begin{tikzcd} \phantom{Q} \arrow[r, "\Cot ",mapsto] & \phantom{Q} \end{tikzcd}$ & $E_{g,V}\in\Obs{\Cot Q}$ \\
 $(Q_1,g_1,V_1)\times (Q_2,g_2,V_2)$ & $\begin{tikzcd} \phantom{Q} \arrow[r, "\Cot ",mapsto] & \phantom{Q} \end{tikzcd}$ & $(\Cot (Q_1\times Q_2),E_{g_1,V_1}+E_{g_2,V_2})$
\end{tabular}
\end{center}
Note that this result in the category of Newtonian configuration spaces motivates the additivity of energy of the combination product of two general phase spaces $(\Cot Q_1\times \Cot Q_2,h_1+h_2)$ discussed above.

\section{Dimensioned Hamiltonian Mechanics} \label{DimHamM}

In this section we set out to generalize Hamiltonian mechanics in such a way that it naturally incorporates the structure of measurand spaces as defined in section \ref{ModMes}. Our new formulation will be shown to exhibit a categorical structure entirely analogous to the Hamiltonian functor $\mathbb{H}_\Conf$ described in section \ref{HamMec} for conventional Hamiltonian mechanics.

\subsection{Introducing Physical Dimension in Geometric Mechanics}\label{DimGeoM}

The central idea in the mathematical formulation of dimensioned geometric mechanics is the introduction of the category of line bundles $\Line_\Man$ as the category of dimensioned configuration spaces. Note first that this is, indeed, a direct generalization of the conventional ``unit-less'' formalism since we can regard the category of conventional configuration spaces, i.e. just smooth manifolds, as the subcategory $\Real_\Man\subset \Line_\Man$ of trivial line bundles with fibre-wise identity line bundle morphisms.\newline

This generalization is conceptually motivated by our discussion in section \ref{ModMes}, where we argued that the mathematical object that captures the desirable formal structure of measurable physical quantities is a measurand space $M$, comprised of tensor powers of some collection of $1$-dimensional vector spaces $(L_1\dots L_k)^\otimes$, rather than simply a copy of the real numbers $\Real$. We emphasize here the two fundamental aspects of this new perspective:
\begin{itemize}
    \item Outcomes of measurements of different observables are no longer regarded as elements of the same mathematical set $\Real$, but, instead, they fit into a framework, a measurand space $M$, where the notion of physical dimension is naturally implemented by the potentil structure.
    \item Physical quantities are treated in an entirely unit-free manner. Not only is the measurand space framework manifestly unit-free, as opposed to unit-less, but it also allows for a rigorous definition of unit of measurement as non-zero elements of lines in a measurand. This provides a natural mathematical language to articulate the question of whether some theoretical construction is unit-independent and thus may be of physical relevance.
\end{itemize}
Recall from our discussion at the beginning of section \ref{ConHam} that the combined principles of Realism, Characterization, Kinematics and Observation, and guided by the, then unquestioned, idea that measurement outcomes in practical science are just numerical values, lead us to the definition of observables (the technical term for ``physical quantity'' in the context of mechanics) as smooth assignments of real values to points of a smooth phase space. Following an analogous reasoning, but with our updated notion of physical quantity as elements of measurand spaces, we are compelled to define the \textbf{category of dimensioned configuration spaces} as the category of line bundles over smooth manifolds with factors, i.e. fibre-wise invertible line bundle maps, as morphisms. Staying consistent with conventional notations for configuration spaces we will denote this category as $\Line_\Conf$ and objects there will be indistinctly denoted by $\lambda:L\to Q$, $L_Q$ or $L$. Having all the technical results of section \ref{CatLineB} at hand, we give a physical interpretation for the categorical structure naturally present in $\Line_\Conf$ that is entirely analogous to the one given for the category of conventional configuration spaces at the beginning of section \ref{CanHamT}:
\begin{itemize}
    \item The physical interpretation of the points of the base manifold of a dimensioned configuration space $L_Q$ is exactly the same as the points of a configuration space, they represent the \textbf{static state} of the physical system. In this sense, the space of static states of a physical system will be independent of the particular dimensions of the observables that will be defined on it, as is indeed the case in conventional geometric mechanics.
    \item The measurable static properties of a fixed physical dimension are identified with the smooth sections of some fixed line bundle $\Sec{L_Q}$. We thus identify the collection of all possible measurable properties of a fixed physical dimension with a choice of dimensioned configuration space $L_Q$. We call these the \textbf{static observables of dimension $L$} of the configuration space $Q$ and denote them by $\Obs{L_Q}:=\Sec{L_Q}$. Properties of factor pull-backs in the category of line bundles then ensure that we have an observable contravariant functor $\text{Obs}:\Line_\Conf\to\textsf{RMod}$.
    \item We recover the notion of ``unit-less'' observable from the conventional treatment of configuration spaces by introducing the notion of a choice of unit. A \textbf{unit} in a dimensioned configuration space $L_Q$ is a local trivialization, equivalently, a local non-vanishing section $u\in\Sec{L_Q}$. Restricting to the open subset where $u$ is defined allows us to see any other static observable $s\in\Obs{L_Q}$, i.e. the restriction of some arbitrary section to the open subset, as a local real-valued function $\Tilde{s}_u$ determined uniquely by the equation $s=\Tilde{s}_u\cdot u$. Thus, a choice of a unit $u$ allows (locally) for a functorial assignment of the form
    \begin{equation*}
        u:\Obs{L_Q}\to \Obs{Q}.
    \end{equation*}
    \item A subsystem is characterised by restricting possible positions of a larger system, that is, by an inclusion of an embedded submanifold $i:S\to Q$. Our discussion about submanifolds of line bundles and embedding factors in section \ref{CatLineB} ensures that this situation is indeed equivalent to the monomorphisms of the category of dimensioned configuration spaces.
    \item Physically-indistinguishable static states must produce measurement outcomes that are indistinguishable as elements of a line bundle over the configuration space and thus should give a basic quotient of line bundles. Our discussion about submersion factors in section \ref{CatLineB} shows that this notion precisely corresponds to the epimorphisms of the category of dimensioned configuration spaces.
    \item Given two dimensioned configuration spaces $L_{Q_1}$ and $L_{Q_2}$, the line product construction of section \ref{CatLineB} gives the direct analogue of the Cartesian product of conventional configuration spaces. We thus regard $L_{Q_1}\utimes L_{Q_2}$ as the categorical implementation of the Principle of Combination for dimensioned configuration spaces.
    \item Paths of a physical system, i.e. temporal series of static states, are simply recovered as smooth curves on the base space of a dimensioned configuration space. In this manner, conventional dynamics $\Dyn{Q}$ are simply recovered as the vector fields on the base manifold. However, the extra structure introduced by the presence of the line bundle induces a new dynamical aspect of configuration spaces. Given a dimensioned configuration space $L_Q$, all the non-zero fibre elements over a point $L_q$ represent different choices of unit for the same type of physical quantity. Then, any measurement performed on a system moving along a path $c(t)$ passing through $q$ at a time $t_0$ will have to be unit-compatible with any measurement performed at a later time $t>t_0$. This means that the choice of unit should be preserved along the motion of a path. Considering a unit $u$ as a local section, this is ensured locally by construction, however, taking all the possible arbitrary choices of local unit around the point $q$, forces the existence of a $1$-parameter family of fibre-wise isomorphisms covering the smooth curve $c(t)$. These are nothing but smooth families of line bundle automorphisms, which are given infinitesimally by line bundle derivations, and thus we identify the \textbf{dimensioned dynamics} of a dimensioned configuration space $L_Q$ as $\Dyn{L_Q}:=\Dr{L_Q}$. The anchor of the der bundle $\delta:\Der L_Q\to \Tan Q$ allows to connect dimensioned dynamics with ordinary dynamics via the surjective map
    \begin{equation*}
        \delta_*:\Dyn{L_Q}\to \Dyn{Q}.
    \end{equation*}
\end{itemize}

Despite having provided ample philosophical motivation for the definition of the category of dimensioned configuration spaces $\Line_\Conf$, and having found seemingly convincing physical interpretation for the natural structure that it exhibits, all this effort would be in vain unless a Hamiltonian functor, analogous to $\mathbb{H}_\Conf$ in section \ref{CanHamT}, can be constructed for this category. Fortunately, this is indeed the case and section \ref{CanHamJ} below contains all the necessary proofs supporting this claim.\newline

\subsection{Canonical Hamiltonian Mechanics as the jet functor $\Jet^1$}\label{CanHamJ}

The basic intuition behind the construction of the Hamiltonian functor of section \ref{CanHamT} was the desire to, given a configuration space, find a phase space that would carry a space of observables with a Lie algebra structure canonically containing the static observables and dynamics as subalgebras. The starting data is now a dimensioned configuration space, specified by some line bundle over a configuration space $L_Q$, whose static observables are $\Sec{L_Q}$ and whose dynamics are $\Dr{L_Q}$. The proposition below establishes the jet bundle $\Jet^1 L_Q$ as the \textbf{canonical dimensioned phase space} of a dimensioned configuration space $L_Q$.
\begin{prop}[The Canonical Contact Manifold Associated to a Line Bundle]\label{CanonContMan}
Let $\lambda:L\to Q$ be a line bundle and $\normalfont \pi:\Jet^1L \to Q$ its jet bundle, then there is a canonical contact structure $\normalfont (\Jet^1 L,\text{H}_L)$ such that the contact line bundle is isomorphic to the pull-back bundle in the category of L-vector bundles, i.e.
\begin{equation*}
\normalfont
    \Tan (\Jet^1 L)/\text{H}_L\cong \pi^*L.
\end{equation*}
We denote this line bundle by $\normalfont L_{\Jet^1 L}$. Furthermore, the non-degenerate Jacobi structure induced in the line bundle $\normalfont (L_{\Jet^1 L},\{,\}_L)$ is fibre-wise linear and it is completely determined by the algebraic structure of derivations acting on sections
\begin{align*}
    \{l_a,l_b\}_L & =  l_{[a,b]}\\
    \{l_a,\pi^*s\}_L & =  \pi^* a [s]\\
    \{\pi^*s,\pi^*r\}_L & = 0
\end{align*}
for all $s,r\in\Sec{L}$, $\normalfont a,b\in\Sec{\Der L}$ and where $\normalfont l:\Sec{\Der L}\to \Sec{L_{\Jet^1 L}}$ is the inclusion of derivations as fibre-wise linear sections of $\normalfont L_{\Jet^1 L}$.
\end{prop}
\begin{proof}
The hyperplane distribution $\text{H}_L\subset \Tan (\Jet^1 L)$ is the usual Cartan distribution defined in general jet bundles of vector bundles, however, in our case it can be regarded as the kernel of the \textbf{canonical contact form} $\theta\in\Sec{\Tan^{*(\pi^*L)}(\Jet^1 L)}$ defined as the line bundle analogue of the Liouville $1$-form on the cotangent bundle. More precisely, denoting by $\varpi:\Jet^1 L\to L$ the surjective bundle map of the Spencer sequence, the canonical contact form is explicitly given at any point $j^1_xu\in\Jet^1 L$ of the the jet bundle by
\begin{equation*}
    \theta_{j^1_xu}:=\Tan_{j^1_xu}\varpi - \Tan_x u\circ \Tan_{j^1_xu}\pi
\end{equation*}
where we note that the explicit use of $\Tan_xu$ is well-defined from the fact that $j^1u$ is defined as the equivalence class of all sections agreeing in value and tangent map at $x\in M$. The map above is mapping tangent spaces of the vector bundles $\theta_{j^1_xu}:\Tan_{j^1_xu}\Jet^1 L\to \Tan_{u(x)}L$, in order to make it into a $L$-valued $1$-form we will use the fact that vertical subspaces of the total space of a vector bundle are canonically isomorphic to the fibres, $\Tan_{u(x)}^{\text{Vert}}L=\Ker{\Tan_{u(x)} \lambda}\cong L_x$ and that the image of $\theta_{j^1_xu}$ is always vertical
\begin{equation*}
    \Tan_{u(x)} \lambda\circ \theta_{j^1_xu}=\Tan_{j^1_xu}(\lambda\circ \varpi)-\Tan_{x}(\lambda\circ u)\circ \Tan_{j^1_xu}\pi=\Tan_{j^1_xu}\pi - \Tan_{j^1_xu}\pi =0,
\end{equation*}
where we have used the fact that sections and jet prolongations fit in the commutative diagram
\begin{equation*}
\begin{tikzcd}
\Jet^1 L \arrow[rr, "\varpi"] \arrow[dr, "\pi"] & & L \arrow[dl, "\lambda"'] \\
 & M \arrow[ul,bend left,"j^1u"]\arrow[ur,bend right,"u"'] & 
\end{tikzcd}
\end{equation*}
We then define $\text{H}_L:=\Ker{\theta}$, which is shown to be a hyperplane distribution from simple point-wise dimension-counting. Applying the first isomorphism theorem for vector spaces fibre-wise, it clearly follows that $\Tan (\Jet^1 L)/\text{H}_L\cong \pi^*L$. The non-degenerate Jacobi structure on $\pi^*L$ appears as the $L$-dual to the Lie(Jacobi) algebroid structure present in $\Der L$ analogously to the canonical symplectic structure being the linear Poisson structure on $\Cot M$ dual to the Lie algebroid $\Tan M$. This can be seen  explicitly via the natural inclusions of fibre-wise constant sections $\pi^*:\Sec{L}\hookrightarrow \Sec{\pi^*L}$ and fibre-wise linear sections $l:\Sec{\Der L}\hookrightarrow \Sec{\pi^*L}$ satisfying
\begin{equation*}
    \pi^*(f\cdot u)=\pi^*f\cdot \pi^*u \qquad l_{f\cdot a}=\pi^*f\cdot l_a
\end{equation*}
for all $f\in\Cin{M}$, $u\in\Sec{L}$ and $a\in\Sec{\Der L}$, which uniquely determine a fibre-wise linear $\Real$-linear bracket $\{,\}$ determined by the defining identities in the proposition and whose Jacobi identity follows from the Jacobi identity of the commutator bracket of derivations. We can show that this bracket extends to a Jacobi structure with an extension-by-symbol argument via some pair $(X,\Lambda)$ by identifying a subspace of spanning functions for $\Cin{\Jet^1 L}$. This subspace is given by the natural inclusions of fibre-wise constant functions $\pi^*:\Cin{M}\hookrightarrow \Cin{\Jet^1 L}$ and fibre-wise linear functions $l:\Sec{\Der L\otimes L^*}\hookrightarrow \Cin{\Jet^1 L}$. Using the only available ways in which sections of these bundles can act on functions, we find the Hamiltonian vector field $X$ defined on spanning functions by the conditions:
\begin{align*}
    X_{l_a}[l_{b\otimes \sigma}]&=l_{[a,b]\otimes \sigma}   & X_{l_a}[\pi^*f]&=\pi^*\delta(a)[f]\\
    X_{\pi^*s}[l_{b\otimes \sigma}]&=-\pi^*\sigma(b[s]) & X_{\pi^*s}[\pi^*f]&=0
\end{align*}
for all $f\in\Cin{M}$, $s\in\Sec{L}$, $a\in\Sec{\Der L}$, $b\otimes \sigma\in\Sec{\Der L\otimes L^*}$; similar formulas are found for the $L$-bivector $\Lambda$ and direct computation shows that they satisfy the compatibility conditions of Propostion \ref{ExtBySymb}.
\end{proof}

Perhaps unsurprisingly after our discussion of section \ref{JacGeo} where Jacobi geometry was introduced as a unit-free version of Poisson geometry, we have found that the jet bundle assignment, as a candidate for dimensioned Hamiltonian functor, sends dimensioned configuration spaces into Jacobi manifolds.\newline

The next proposition establishes the interaction between taking jet bundles and taking line products of line bundles.
\begin{prop}[Jet Bundle of a Line Product]\label{JetLineProd}
Let $\lambda_1:L_1\to Q_1$, $\lambda_2:L_2\to Q_2$ two line bundles and denote their line product by $L_1\utimes L_2$, then there is a canonical factor $W$ giving the isomorphism of line bundles
\begin{equation*}
\normalfont
\begin{tikzcd}
L_{\Jet^1 (L_1\utimes L_2)} \arrow[r, "W"] \arrow[d] & L_{\Jet^1 L_1} \utimes L_{\Jet^1 L_2} \arrow[d] \\
\Jet^1 (L_1\utimes L_2) \arrow[r, "w"'] & \Jet^1 L_1\dtimes \Jet^1 L_2
\end{tikzcd}
\end{equation*}
Furthermore, the factor $W$ induces an isomorphism of Jacobi manifolds.
\end{prop}
\begin{proof}
Let us first construct the isomorphism factor $W$ explicitly. Recall from Proposition \ref{DerLineProd} that $\Der(L_1\utimes L_2)\cong p_1^*\Der L_1\oplus \Der L_2$, this allows us to write
\begin{equation*}
    \Jet^1 (L_1\utimes L_2)\cong(p_1^*\Der L_1\oplus \Der L_2)^*\otimes (L_1\utimes L_2),
\end{equation*}
then, using the swapping isomorphism factor for the second term, we find the following isomorphism of vector bundles
\begin{equation*}
    t:\Jet^1 (L_1\utimes L_2)\to p_1^*\Jet^1 L_1\oplus p_2^*\Jet^1 L_2.
\end{equation*}
covering the identity map on $Q_1\dtimes Q_2$. Now, by the definition of line product, it is clear that we can define the following map
\begin{align*}
z: p_1^*\Jet^1 L_1\oplus p_2^*\Jet^1 L_2 & \to \Jet^1 L_1\dtimes \Jet^1L_2\\
(\alpha_{q_1}\oplus \beta_{q_2})_{B_{q_1q_2}} & \mapsto C_{\alpha_{q_1}\beta_{q_2}}
\end{align*}
where $C_{\alpha_{q_1}\beta_{q_2}}=B_{q_1q_2}$, which is well-defined since the line bundle over a jet bundle is defined as the pull-back line bundle from the base. Conversely, given a fibre-wise invertible map between pull-back line bundles we can project to a fibre-wise map between the base line bundles, denote this projection by $\overline{\pi}:\Jet^1 L_1 \dtimes \Jet^1 L_2 \to Q_1 \dtimes Q_2$, so there is an obvious inverse for the map above
\begin{equation*}
    z^{-1}(C_{\alpha_{q_1}\beta_{q_2}})=(\alpha_{q_1}\oplus \beta_{q_2})_{\overline{\pi}(C_{\alpha_{q_1}\beta_{q_2}})}.
\end{equation*}
We thus find the desired diffeomorphism
\begin{equation*}
    w:=z\circ t:\Jet^1 (L_1\utimes L_2) \to \Jet^1 L_1\dtimes \Jet^1L_2.
\end{equation*}
Let us write the line product commutative diagram for the line bundles over the jets as
\begin{equation*}
\begin{tikzcd}
L_{\Jet^1L_1} \arrow[d, "\mu_1"'] & L_{\Jet^1L_1}\utimes L_{\Jet^1L_2} \arrow[l,"R_1"']\arrow[d, "\mu_{12}"]\arrow[r,"R_2"] & L_{\Jet^1L_2} \arrow[d,"\mu_2"] \\
\Jet^1L_1 & \Jet^1L_1 \dtimes \Jet^1L_2 \arrow[l,"r_1"]\arrow[r,"r_2"'] & \Jet^1L_2
\end{tikzcd}
\end{equation*}
and denote the compositions $o_i:=\pi_i\circ r_i:\Jet^1L_1\dtimes \Jet^1L_2\to Q_i$. It then follows by construction that
\begin{equation*}
    L_{\Jet^1L_1}\utimes L_{\Jet^1L_2}=o_1^*L_1
\end{equation*}
and
\begin{equation*}
    L_{\Jet^1(L_1\utimes L_2)}\cong w^*o_1^*L_1
\end{equation*}
thus showing
\begin{equation*}
    L_{\Jet^1(L_1\utimes L_2)}\cong w^*(L_{\Jet^1L_1}\utimes L_{\Jet^1L_2}).
\end{equation*}
This is, of course, tantamount to there being a factor $W:L_{\Jet^1(L_1\utimes L_2)}\to L_{\Jet^1L_1}\utimes L_{\Jet^1L_2}$ covering the diffeomorphism $w$ that essentially acts as the fibre-wise identity on $L_1$. In order to show that this factor is, in fact, a Jacobi map we can write the pull-backs of (non-zero) brackets of spanning sections directly
\begin{align*}
    W^*\{R_i^*l_{a_i},R_i^*l_{b_i}\}_{12}&=W^*R_i^*\{l_{a_i},l_{b_i}\}_i=(R_i\circ W)^*l_{[a_i,b_i]}\\
    W^*\{R_i^*l_{a_i},R_i^*\pi_i^*u_i\}_{12}&=W^*R_i^*\{l_{a_i},\pi_i^*u_i\}_i=(R_i\circ W)^*\pi_i^*a_i[u_i]
\end{align*}
for $u_i\in\Sec{L_i}$, $a_i,b_i\in\Sec{\Der L_i}$, $i=1,2$, and where the definition of product Jacobi structure has been used to write the RHS expressions. Note that the construction of $W$ above is such that
\begin{equation*}
    (R_i\circ W)^*l_{a_i}=\overline{l}_{k_i(a_i)}\qquad (R_i\circ W)^*\pi_i^*u_i=\overline{\pi}^*P_i^*u_i
\end{equation*}
where $\overline{\pi}:\Jet^1(L_1\utimes L_2)\to Q_1\dtimes Q_2$ is the jet bundle projection for the line product, $k_i:\Sec{\Der L_i}\to \Sec{\Der (L_1\utimes L_2)}$ are the natural Lie algebra injections of derivations from Proposition \ref{DerivLineProd} and $\overline{l}:\Sec{\Der (L_1\utimes L_2)}\to L_{\Jet^1(L_1\utimes L_2)}$ is the inclusion of fibre-wise linear sections. A direct computation then gives the brackets of pull-backs of spanning sections
\begin{align*}
    \{W^*R_i^*l_{a_i},W^*R_i^*l_{b_i}\}&=\overline{l}_{[k_i(a_i),k_i(b_i)]}=(R_i\circ W)^*l_{[a_i,b_i]}\\
    \{W^*R_i^*l_{a_i},W^*R_i^*\pi_i^*u_i\}&=\overline{\pi}^*k_i(a_i)[P_i^*u_i]=\overline{\pi}^*P_i^*a_i[u_i]=(R_i\circ W)^*\pi_i^*a_i[u_i]
\end{align*}
which agree with the pull-backs of brackets above, thus completing the proof.
\end{proof}

Now that we see that taking jet bundle respects the categorical product of line bundles, we are in the position to prove that $\Jet^1$ is indeed a functor sending line bundles to canonical contact manifolds and sending factors to coisotropic relations
\begin{equation*}
    \Jet^1: \Line_\Conf \to \textsf{Cont}.
\end{equation*}
We warn the reader that the category of contact manifolds $\textsf{Cont}$ should be considered with the same precaution as the category $\textsf{Symp}$ of section \ref{CanHamT}. Consider a factor between some pair of line bundles $B:L_1\to L_2$ covering a smooth map $b:Q_1\to Q_2$. Recall that the der map always gives a well-defined morphism of L-vector bundles $\Der B:\Der L_1\to \Der L_2$, however only factors covering diffeomorphisms can be dualized to give a well-defined map between the jet bundles. We can, nevertheless, define the \textbf{jet lift} of the factor $B$ as the following subamanifold of the base product of jet bundles
\begin{equation*}
    \Jet^1 B:=\{C_{\alpha_{q_1}\beta_{q_2}}\in \Jet^1 L_1\dtimes \Jet^1 L_2|\quad q_2=b(q_1),\quad (\Der_{q_1}B)^{*L}\beta_{q_2}=\alpha_{q_1}, \quad C_{\alpha_{q_1}\beta_{q_2}}=B_{q_1}\}
\end{equation*}
where the last condition is understood using the fact that the line bundles over the jet bundles are simply the pull-back bundles by the canonical projection to the base. The following proposition shows that the jet lift of a factor is indeed a coisotropic submanifold, thus proving functoriality of the jet assignment $\Jet^1$.

\begin{prop}[Jet Lift of a Factor]\label{JetLiftFact}
Let $B:L_1\to L_2$ be a factor between two line bundles, then its cotangent lift
\begin{equation*}
\normalfont
    \Jet^1 B\subset \Jet^1 L_1 \dtimes \Jet^1 L_2
\end{equation*}
is a coisotropic submanifold of the base of the Jacobi product $\normalfont L_{\Jet^1 L_1} \utimes \overline{L}_{\Jet^1 L_2}$.
\end{prop}
\begin{proof}
To show that $\Jet^1B$ is a coisotropic submanifold we will identify a set of generating sections of its vanishing submodule $\Gamma_{\Jet^1B}\subseteq \Sec{L_{\Jet^1(L_1\utimes L_2)}}$ and show that they indeed form a Lie subalgebra. Firstly, given $u_i\in\Sec{L_i}$, $a_i\in\Sec{\Der L_i}$, $i=1,2$, let us define the following sections of the line product of the canonical contact line bundles on the jet bundles $\Jet^1L_i$ following the notation introduced in Proposition \ref{JetLineProd}:
\begin{align*}
    l_{a_1a_2}&:=R_1^*l_{a_1}-R_2^*l_{a_2}\\
    \pi^*_{u_1u_2}&:= R_1^*\pi_1^*u_1-R_2^*\pi_2^*u_2.
\end{align*}
Consider a point on the jet lift $C_{\alpha_{q_1}\beta_{q_2}}\in \Jet^1 B$ so that $q_2=b(q_1)$, $(\Der_{q_1}B)^{*L}\beta_{q_2}=\alpha_{q_1}$ and $C_{\alpha_{q_1}\beta_{q_2}}=B_{q_1}$. Let us evaluate the defined sections on it:
\begin{align*}
l_{a_1a_2}(C_{\alpha_{q_1}\beta_{q_2}}) & = R_1|^{-1}_{C_{\alpha_{q_1}\beta_{q_2}}}l_{a_1}(\alpha_{q_1})-R_2|^{-1}_{C_{\alpha_{q_1}\beta_{q_2}}}l_{a_2}(\beta_{q_2})\\
& = l_{a_1}(\alpha_{q_1})-B^{-1}_{q_1}(l_{a_2}(\beta_{q_2}))\\
& = B^{-1}_{q_1}\beta_{q_2}(\Der_{q_1}B(a_1|_{q_1})-a_2|_{q_2})\\
& = B^{-1}_{q_1}\beta_{q_2}(\Der B \circ a_1 - a_2\circ b)(q_1)
\end{align*}
and
\begin{align*}
\pi^*_{u_1u_2} & = R_1|^{-1}_{C_{\alpha_{q_1}\beta_{q_2}}}u_1(q_1)-R_2|^{-1}_{C_{\alpha_{q_1}\beta_{q_2}}}u_2(q_2)\\
& = u_1(q_1)-B_{q_1}^{-1}u_2(q_2)\\
& = u_1(q_1)-B_{q_1}^{-1}u_2(b(q_1))\\
& = (u_1-B^*u_2)(q_2).
\end{align*}
It is then obvious that these spanning sections vanish on the jet lift $\Jet^1 B$ iff the derivations are $B$-related and the sections are mapped by the pull-back $B^*$, i.e.
\begin{align*}
    \Gamma_{\Jet^1 B}\ni l_{a_1a_2} &\Leftrightarrow   a_1\sim_B a_2\\
    \Gamma_{\Jet^1 B}\ni \pi^*_{u_1u_2}  &\Leftrightarrow  u_1=B^*u_2.
\end{align*}
These are the generating vanishing sections so it will suffice to check that evaluations of brackets among these on an arbitrary point of the jet lift $C_{\alpha_{q_1}\beta_{q_2}}\in \Jet^1 B$ vanish. From the defining conditions of a the linear Jacobi we see that $\{\pi^*_{u_1u_2},\pi^*_{u_1'u_2'}\}=0$, so we are left with the two other possible brackets. For the bracket of fibre-wise linear sections we compute explicitly using the defining properties of the product Jacobi bracket:
\begin{equation*}
    \{l_{a_1a_2} ,l_{a_1'a_2'} \}= R^*_1\{l_{a_1},l_{a_1'}\}_1-R^*_2\{l_{a_2},l_{a_2'}\}_2=R^*_1l_{[a_1,a_1']}-R^*_2l_{[a_2,a_2']}
\end{equation*}
for $a_i,a_i'\in\Sec{\Der L_i}$, $i=1,2$. Then if $l_{a_1a_2}$ and $l_{a_1'a_2'}$ are vanishing sections, the derivations are $B$-related and, by virtue of Proposition \ref{DerFunct}, where $\Der B$ is shown to be a Lie algebroid morphism, the brackets are also $B$-related $[a_1,a_1']\sim_B [a_2,a_2']$ making the expression above into a vanishing section. Only the cross bracket $\{l_{a_1a_2} ,\pi^*_{u_1u_2} \}= R^*_1\pi_1^*a_1[u_1]_1-R^*_2\pi_2^*a_2[u_2]$ remains, for which we evaluate on a point of the jet lift and show it vanishes by directly computing using $u_1=B^*u_2$ and $a_1\sim_B a_2$:
\begin{align*}
\{l_{a_1a_2} ,\pi^*_{u_1u_2} \}(C_{\alpha_{q_1}\beta_{q_2}}) & = a_1|_{q_1}(u_1)-B_{q_1}^{-1}a_2|_{q_2}(u_2)\\
& = a_1|_{q_1}(u_1)-B_{q_1}^{-1}\Der_{q_1}B(a_1|_{q_1})(u_2)\\
& = a_1|_{q_1}(B^*u_2)-B_{q_1}^{-1}B_{q_1}a_1|_{q_1}(B^*u_2)\\
& = a_1|_{q_1}(B^*u_2)-a_1|_{q_1}(B^*u_2)\\
& = 0.
\end{align*}
\end{proof}

The next proposition gives a consistency condition for the two alternative ways to regard the jet bundle of the line bundle induced on a submanifold of the base. On the one had, we could take the line bundle as a an embedded subbundle and construct its jet bundle from the ambient jet bundle. On the other, we could simply regard the restricted line bundle as an intrinsic line bundle a canonically construct its jet bundle. We prove below that these two constructions are equivalent.
\begin{prop}[Canonical Coisotropic Reduction]\label{CanCoisoRedJet}
Let $i:S\hookrightarrow Q$ be a submanifold of a line bundle $L$, then the restriction of the ambient jet bundle to the submanifold $\normalfont (\Jet^1 L)|_S$ is a coisotropic submanifold with respect to the canonical contact structure on $\normalfont \Jet^1 L$. Furthermore, there is a submersion factor covering the surjective submersion $\normalfont z:(\Jet^1 L)|_S\twoheadrightarrow \Jet^1 L_S$ given by the fibre-wise quotient:
\begin{equation*}
\normalfont
    \Jet^1_qL/(\Der L_S)^{0L}\cong \Jet^1_q L_S\qquad q\in S
\end{equation*}
so that the canonical contact structure on $\normalfont \Jet^1 L$ Jacobi reduces to the canonical contact structure on $\normalfont \Jet^1 L_S$.
\end{prop}
\begin{proof}
Let us first prove that the above map fits in a Jacobi reduction scheme. It is a direct implication of Proposition \ref{DerSubMan} and the basic properties of L-vector spaces presented at the end of section \ref{CatLine} applied fibre-wise, that we have the following diagram of line bundle morphisms
\begin{equation*}
\begin{tikzcd}[sep=small]
L_{\Jet^1L|_{S}} \arrow[rr,"\iota"] \arrow[dd,"\zeta"'] \arrow[dr]& & L_{\Jet^1L} \arrow[dr] & \\
& \Jet^1L|_S \arrow[rr,"i"', hook] \arrow[dd, "z",twoheadrightarrow] & & \Jet^1L \\
L_{\Jet^1L_S} \arrow[dr] & & & \\
 & \Jet^1L_S &  & 
\end{tikzcd}
\end{equation*}
where $\iota$ denotes, abusing notation, the embedding factor induced by the submanifold $i:S\hookrightarrow Q$ and $\zeta$ is defined as the fibre-wise identity of the pull-back line bundles covering the point-wise linear submersion
\begin{equation*}
    z_q:\Jet^1_q L\twoheadrightarrow \Jet^1_qL/(\Der L_S)^{0L}\cong \Jet^1_q L_S, \qquad q\in S.
\end{equation*}
In light of the discussion of section \ref{CatLineB} on the derivations of the embedded line bundle over a submanifold, and Proposition \ref{DerivSubMan} in particular, we find natural isomorphisms $\Sec{L_S}\cong\Sec{L}/\Gamma_S$ and $\Dr{L_S}\cong\text{Der}_S(L)/\text{Der}_{0S}(L)$, and denoting the natural inclusions of spanning sections on $L_{\Jet^1 L_S}$ by $\overline{l}$ and $\overline{\pi}$, we can write the fibre-wise linear Jacobi structure on $\Jet^1L_S$ equivalently as
\begin{align*}
    \{\overline{l}_{\overline{a}},\overline{l}_{\overline{b}}\}_S & =  \overline{l}_{\overline{[a,b]}}\\
    \{\overline{l}_{\overline{a}},\overline{\pi}^*\overline{u}\}_S & =  \overline{\pi}^*\overline{a[u]}\\
    \{\overline{\pi}^*\overline{u},\overline{\pi}^*\overline{v}\}_S & = 0
\end{align*}
with $\overline{u},\overline{v}\in\Sec{L}/\Gamma_S$ and $\overline{a},\overline{b}\in\text{Der}_S(L)/\text{Der}_{0S}(L)$, and which is well defined-precisely from the description of derivations as a subquotient of Lie algebras. The submersion factor $\zeta:L_{\Jet^1L|_{S}}\to L_{\Jet^1L_S}$ covering the quotient map $z:\Jet^1L|_{S}\to \Jet^1L_S$ has been defined such that it is the point-wise counterpart to the isomorphisms used above to rewrite the linear Jacobi bracket. Pull-backs via these factors satisfys the following identities by construction
\begin{equation*}
    \zeta^*\overline{\pi}^*\overline{u} =\iota^*\pi^*u \qquad \zeta^*\overline{l}_{\overline{a}} = \iota^* l_{a}
\end{equation*}
for all $u\in\Sec{L}$ and $a\in\text{Der}_S(L)$. This now clearly implies the reduction condition for all spanning sections
\begin{align*}
    \zeta^*\{\overline{l}_{\overline{a}},\overline{l}_{\overline{b}}\}_S & = \iota^* l_{[a,b]}=\iota^*\{l_a,l_b\}_L\\
    \zeta^*\{\overline{l}_{\overline{a}},\overline{\pi}^*\overline{u}\}_S & = \iota^* \pi^*a[u]=\iota^* \{l_a,\pi^*u\}_L\\
    \zeta^*\{\overline{\pi}^*\overline{u},\overline{\pi}^*\overline{v}\}_S & = 0 = \iota^* \{\pi^*u,\pi^*v\}_L
\end{align*}
thus showing that the linear Jacobi $(L_{\Jet^1L},\{,\}_L)$ reduces to $(L_{\Jet^1L_S},\{,\}_S)$. Lastly, it is easy to see that the vanishing sections of $\Jet^1L|_S$ seen as a submanifold of the jet bundle are precisely those of the form $l_{\text{Der}_{0S}(L)}$ and $\pi^*\Gamma_S$. It follows again from Proposition \ref{DerivSubMan} that these form a Lie subalgebra of the linear Jacobi structure $(L_{\Jet^1L},\{,\}_L)$, thus making $\Jet^1L|_S\subseteq \Jet^1L$ into a coisotropic submanifold.
\end{proof}

Consider now a line bundle action $G\Acts L$. Recall that, in the case of a free and proper action, the orbit space is canonically a line bundle, denoted by $L/G$, and that there is a natural submersion factor $\sigma:L\to L/G$. The following proposition relates the canonical contact structures associated to these two line bundles.
\begin{prop}[Canonical Hamiltonian Reduction]\label{CanHamRedJet}
Let $\Phi:G\times L\to L$ be a free and proper line bundle action of a connected Lie group $G$ on $L$, then the canonical contact structure on $\normalfont \Jet^1 L$ Jacobi reduces to the canonical contact structure on $\normalfont \Jet^1 (L/G)$. This reduction is, in fact, Hamiltonian: the jet lift of the line bundle action $\normalfont G\Acts \Jet^1 L$ preserves the canonical contact structure and has a natural comoment map given by
\begin{align*}
\overline{\mu}: \mathfrak{g} &\normalfont \to \Sec{L_{\Jet^1 L}}\\
\xi & \mapsto l_{\Psi(\xi)},
\end{align*}
where $\normalfont \Psi:\mathfrak{g}\to \Dr{L}$ is the infinitesimal line bundle action and $\normalfont l:\Dr{L}\to \Sec{L_{\Jet^1 L}}$ is the natural inclusion of derivations as fibre-wise linear sections on the jet bundle.
\end{prop}
\begin{proof}
Before we address the specific case of a group action, let us discuss the jet lift of a general diffeomorphic factor, i.e. a fibre-wise invertible line bundle morphism $B:L\to L$ covering a diffeomorphism $b:Q\to Q$. The \textbf{jet lift} of $B$ is a factor of the canonical contact line bundle on the jet bundle $L_{\Jet^1 L}$ defined by
\begin{equation*}
    \Jet^1 B:(\alpha_q,l_q)\mapsto (\Jet^1 b(\alpha_q),B^{-1}_{q}l_q)
\end{equation*}
where
\begin{equation*}
    \Jet^1 b(\alpha_q)=(\Der_qB)^{L_q*L_{b^{-1}(q)}}(\alpha_q)\in\Jet_{b^{-1}(q)}^1L.
\end{equation*}
The cotangent lift then induces the following commutative diagram of vector bundle morphisms
\begin{equation*}
    \begin{tikzcd}[sep=small]
     L_{\Jet^1L}\arrow[r,"\Jet^1B"]\arrow[d] & L_{\Jet^1L}\arrow[d] \\
     \Jet^1L\arrow[r,"\Jet^1 b"] \arrow[d] & \Jet^1L \arrow[d] \\
     Q\arrow[r,"b^{-1}"] & Q
    \end{tikzcd}
\end{equation*}
For any two diffeomorphic factors $B,F:L\to L$ and the identity factor $\Id_L:L\to L$ it is a simple check to show that
\begin{equation*}
    \Jet^1(B\circ F)=\Jet^1F\circ \Jet^1B, \qquad \Jet^1(\Id_L)=\Id_{L_{\Jet^1L}}.
\end{equation*}
The spanning sections $l_a,\pi^*u\in\Sec{L_{\Jet^1L}}$ transform under pull-back by a jet lift of a diffeomorphic factor according to the following expressions
\begin{equation*}
    \Jet^1B^*l_a=l_{B_*a}, \qquad \Jet^1B^*\pi^*u=\pi^*(B^{-1})^*u,
\end{equation*}
then it follows that the jet lift of a diffeomorphic factor $\Jet^1B$ is indeed a Jacobi map of the canonical contact structure on the jet bundle:
\begin{align*}
    \Jet^1B^*\{l_a,l_b\}_L & = \Jet^1B^*l_{[a,b]}=l_{[B_*a,B_*b]}=\{l_{B_*a},l_{B_*b}\}_L=\{\Jet^1B^*l_a,\Jet^1B^*l_b\}_L \\
    \Jet^1B^*\{l_a,\pi^*u\}_L & = \Jet^1B^*\pi^*a[u]=\pi^*(B_*a)[(B^{-1})^*u]=\{l_{B_*a},\pi^*(B^{-1})^*u\}_L=\{ \Jet^1B^*l_a, \Jet^1B^* \pi^*u\}_L \\
    \Jet^1B^*\{\pi^*u,\pi^*v\}_L & = 0 = \{\pi^*(B^{-1})^*u,\pi^*(B^{-1})^*v\}_L=\{\Jet^1B^*\pi^*u,\Jet^1B^*\pi^*v\}_L.
\end{align*}
The jet lift of the group action $G\Acts \Jet^1 L$ is defined by the jet lifts of the diffeomorphic factors corresponding to each group element
\begin{equation*}
    (\Jet^1 \Phi)_g:=\Jet^1\Phi_g,
\end{equation*}
which, in light of the above results for general jet lifts of diffeomorphic factors, is readily checked to be a group action that acts via Jacobi maps. This is a Hamiltonian action with comoment map simply given by $\overline{\mu}:=l\circ \Psi :\mathfrak{g}\to \Sec{L_{\Jet^1L}}$. Observe that the zero locus of the moment map is naturally identified with the L-annihilator of the subspace of derivations spanned by the infinitesimal generators regarded as a subbundle of the jet bundle
\begin{equation*}
    \mu^{-1}(0)=\Psi(\mathfrak{g})^{0L}\subseteq \Jet^1L.
\end{equation*}
Note that $G$-equivariance of the infinitesimal action $\Psi$ implies that the jet lifted action of $G$ restricts to a $G$-action on $\Psi(\mathfrak{g})^{0L}$, indeed we check for any $j_q^1u\in\Psi(\mathfrak{g})^{0L}$ and $\xi\in\mathfrak{g}$
\begin{equation*}
    \Jet^1\phi_g(j_q^1 u)(\Psi(\xi))=\Psi(\xi)[\Phi^*_gu]=(\Phi_g)_{b^{-1}(q)}\Der_q\Phi_g\Psi(\xi)=(\Phi_g)_{b^{-1}(q)}j_q^1u(\Psi(\text{Ad}_g(\xi))=0.
\end{equation*}
Using Proposition \ref{DerGAct} and simple linear algebra of L-vector bundles we find the following point-wise isomorphism
\begin{equation*}
    \Jet^1_{[q]}(L/G):=(\Der_{[q]}(L/G))^{*L/G}\cong(\Der_q/\Psi(\mathfrak{g})_q)^{*L_q}\cong \Psi(\mathfrak{g})_q^{0L_q}.
\end{equation*}
This allows us to write the following factor reduction diagram
\begin{equation*}
\begin{tikzcd}[sep=small]
L_{\Psi(\mathfrak{g})^{0L}} \arrow[rr,"\iota"] \arrow[dd,"\zeta"'] \arrow[dr]& & L_{\Jet^1L} \arrow[dr] & \\
& \Psi(\mathfrak{g})^{0L} \arrow[rr,"i"', hook] \arrow[dd, "z",twoheadrightarrow] & & \Jet^1L \\
L_{\Jet^1(L/G)} \arrow[dr] & & & \\
 & \Jet^1(L/G) &  & 
\end{tikzcd}
\end{equation*}
where $\iota$ is the embedding factor for the L-annihilator $\Psi(\mathfrak{g})^{0L}$ seen as a subbundle (submanifold) of the jet bundle and $\zeta$ is the submersion factor induced by the jet lifted action restricted $\Psi(\mathfrak{g})^{0L}$, which is clearly free and proper. Since $G$ is connected, recall that the sections and derivations of the quotient line bundle can be equivalently regarded as
\begin{equation*}
    \Sec{L/G}\cong \Sec{L}^\mathfrak{g}, \qquad \Dr{L/G}\cong \Dr{L}^{\mathfrak{g}}/\Psi(\mathfrak{g}),
\end{equation*}
thus the linear Jacobi structure determined by the spanning sections $\overline{l}_{\overline{a}},\overline{\pi}^*\overline{u}\in\Sec{L_{\Jet^1(L/G)}}$ can be fully characterised under these isomorphisms. The factors constructed above are such that we have the following explicit treatment of extensions of spanning sections
\begin{align*}
    \zeta^*\overline{\pi}^*\overline{u}=\iota^*\pi^*u \quad &\Leftrightarrow \quad u\in\Sec{L}^\mathfrak{g}\\
    \zeta^*\overline{l}_{\overline{a}}=\iota^*l_a \quad &\Leftrightarrow \quad \overline{a}=a+\Psi(\mathfrak{g}), a\in\Dr{L}^\mathfrak{g}
\end{align*}
The reduction condition is now easily checked
\begin{align*}
    \zeta^*\{\overline{l}_{\overline{a}},\overline{l}_{\overline{b}}\}_{L/G} & = \iota^* l_{[a,b]}=\iota^*\{l_a,l_b\}_L\\
    \zeta^*\{\overline{l}_{\overline{a}},\overline{\pi}^*\overline{u}\}_{L/G} & = \iota^* \pi^*a[u]=\iota^* \{l_a,\pi^*u\}_L\\
    \zeta^*\{\overline{\pi}^*\overline{u},\overline{\pi}^*\overline{v}\}_{L/G} & = 0 = \iota^* \{\pi^*u,\pi^*v\}_L
\end{align*}
for all spanning sections $\overline{l}_{\overline{a}},\overline{l}_{\overline{b}},\overline{\pi}^*\overline{u},\overline{\pi}^*\overline{v}\in\Sec{L_{\Jet^1(L/G)}}$ and extensions $l_a,l_b,\pi^*u,\pi^*v\in\Sec{L_{\Jet^1L}}$, thus concluding the proof.
\end{proof}

The propositions presented thus far in this section then allow for the definition of the \textbf{dimensioned Hamiltonian functor} $\mathbb{H}_{\Line_\Conf}$ as the following list of associations
\begin{center}
\begin{tabular}{ c c c }
Dimensioned Configuration Spaces & Hamiltonian Functor  & Dimensioned Phase Spaces \\
\hline
 $\Line_\Conf$ & $\begin{tikzcd}\phantom{A} \arrow[r,"\mathbb{H}_{\Line_\Conf}"] & \phantom{B} \end{tikzcd}$ & $\textsf{Cont}$ \\ 
 $L_Q$ & $\begin{tikzcd} \phantom{Q} \arrow[r, "\Jet^1 ",mapsto] & \phantom{Q} \end{tikzcd}$ & $(\Jet^1 L_Q,\text{H}_{L_Q})$ \\ 
 $\Obs{L_Q}$ & $\begin{tikzcd} \phantom{Q} \arrow[r,hookrightarrow, "\pi^*"] & \phantom{Q}  \end{tikzcd}$ & $\Obs{L_{\Jet^1 L_Q}}$ \\
 $\Dyn{L_Q}$ & $\begin{tikzcd} \phantom{Q} \arrow[r,hookrightarrow, "l"] & \phantom{Q}  \end{tikzcd}$ & $\Obs{L_{\Jet^1 L_Q}}$ \\
 $L_1\utimes L_2$ & $\begin{tikzcd} \phantom{Q} \arrow[r,mapsto, "\Jet^1 "] & \phantom{Q}  \end{tikzcd}$ & $\Jet^1 L_1\dtimes \Jet^1 L_2$ \\
 $(L_1,s_1), (L_2,s_2)$ & $\begin{tikzcd} \phantom{Q} \arrow[r,mapsto, "\Jet^1 "] & \phantom{Q}  \end{tikzcd}$ & $(L_{\Jet^1 L_1}\utimes L_{\Jet^1 L_2},R_1^*\pi_1^*s_1 + R_2^*\pi^*_1s_2)$ \\
 $B:L_1\to L_2$ & $\begin{tikzcd} \phantom{Q} \arrow[r,mapsto, "\Jet^1 "] & \phantom{Q}  \end{tikzcd}$ & $\Jet^1 B\subset \Jet^1 L_1\dtimes \Jet^1 L_2$ coisotropic\\
 $\iota:L_S\hookrightarrow L_Q$ & $\begin{tikzcd} \phantom{Q} \arrow[r,mapsto, "\Jet^1"] & \phantom{Q}  \end{tikzcd}$ & $(\Jet^1 L_Q)|_S\subset \Jet^1 L_Q$ coisotropic \\
 $G\Acts L$ & $\begin{tikzcd} \phantom{Q} \arrow[r,mapsto, "\Jet^1"] & \phantom{Q}  \end{tikzcd}$ & $G\Acts \Jet^1 L$ Hamiltonian action \\
\end{tabular}
\end{center}
Comparing the list above with the Hamiltonian functor of section \ref{CanHamT}, we readily conclude that the proposed theory of dimensioned Hamiltonian mechanics shares its fundamental categorical structure with the conventional unit-less formulation. We are then compelled to coin the term \textbf{canonical dimensioned Hamiltonian mechanics} for this framework of geometric foundations of physical theories naturally incorporating the notion of physical dimension and units of measurement.

\section{Conclusion} \label{Conc}

Despite the relative success of finding a Hamiltonian functor for the category of dimensioned configuration spaces, there are several pressing questions that arise in the context of this generalization, both physically and mathematically motivated, that need to be addressed in order to fully justify the adequacy of our formalism as practical foundations of physical theories. Although the detailed analysis of such questions will not be carried out here, partly due to many of them still being under investigation by the author at the time of writing, we list below the most prominent ones and comment on the implications and future lines of research that follow from the results presented in this paper.
\begin{itemize}
    \item \textbf{The General Philosophy of Dimensioned Geometric Mechanics.} A theory of dimensioned phase spaces, mathematically articulated with ``unit-free'' manifolds/line bundles (section \ref{CatLineB}) and physically motivated by the structure of measurand spaces (section \ref{ModMes}), seems to spring from the findings presented in this paper. Note how most of our technical efforts, e.g. identifying measurand spaces within the $\Line$ category or implementing Hamiltonian dynamics in canonical contact manifolds, are focused on integrating the extra information of physical quantities carrying units within a mathematical framework that should generalize its ``unit-less'' counterpart. Note, as well, that once the generalization is made, so that physical dimensionality and units are systematically accounted for, the formal treatment of physical quantities is manifestly unit-free. In this manner our approach could be seen as the transition from the ``unit-less'' to the ``unit-free''. Indeed, we put all this mathematical technology in place to account for units to then give a description that is explicitly independent of units. This reminds the author of the quote \emph{``Tout pour le peuple, rien par le peuple''} proclaimed by several enlightened rulers of 18\textsuperscript{th} century Europe. Paraphrasing this motto, a slogan summarizing our philosophy could be:
        \begin{center}
            Dimensioned Mechanics\\
            -\\
            All for the units, but without the units.
        \end{center}
    \item \textbf{Dimensioned Dynamics.} Recall that the natural notion of dynamics in the context of dimensioned configuration spaces (section \ref{CanHamJ}) was given by the derivations of a line bundle $\Dyn{L_Q}=\Dr{L_Q}$. Integrating equations of motion will thus give a $1$-parameter family of line bundle automorphisms or, equivalently, a smooth curve on the general linear groupoid of the line bundle $\text{GL}(L_Q)$. Hence, it seems that there may be a correspondence of the form $\Dyn{L_Q}\cong \Sec{\Tan (\text{GL}(L_Q))}$. When the line bundle is the canonical contact structure present in a jet bundle $L_{\Jet^1L}$, what sort of dynamical equations on $\text{GL}(L_{\Jet^1L})$ are induced by a Hamiltonian flow of derivations on $L_{\Jet^1L}$? Answering these questions successfully will provide the correct notion of ``unit-free'' dynamics and will make sense of the idea of the time-evolution of a measurand.
    \item \textbf{Recovering Conventional Hamiltonian Dynamics from the Dimensioned Formalism.} Introducing an adapted local coordinate chart on a jet bundle $\Jet^1L_Q$ renders the canonical Jacobi structure trivial and determined by a non-zero Jacobi pair $(\pi_Q,R_Q)$. It follows by construction that $(\Cot Q,\pi_Q)$ is, in fact, the canonical symplectic structure on the cotangent bundle. It is then clear that the subspace of local observables $f\in\Cin{\Jet^1L_Q}$ that are $R_Q$-invariant form a Poisson algebra that is isomorphic to the canonical Poisson algebra of the cotangent bundle. The author suspects that this is closely related to our notion of local choice of unit and further investigation may reveal that they are, indeed, equivalent ways to recover conventional Hamiltonian dynamics in the dimensioned framework. Once such a connection is established, what is the relation to the dynamics on the general linear groupoid $\text{GL}(L_{\Jet^1L_Q})$ described in the previous point?
    \item \textbf{Where is the commutative product of classical observables in dimensioned Hamiltonian mechanics?} One of the trade-offs of the notion of unit-free manifold as a line bundle $L_Q$ is that the replacement of the ring of functions, the module of sections, no longer has a commutative product. However, the tensor product in the category of line bundles $\Line_\Man$ defined in section \ref{CatLineB} allows us to recover a notion of commutative product by taking the tensor product of sections which, under trivialization, becomes the ordinary product of functions. In this manner, the datum of a single dimensioned configuration space naturally carries an entire infinite family of line bundles given by its potentil of tensor powers $\{\otimes^n L_Q^*\otimes^m L_Q\}_{n,m=0}^\infty$. It is within this potentil that a consistent definition of a ``unit-free'' commutative product of ``unit-free'' functions seems possible. The author suspects that the infinite family of associated canonical Jacobi structures can be made into some sort of graded Lie algebra where the ``unit-free'' commutative product can be regarded as the commutative product of a graded Poisson algebra. The details of this construction are currently being worked on by the author.
    \item \textbf{Dimensioned Newtonian Mechanics.} At the end of section \ref{CanHamT} the category of Newtonian configuration spaces was introduced as the historical motivation for the notion of energy in phase space. The key ingredient of Newtonian mechanics is the additional structure of a Riemannian metric on configuration spaces. In the context of our generalization to dimensioned configuration spaces, an obvious way to generalize the notion of Newtonian configuration spaces would be to consider L-valued bilinear forms on der bundles, indeed the ``unit-free'' analogue of conventional metrics. This follows from the fact that a L-metric is easily shown to induce analogous musical isomorphisms between jet and der bundles and a canonical choice of dimensioned observable given by the L-valued kinetic energy. It seems plausible that ``unit-free'' Riemannian geometry will resemble conventional Riemannian geometry but the details of this remain unclear to the author. Also, with physical applications in mind, note that the metric of a Newtonian system seems to be entirely independent of all the possible choices of observables that may be considered for the dynamical theory, whereas a choice of L-metric is tied to a particular dimensionality of observables. Details on how to formally implement the equivalent of Riemannian geometry in ``unit-free'' manifolds, and thus how to consistently define dimensioned Newtonian mechanics, will follow in future works by the author.
    \item \textbf{What does it mean for a dimensioned observable to be a section of a non-trivial line bundle?} Having chosen generically non-trivialized line bundles as our mathematical model for dimensioned configuration spaces, something that is necessary in order to explicitly implement the structure of measurand space in classical mechanics, raises the question of whether it is physically meaningful to consider strictly non-trivial line bundles. Intuitively, this would correspond to a physical system with a property for which a consistent choice of unit cannot be simultaneously made for the entirety of the system. As an illustrative example consider the simplest topologically non-trivial configuration space: the circle $S^1$. We may interpret this as the positions of the classical pendulum. Typical observables of this system, such as a gravitational potential, are sections of the trivial bundle $S^1\times\Real$. These can always be shifted by adding some global constant and be seen as non-zero global functions on $S^1$, this is ensured since the bundle is trivial. However, if we now considered the Moebius band as the dimensioned configuration space, all global sections must vanish at least at one point of $S^1$. It is unclear what physical situation this is describing. A first approach to shed some light on this question would be to find the closest analog to the classical pendulum as dimensioned Newtonian dynamics on the Moebius band and then, using some of the necessary developments mentioned above for integrating dimensioned dynamics, find what motions they project on $S^1$ and how they compare to the usual motions of a classical pendulum.
    \item \textbf{Homological Reduction in Dimensioned Phase Spaces.} Coisotropic and Hamiltonian reduction in Poisson manifolds has been successfully described in homological terms by formalisms such as BRST cohomology. In section \ref{CanHamJ} we found explicit examples of coisotropic and Hamiltonian reduction of Jacobi manifolds and so it is natural to ask whether there is an analogue of BRST cohomology for the case of Jacobi manifolds.
    \item \textbf{Quantization of Dimensioned Phase Spaces.} Assuming the open questions about the physical validity of our proposed generalization are resolved positively, so that dimensioned phase spaces are given clear physical interpretation, an immediate parallel from the history of physics prompts us to ask: what is the quantization of a dimensioned phase space? This question poses, at least, as many difficulties as the conventional problem of quantization but it may be that the objects quantizing dimensioned phase spaces are of mathematical interest in their own right since they may not have been identified and systematically studied before.
    \item \textbf{Dimensioned Quantum Mechanics.} Moving away from the realm of configuration spaces and classical mechanics as a whole, we note that our philosophy to articulate known geometric physical theories that are conventionally ``unit-less'' in terms of ``unit-free'' manifolds could be extended easily to other areas. In particular, it may be interesting to explore how the usual formulation of quantum mechanics in Hilbert spaces is modified after enforcing the ``unit-free'' point of view and whether this new dimensioned quantum mechanics can be regarded, at least in some simple cases, as a quantization of classical dimensioned mechanics.
\end{itemize}

\section*{Acknowledgements}

Firstly, I would like to thank Jonas Schnitzer, who first introduced me to the topic of Jacobi geometry during a common doctoral stay at IMPA, Rio de Janeiro, in April 2018, and who guided me into the current state-of-the-art of the subject. Secondly, I thank Henrique Bursztyn for the opportunity to visit IMPA, for welcoming me in such an stimulating research environment and for the financial support that covered for my stay in Rio de Janeiro. I am also grateful to Jos{\'e} Figueroa-O'Farrill, Pau Enrique and Federico Zalamea for their useful comments and discussions.\newline

The first time I encountered the question of how to treat observables with physical dimension in mechanics was at the Deep Thinking room of Strohlehenalm, a mountain cabin situated in the snowy Austrian alps, where the Rethinking Foundations of Physics Workshop 2017 was taking place. I would like to thank the organizers of that particular event, and the members of the Basic Research Community for Physics as a whole, for encouraging me to think about these matters.

\printbibliography

\end{document}